\documentclass[11pt]{aart}

\usepackage[letterpaper, hmargin=1in, top=1in, bottom=1.2in, footskip=0.6in]{geometry}
\usepackage{titlesec}
\titleformat{\section}[block]{\filcenter\normalfont\bfseries\large}{\thesection.}{.5em}{}\titlespacing*{\section}{0pt}{2\baselineskip}{1\baselineskip}
\titleformat{\subsection}[runin]{\normalfont\bfseries}{\thesubsection.}{.4em}{}[.]\titlespacing{\subsection}{0pt}{2ex plus .1ex minus .2ex}{.8em}
\titleformat{\subsubsection}[runin]{\normalfont\itshape}{\thesubsubsection.}{.3em}{}[.]\titlespacing{\subsubsection}{0pt}{1ex plus .1ex minus .2ex}{.5em}
\titleformat{\paragraph}[runin]{\normalfont\itshape}{\theparagraph.}{.3em}{}[.]\titlespacing{\paragraph}{0pt}{1ex plus .1ex minus .2ex}{.5em}

\usepackage[labelfont=bf,font=small,labelsep=period]{caption}
\usepackage[T1]{fontenc}
\usepackage[utf8]{inputenc}
\usepackage{lmodern}

\let\originalleft\left
\let\originalright\right
\renewcommand{\left}{\mathopen{}\mathclose\bgroup\originalleft}
\renewcommand{\right}{\aftergroup\egroup\originalright}

\usepackage{amsmath}
\usepackage{amssymb}
\usepackage{amsfonts}
\usepackage{latexsym}
\usepackage{amsthm}
\usepackage{amsxtra}
\usepackage{amscd}
\usepackage{bbm}
\usepackage{mathrsfs}
\usepackage{bm}
\usepackage{mathtools}

\usepackage{tikz}
\usetikzlibrary{arrows,decorations.pathmorphing,backgrounds,positioning,fit,petri}

\usepackage{graphicx, color}

\definecolor{darkred}{rgb}{0.9,0,0.3}
\definecolor{darkblue}{rgb}{0,0.3,0.9}

\definecolor{vdarkred}{rgb}{0.6,0,0.2}
\definecolor{vdarkblue}{rgb}{0,0.2,0.6}
\usepackage[pdftex, colorlinks, linkcolor=vdarkblue,citecolor=vdarkred]{hyperref}

\usepackage{booktabs}
\usepackage[nottoc,notlof,notlot]{tocbibind}
\usepackage{cite} 

\numberwithin{equation}{section}
\numberwithin{figure}{section}

\usepackage{enumitem}

\theoremstyle{plain} 
\newtheorem{theorem}{Theorem}[section]
\newtheorem*{theorem*}{Theorem}
\newtheorem{lemma}[theorem]{Lemma}
\newtheorem*{lemma*}{Lemma}
\newtheorem{corollary}[theorem]{Corollary}
\newtheorem*{corollary*}{Corollary}
\newtheorem{proposition}[theorem]{Proposition}
\newtheorem*{proposition*}{Proposition}

\newtheorem*{conjecture*}{Conjecture}

\theoremstyle{definition} 
\newtheorem{definition}[theorem]{Definition}
\newtheorem*{definition*}{Definition}

\newtheorem*{example*}{Example}
\newtheorem{remark}[theorem]{Remark}
\newtheorem*{remark*}{Remark}
\newtheorem{assumption}[theorem]{Assumption}
\newtheorem*{assumption*}{Assumption}

\newcommand{\f}[1]{\boldsymbol{\mathrm{#1}}} 
\newcommand{\bb}{\mathbb} 
\renewcommand{\cal}{\mathcal} 
 
\newcommand{\fra}{\mathfrak} 
\newcommand{\ol}[1]{\overline{#1} \!\,} 
\newcommand{\wh}{\widehat}
\newcommand{\wt}{\widetilde}
\newcommand{\op}{\operatorname}

\renewcommand{\P}{\mathbb{P}}
\newcommand{\E}{\mathbb{E}}
\newcommand{\R}{\mathbb{R}}
\newcommand{\C}{\mathbb{C}}
\newcommand{\N}{\mathbb{N}}
\newcommand{\Z}{\mathbb{Z}}

\newcommand{\ee}{\mathrm e}
\newcommand{\ii}{\mathrm i}
\newcommand{\dd}{\mathrm d}
\newcommand{\col}{\vcentcolon}
\newcommand*{\deq}{\mathrel{\vcenter{\baselineskip0.65ex \lineskiplimit0pt \hbox{.}\hbox{.}}}=}
\newcommand*{\eqd}{=\mathrel{\vcenter{\baselineskip0.65ex \lineskiplimit0pt \hbox{.}\hbox{.}}}}

\renewcommand{\leq}{\leqslant}
\renewcommand{\geq}{\geqslant}
\renewcommand{\epsilon}{\varepsilon}

\newcommand{\oo}{\mathrm o}

\newcommand{\floor}[1] {\lfloor #1 \rfloor}
\newcommand{\ceil}[1]  {\lceil  #1 \rceil}

\newcommand{\ind}[1]{\f 1_{#1}}

\newcommand{\p}[1]{(#1)}
\newcommand{\pb}[1]{\bigl(#1\bigr)}
\newcommand{\pB}[1]{\Bigl(#1\Bigr)}
\newcommand{\pbb}[1]{\biggl(#1\biggr)}
\newcommand{\pBB}[1]{\Biggl(#1\Biggr)}

\newcommand{\q}[1]{[#1]}

\newcommand{\qbb}[1]{\biggl[#1\biggr]}
\newcommand{\qBB}[1]{\Biggl[#1\Biggr]}

\newcommand{\h}[1]{\{#1\}}
\newcommand{\hb}[1]{\bigl\{#1\bigr\}}

\newcommand{\hbb}[1]{\biggl\{#1\biggr\}}
\newcommand{\hBB}[1]{\Biggl\{#1\Biggr\}}

\newcommand{\abs}[1]{\lvert #1 \rvert}
\newcommand{\absb}[1]{\bigl\lvert #1 \bigr\rvert}
\newcommand{\absB}[1]{\Bigl\lvert #1 \Bigr\rvert}
\newcommand{\absbb}[1]{\biggl\lvert #1 \biggr\rvert}
\newcommand{\absBB}[1]{\Biggl\lvert #1 \Biggr\rvert}

\newcommand{\norm}[1]{\lVert #1 \rVert}
\newcommand{\normb}[1]{\bigl\lVert #1 \bigr\rVert}

\newcommand{\normbb}[1]{\biggl\lVert #1 \biggr\rVert}

\newcommand{\ang}[1]{\langle #1 \rangle}

\newcommand{\scalar}[2]{\langle#1 \mspace{2mu}, #2\rangle}
\newcommand{\scalarb}[2]{\bigl\langle#1 \mspace{2mu}, #2\bigr\rangle}

\newcommand{\scalarbb}[2]{\biggl\langle#1 \,\mspace{2mu},\, #2\biggr\rangle}

\DeclareMathOperator{\diag}{diag}
\DeclareMathOperator{\tr}{tr}
\DeclareMathOperator{\Tr}{Tr}

\DeclareMathOperator{\re}{Re}

\DeclareMathOperator{\spec}{spec}

\newcommand{\wick}[1]{{\col\!#1\!\col}}
\newcommand{\ops}{b}

\title{The mean-field limit of quantum Bose gases \\ at positive temperature}

\author{	
J\"urg Fr\"ohlich
\and Antti Knowles
\and Benjamin Schlein
\and Vedran Sohinger
}

\begin{document}

\maketitle

\begin{abstract}
We prove that the grand canonical Gibbs state of an interacting quantum Bose gas converges to the Gibbs measure of a nonlinear Schr\"odinger equation in the mean-field limit, where the density of the gas becomes large and the interaction strength is proportional to the inverse density. Our results hold in dimensions $d \leq 3$. For $d > 1$ the Gibbs measure is supported on distributions of negative regularity and we have to renormalize the interaction.
More precisely, we prove the convergence of the relative partition function and of the reduced density matrices in the $L^r$-norm with optimal exponent $r$. Moreover, we prove the convergence in the $L^\infty$-norm of Wick-ordered reduced density matrices, which allows us to control correlations of Wick-ordered particle densities as well as the asymptotic distribution of the particle number. Our proof is based on a functional integral representation of the grand canonical Gibbs state, in which convergence to the mean-field limit follows formally from an infinite-dimensional stationary phase argument for ill-defined non-Gaussian measures. We make this argument rigorous by introducing a white-noise-type auxiliary field, through which the functional integral is expressed in terms of propagators of heat equations driven by time-dependent periodic random potentials and can, in turn, be expressed as a gas of interacting Brownian loops and paths. When the gas is confined by an external trapping potential, we control the decay of the reduced density matrices using excursion probabilities of Brownian bridges.
\end{abstract}

\tableofcontents

\section{Introduction}

\subsection{Overview}
A quantum system of $n$ spinless bosons of mass $m$ in a region $\Lambda$ of physical space $\R^d$ is described by its Hamiltonian
\begin{equation*}
\bb H_n \deq - \sum_{i = 1}^n \frac{\Delta_i}{2m} + \frac{\lambda}{2} \sum_{i,j = 1}^n v(x_i - x_j)
\end{equation*}
acting on the space $\cal H_n$ of square-integrable wave functions that are symmetric in their arguments $x_1, \dots, x_n$ and supported in $\Lambda^n$. Here $\Delta_i$ is the Laplacian in the variable $x_i$ (with appropriate boundary conditions), $\lambda$ is a coupling constant, and $v$ is the two-body interaction potential. For conciseness, in this overview we assume that $\Lambda$ is a torus. We describe the system in the grand canonical ensemble at positive temperature, characterized by the density matrix
\begin{equation} \label{gc_density}
\frac{1}{Z} \bigoplus_{n \in \N} \ee^{-\beta(\bb H_n - \mu n)}
\end{equation}
acting on Fock space $\cal F = \bigoplus_{n \in \N} \cal H_n$, where $\beta < \infty$ is the inverse temperature, $\mu$ is the chemical potential, and $Z$ is a normalization factor.

In the grand canonical ensemble \eqref{gc_density}, the parameters characterizing the state of the system are the mass of the particles $m$, the coupling constant $\lambda$, the inverse temperature $\beta$, and the chemical potential $\mu$. Henceforth, we set $\beta = 1$ because it can be absorbed into the three remaining parameters. Moreover, we replace the parameters $m$ and $\mu$ with the new parameters $\nu \deq 1/m$ and $\kappa \deq -\mu m$, so that the grand canonical ensemble \eqref{gc_density} can be written as
\begin{equation} \label{H_gc}
\frac{1}{Z} \bigoplus_{n \in \N} \ee^{-H_n}\,,
\end{equation}
where
\begin{equation}
\label{Hamiltonian_H_n}
H_n
\deq \nu \sum_{i = 1}^n h_i + \frac{\lambda}{2} \sum_{i,j = 1}^n v(x_i - x_j)\,,
\end{equation}
and $h_i$ denotes the one-particle Hamiltonian $h = -\Delta / 2 + \kappa$ in the variable $x_i$. (As we shall see, in dimensions $d > 1$ it is necessary to renormalize the interaction, which amounts to a shift in the value of $\kappa$.)

The goal of this paper is to analyse the \emph{mean-field} limit $\nu \rightarrow 0$, with $\lambda = \nu^2$, of the grand canonical ensemble \eqref{H_gc} in dimensions $d = 1,2,3$. Physically, the mean-field limit $\nu \to 0$ corresponds to a high-density limit where the mass $m$ (or the temperature $1/\beta$) tends to infinity. We refer to the end of this subsection for a more detailed physical interpretation of our choices of parameters and of the mean-field limit. The reason for choosing to scale $\lambda = \nu^2$ is apparent already for $d=1$, where the expected number of particles is proportional to $1/\nu$ (this is an immediate consequence of Theorem \ref{thm:L2} below); in order to obtain an interaction energy in \eqref{Hamiltonian_H_n} that is comparable to the kinetic energy, we therefore require $\lambda \propto 1/n^2 \propto \nu^2$.

The grand canonical ensemble \eqref{H_gc} can be conveniently written using bosonic creation and annihilation operators $a^*(x)$ and $a(x)$ (see \eqref{def_b2} and \eqref{def_b1} below) as $\frac{1}{Z} \ee^{-H}$, where $H$ denotes the Hamiltonian on Fock space $\cal F$ given by
\begin{equation} \label{H_intro}
H = \bigoplus_{n \in \N} H_n = \nu \int_\Lambda \dd x \, a^*(x) (h a) (x)  + \frac{\lambda}{2} \int_\Lambda \dd x \, \dd \tilde x \, a^*(x) a(x) \, v(x - \tilde x) \, a^*(\tilde x) a (\tilde x)\,.
\end{equation}
Under the mean-field scaling $\lambda = \nu^2$, $H$ can be written in terms of the \emph{rescaled creation and annihilation operators} $a_\nu^\# \deq \sqrt{\nu} a^\#$ as
\begin{equation} \label{H_nu}
H = \int_\Lambda \dd x \, a^*_\nu(x) (h a_\nu) (x)  + \frac{1}{2} \int_\Lambda \dd x \, \dd \tilde x \, a_\nu^*(x) a_\nu(x)\,v(x - \tilde x)\, a_\nu^*(\tilde x) a_\nu (\tilde x)\,.
\end{equation}
The rescaled creation and annihilation operators satisfy the canonical commutation relations
\begin{equation} \label{CCR_intro}
[a_\nu(x), a_\nu^*(\tilde x)] = \nu \delta(x - \tilde x) \,, \qquad [a_\nu(x), a_\nu(\tilde x)] = [a_\nu^*(x), a_\nu^*(\tilde x)] = 0\,.
\end{equation}
Thus, the parameter $\nu$ appears as a \emph{semiclassical parameter} of our theory, playing the role of Planck's constant $\hbar$ in semiclassical analysis. Here, the physical interpretation of $\nu$ is the inverse temperature or the inverse mass of the particles.

In this semiclassical picture, we expect the quantum many-body theory to converge to a classical field theory in the limit $\nu \to 0$, the latter being obtained by replacing the quantum field $a$ by a commuting classical field $\phi$. The commutators \eqref{CCR_intro} then become Poisson brackets
\begin{equation} \label{Poisson}
\{\phi(x),\bar \phi(\tilde x)\}=\ii \delta(x-\tilde x)\,,\quad \{\phi(x),\phi(\tilde x)\}=\{\bar \phi(x),\bar \phi(\tilde x)\}=0\,.
\end{equation}
The theory of the classical field $\phi$ is governed by the \emph{classical action}\footnote{In the present context, $s$ is actually a classical Hamilton function defined on an affine complex phase space; it is sometimes called the Hartree energy functional. We call it \emph{action} because that is what it is called in Euclidean field theory; see e.g.\ \cite{glimm2012quantum, Simon74}.}
\begin{equation} \label{intro_interaction}
s(\phi) \deq \int_\Lambda \dd x \, \bar \phi(x) (h \phi)(x) + \frac{1}{2} \int_\Lambda \dd x \, \dd \tilde x \, \abs{\phi(x)}^2 \, v(x - \tilde x) \, \abs{\phi(\tilde x)}^2\,,
\end{equation}
which is the classical analogue of the Hamiltonian \eqref{H_nu}. At positive temperature, the equilibrium state of the classical field is (formally) described by the probability measure
\begin{equation} \label{intro_field}
\frac{1}{c} \ee^{- s(\phi)} \, \mathrm D \phi\,,
\end{equation}
where $\mathrm D \phi$ is the formal uniform measure $\prod_{x \in \Lambda} \dd \phi(x)$ (with $\dd \phi(x)$ denoting the Lebesgue measure on $\C$) and $c$ is a normalization constant. 
Moreover, the equation of motion generated by the Hamilton function \eqref{intro_interaction}, together with the Poisson bracket \eqref{Poisson}, is the \emph{nonlinear Schr\"{o}dinger (NLS)} equation 
\begin{equation} \label{NLS}
\ii \partial_t \phi (x)= (h \phi) (x)+\int_\Lambda \dd \tilde x\,|\phi(\tilde x)|^2\,v(x-\tilde x)\,\phi(x)\,.
\end{equation}
Formally, the measure \eqref{intro_field} is invariant under the Hamiltonian flow generated by \eqref{NLS}.

Hence, we expect the mean-field limit of an interacting Bose gas at positive temperature to be governed by a \emph{scalar Euclidean field theory} with an action of the form \eqref{intro_interaction}, corresponding to a quartic self-interaction. The construction of measures of the form \eqref{intro_field} is a central topic in many areas of mathematics. We list three main ones, the first two of which are closely related.
\begin{enumerate}[label=(\roman*)]
\item
\emph{Constructive quantum field theory.} When formulated in the Euclidean region, the construction of scalar quantum field models amounts to constructing non-Gaussian measures on distribution spaces; see the classic works \cite{glimm2012quantum,nelson1973free,Simon74} and references given there. The mathematical methods developed in this context have inspired some of the constructions used in the present paper.

\item
\emph{Stochastic nonlinear partial differential equations.} A measure of the form \eqref{intro_field} is the invariant measure of a nonlinear heat equation driven by space-time white noise, which can be regarded as the Langevin equation for a time-dependent field $\phi$ with potential given by the action \eqref{intro_interaction}. For \eqref{intro_field}, this is the dynamical $\phi^4_d$ model with nonlocal interaction. Constructing measures of the form \eqref{intro_field} by exhibiting them as invariant measures of stochastic nonlinear partial differential equations is the goal of stochastic quantization \cite{parisi1981perturbation, lebowitz1988statistical, nelson1966derivation}. See e.g.\ \cite{hairer2014theory,gubinelli2015paracontrolled, kupiainen2016renormalization,da2003strong} for recent developments.

\item
\emph{Probabilistic Cauchy theory of nonlinear dispersive equations.} Gibbs measures for the NLS have proven to be a powerful tool for constructing almost sure global solutions with random initial data of low regularity. One considers the flow of the NLS \eqref{NLS} with random initial data distributed according to \eqref{intro_field}. The invariance of the measure \eqref{intro_field} under the NLS flow (in low dimensions) serves as a substitute for energy conservation, which is not available owing to the low regularity of the solutions.
See for instance the seminal works \cite{bourgain1994periodic,bourgain1994_z,bourgain1996_2d,bourgain1997invariant,bourgain2000_infinite_volume,lebowitz1988statistical} as well as \cite{Carlen_Froehlich_Lebowitz_2016, Carlen_Froehlich_Lebowitz_Wang_2019,GLV1,GLV2,McKean_Vaninsky2,NORBS,NRBSS,BourgainBulut4,BrydgesSlade,BurqThomannTzvetkov,FOSW,Thomann_Tzvetkov} and references given there for later developments.

\end{enumerate}
The main difficulty in all of these works, which we also face in this paper, is that, in dimensions larger than one, the field $\phi$ is almost surely a distribution of negative regularity, and hence the interaction term in \eqref{intro_interaction} is ill-defined. This is an \emph{ultraviolet} problem: a divergence for large momenta producing small-scale singularities in the field.
The solution to this problem is to \emph{renormalize} the interaction by subtracting infinite counterterms that render the interaction term well defined.

In this paper we focus on the analysis of interacting Bose gases in dimensions $d = 1,2,3$, and we renormalize the interaction term in \eqref{intro_interaction} by \emph{Wick ordering} the factors $\abs{\phi}^2$. In contrast to the measure in \eqref{intro_field}, the many-body grand canonical ensemble \eqref{H_gc} is well defined \emph{without} renormalization, because it possesses an intrinsic ultraviolet cutoff at energies of order $1/\nu$. However, without a $\nu$-dependent, finite renormalization, it does not have a well defined mean-field limit as $\nu \to 0$. Thus, we also need to Wick order the interaction term of the quantum Hamiltonian \eqref{H_intro}. By varying $\kappa$ in the definition of $h$, we may (at least for systems confined to a torus; see Section \ref{sec:traps_intro} below for details) rewrite \eqref{H_nu} as

\begin{equation*}
H = \int_\Lambda \dd x \, a^*_\nu(x) (h a_\nu) (x)  + \frac{1}{2} \int_\Lambda \dd x \, \dd \tilde x \, \pb{a_\nu^*(x) a_\nu(x) - \varrho} \, v(x - \tilde x) \, \pb{ a_\nu^*(\tilde x) a_\nu (\tilde x) - \varrho}\,,
\end{equation*}
where $\varrho$ is a $\nu$-dependent parameter. In order to obtain a well-defined limit $\nu \to 0$, keeping the parameter 
$\kappa$ fixed, we choose $\varrho$ to be the mean density of particles in the gas, which, for $d > 1$, \emph{diverges} as $\nu \to 0$. This suggests a plausible physical interpretation of the Wick ordering of the factors $\abs{\phi}^2$ in \eqref{intro_interaction}, which amounts, formally, to replacing $\abs{\phi}^2$ with $\abs{\phi}^2 - \infty$. Under our assumptions, the analysis of the classical field theory \eqref{intro_field} is straightforward and can essentially be found in the existing literature. In contrast, the analysis of the convergence of the grand canonical ensemble \eqref{H_gc} to the classical field theory \eqref{intro_field} is rather involved, owing mainly to the non-commutativity of the field $a$.

Our main result is the convergence of the quantum state \eqref{H_gc} to the classical state \eqref{intro_field} in dimensions $d = 1,2,3$. More precisely, we prove that the relative partition function of \eqref{H_gc} converges to that of \eqref{intro_field}, and that the reduced density matrices of \eqref{H_gc} converge to the correlation functions of \eqref{intro_field} in the $L^r$-norm with optimal exponent $r$ (depending on $d$). Our only assumptions on the interaction potential $v$ are that it be continuous and of positive type, which guarantees that the gas is thermodynamically stable (i.e.\ that the interaction term in \eqref{intro_interaction} is nonnegative). Furthermore, we prove convergence in the $L^\infty$-norm of \emph{renormalized} (or \emph{Wick-ordered}) reduced density matrices to the corresponding renormalized correlation functions of the field $\phi$. Without renormalization, these objects exhibit small-scale singularities at coinciding arguments in dimensions $d > 1$. This convergence allows us to analyse correlations of the renormalized particle density
as well as the asymptotic distribution of the particle number.
We state and prove our results for translation-invariant systems, where the particles move on a torus, as well as for particles moving in Euclidean space confined by an external trapping potential, in which case we obtain precise control on the decay of the renormalized reduced density matrices. We believe our decay estimates to be optimal.

We now review previous results on mean-field limits in thermal equilibrium. The first result establishing convergence of \eqref{H_gc} to \eqref{intro_field} is \cite{lewin2015derivation}, where an approach based on the Gibbs variational principle and the quantum de Finetti theorem is used. For $d=1$, the authors of \cite{lewin2015derivation} prove convergence of the relative partition function and of all correlation functions; these results are extended by the same authors to sub-harmonic traps in \cite{lewin2018gibbs}. In \cite{lewin2015derivation}, they also present some results in dimensions $d=2,3$, but only for systems of particles interacting through a smooth, non translation-invariant potential, which obviates the need for renormalization. A different approach to these problems is developed in \cite{frohlich2017gibbs}; it is based on a perturbative expansion of the state in powers of the interaction term for both the quantum and the classical theories, on the comparison of (fully expanded) terms in the two series, and on Borel resummation. This approach yields a new proof of convergence for the one-dimensional systems, and, for $d = 2,3$, proves convergence towards the Wick-ordered version of \eqref{intro_field}, starting, however, from a slightly modified grand canonical state. In \cite{S19}, the results of \cite{frohlich2017gibbs} for $d=2,3$, which are restricted to bounded interactions, are extended to potentials in the optimal $L^r$ spaces. Moreover, in \cite{frohlich2019microscopic}, convergence of the \emph{time-dependent} correlation functions is proved for $d = 1$. In \cite{LNR2}, by a highly nontrivial extension of the the techniques of \cite{lewin2015derivation}, the authors establish convergence of the grand canonical state to the Wick-ordered measure \eqref{intro_field} for $d = 2$. Subsequently, the authors of \cite{LNR2} announced an extension to $d = 3$; this is the content of the preprint \cite{LNR3} which appeared simultaneously with the current paper.

In this paper we develop a new approach to the problem, which is based on a functional integral formulation of the grand canonical ensemble \eqref{H_gc}. It is inspired, in part, by methods developed in Euclidean field theory, in \cite{ginibre1971some} on quantum gases at high temperature, and in \cite{Sym68}, where the Euclidean $\phi^4$ field theory is represented as a gas of interacting Brownian loops and paths.
Our starting point is to formally rewrite the grand canonical ensemble \eqref{H_gc} as a complex measure, $\frac{1}{C} \ee^{-\cal S(\Phi)} \, \mathrm D \Phi$, for a Euclidean field theory, where $\Phi$ is a complex field on space-time $[0,\nu] \times \Lambda$, $\cal S$ is an action given by
\begin{equation} \label{S_intro}
\cal S(\Phi) \deq \int_0^\nu \dd \tau \int_\Lambda \dd x \, \bar \Phi(\tau,x) \pb{\partial_\tau + \kappa - \Delta/2} \Phi(\tau,x) + 
\frac{\lambda}{2 \nu} \int_0^\nu \dd \tau \int_\Lambda \dd x \, \dd \tilde x \, \abs{\Phi(\tau,x)}^2 \, v(x - \tilde x) \, \abs{\Phi(\tau, \tilde x)}^2\,,
\end{equation}
and $C$ is a normalization constant. Thus, we replace the non-commuting quantum field $a$ with a commuting classical field $\Phi$, at the expense of having a field on space-time with a complex action. The latter property represents in fact a serious complication, and it is known (see Section \ref{sec:fct_int} below for details) that even if $\lambda=0$, i.e.\ for an ideal Bose gas, the action \eqref{S_intro} does not give rise to a well-defined Gaussian measure on the space of fields, owing to the unbounded anti-self-adjoint operator $\partial_\tau$ appearing in \eqref{S_intro}. However, heuristically, the formal functional integral determined by the action $\mathcal{S}(\Phi)$ in \eqref{S_intro} is very useful in guessing what ought to be true: by a simple rescaling of the field $\Phi$ and applying a formal stationary phase argument, we see that the measure $\frac{1}{C} \ee^{-\cal S(\Phi)} \mathrm D \Phi$ should converge, as $\lambda = \nu^2 \to 0$, to the measure \eqref{intro_field} (which is a well-defined probability measure). Thus, our proof is inspired by a formal infinite-dimensional stationary phase argument applied to (non-existent) measures of the form $\frac{1}{C} \ee^{-\cal S(\Phi)} \, \mathrm D \Phi$.
In order to make this argument rigorous, we introduce, in the form of a Hubbard-Stratonovich transformation, an auxiliary Gaussian field, which is white noise in the time direction (and hence has to be suitably regularized). We then express the resulting functional integral in terms of propagators of heat equations driven by periodic time-dependent random potentials, which can, in turn, be expressed in terms of interacting Brownian loops and paths. This allows us to bypass the actual construction of the (non-existent) measure $\frac{1}{C} \ee^{-\cal S(\Phi)} \mathrm D \Phi$, replacing it with a linear functional on a suitably large class of functions of the field $\Phi$. In particular, we express the reduced density matrices of \eqref{H_gc} as correlation functions of the field $\Phi$. As explained above, for $d > 1$ both field theories $\frac{1}{c} \ee^{-s(\phi)} \mathrm D \phi$ and $\frac{1}{C} \ee^{-\cal S(\Phi)} \mathrm D \Phi$ have to be renormalized by Wick ordering the actions \eqref{intro_interaction} and \eqref{S_intro}, respectively; this can be achieved by including suitably chosen divergent phases in the integrals over $\phi$ and $\Phi$. In this new formulation, the proof of convergence of well-defined objects corresponding to the formal measures $\frac{1}{C} \ee^{-\cal S(\Phi)} \mathrm D \Phi$ can be reduced to proving convergence of certain functional integrals involving Riemann sums to their limiting expressions. This is established by using simple continuity properties of Brownian bridges. When the particles of the Bose gas are confined by an external trap in Euclidean space, we obtain precise decay estimates and integrability of the reduced density matrices from excursion probabilities of Brownian bridges. For a more detailed overview of our proof, we refer to Section \ref{sec:fct_int} below.

Aside from its conceptual simplicity, the strengths of our approach may be summarized as follows. First, it naturally yields control of Wick-ordered reduced density matrices in the supremum norm. As a consequence, it controls all moments of the particle number, hence also characterizing its asymptotic distribution. Second, it works for an arbitrary continuous interaction potential. Third, in the presence of an external trap in Euclidean space, it yields precise bounds on the rate of decay of the Wick-ordered reduced density matrices, which we believe to be optimal. Hence, we obtain sharp control of both the ultraviolet and infrared behaviour of the Wick-ordered reduced density matrices.

Next, we clarify the physical interpretation of our results. We begin by noting that the mean-field scaling, $\lambda = \nu^{2}$, $\nu \rightarrow 0$, corresponds to a high-temperature and high-density limit of the gas confined to a container of fixed volume. This regime places the gas just above the critical temperature of Bose-Einstein condensation, or, equivalently, just below the critical density. To see this, we reintroduce the inverse temperature $\beta$ in front of the Hamiltonian \eqref{H_intro} that appears in the grand canonical ensemble $\frac{1}{Z} e^{-\beta H}$ from \eqref{H_gc} and in the measure $\frac{1}{c} \ee^{- \beta s(\phi)} \, \mathrm D \phi$ from \eqref{intro_field}. Taking $\beta \to \infty$ after having taken the limit $\nu \to 0$ then yields the delta measure at the minimizer of \eqref{intro_interaction}, which corresponds precisely to Bose-Einstein condensation.
An alternative interpretation of our results is obtained by considering a gas of particles on a torus with sides of length $L$ at a fixed temperature and small chemical potential, in the thermodynamic limit $L \to \infty$. For simplicity, we consider the ideal Bose gas, $\lambda = 0$ (the interaction potential can be easily included by an appropriate rescaling). Consider particles of unit mass at inverse temperature $\beta$ with chemical potential $\mu = -\kappa / L^2$ confined to a torus with sides of length $L$. Anticipating the tools summarized in Section \ref{sec:qf} below, we find (see \eqref{quasi_free} with $b = \ee^{-\beta (-\Delta / 2 + \kappa/L^2)}$ and use \eqref{QWT_i_1}) that the expected number of particles is
\begin{equation} \label{R_sum}
\tr \frac{1}{\ee^{\beta (-\Delta / 2 + \kappa/L^2)} - 1} =  \sum_{k \in \Z^d} \frac{1}{\ee^{\frac{\beta}{L^2} ((2 \pi)^2 \abs{k}^2/2 + \kappa)} - 1} \asymp
\begin{cases}
L^2 & \text{if } d = 1
\\
L^2 \log L & \text{if }d = 2
\\
L^3 & \text{if }d = 3\,.
\end{cases}
\end{equation}
We conclude that, in the thermodynamic limit $L \to \infty$, the particle density, which is equal to $\text{\eqref{R_sum}}$ divided by the volume $L^d$, diverges to $\infty$ if $d < 3$, and if $d = 3$ it converges from below to 
the critical density for Bose-Einstein condensation, $\rho_c = \beta^{-3/2} \int_{\R^3} \frac{\dd k}{(2 \pi)^3} \frac{1}{\ee^{\abs{k}^2/2} - 1}< \infty$. By rescaling the torus to have unit side length, we can write the left-hand side of \eqref{R_sum} as $\tr \frac{1}{\ee^{\nu h} - 1}$, where $\nu \deq \beta/L^2$ and $h = -\Delta/2 + \kappa$ acts on wave functions with support in the unit torus. This is precisely the expected number of particles in the state \eqref{H_gc} with $\lambda = 0$, and hence it approaches the Gaussian field limit as $\nu \to 0$ (i.e.\ $L \to \infty$ for fixed 
$\beta$). We conclude that a Bose gas on the torus with sides of length $L$ at a fixed temperature, in the limit where $L \to \infty$, is equivalently described by a Bose gas on the unit torus in the limit $\nu \to 0$. For $d = 3$, it has a particle density below, but asymptotically equal to, the critical density for Bose-Einstein condensation. The same claim remains valid for interacting Bose gases provided that the interaction potential is suitably rescaled (we omit further details).

We believe that the functional integral developed in this paper has potential to address several other questions of physical interest. Some of them are reviewed in the survey \cite{frohlich2020path}. In particular, we hope that it will prove useful in understanding Bose-Einstein condensation for an interacting Bose gas at low temperature. Our techniques also provide a natural regularization of Euclidean field theories in terms of a physically realistic microscopic model of statistical mechanics.

We conclude this overview with a brief survey of various related results. The mean-field limit of an interacting Bose gas has also been studied at zero temperature, i.e.\ for $\beta = \infty$. At zero temperature, the ground state energy per particle 
of \eqref{Hamiltonian_H_n}, with $\lambda=\nu^2 = 1/n^2$, converges to the minimum of the Hartree energy functional \eqref{intro_interaction} on the space of fields $\phi \in L^2 (\Lambda)$ with $\| \phi \| = 1$. Moreover, the ground state vector exhibits complete Bose-Einstein condensation into the minimizer $\phi_*$ of \eqref{intro_interaction}. This means that all particles, up to a fraction that vanishes in the limit of large $n$, can be described by the same one-particle wave function $\phi_*$. These and other results are reviewed in the book \cite{lieb2006mathematics} and in \cite{Bach,BL,Kies,sol,LY,LNR0,FSV,RW}.
Results on the low-energy excitation spectrum can be found in \cite{froh_houches2, seir,GrSei,LNSS}. In \cite{Pizzo17}, a convergent expansion is obtained for the ground state of a truncation of the Bose gas in the vicinity of the mean-field limit. The mean-field limit of a Bose gas has also been studied in a dynamical setting, where one analyses the time evolution $\ee^{-\ii t H_n / \nu} \phi^{\otimes n}$ of a product state $\phi^{\otimes n}$ under the dynamics generated by the Hamiltonian \eqref{Hamiltonian_H_n}. In the seminal work \cite{hepp} it is shown that $\ee^{-\ii t H_n / \nu} \phi^{\otimes n} \approx \phi_t^{\otimes n}$ in some appropriate sense, where $\phi_t$ is the solution of the NLS \eqref{NLS} with initial data $\phi_0 = \phi$. We refer to the introduction of \cite{frohlich2017gibbs} for references to later developments.

\subsection{One-particle space and interaction potential}

Let $d = 1,2,3$ and let $\Lambda \subset \R^d$ be a domain, which, in the sequel, will be either a torus or all of Euclidean space. We use the shorthand $\int \dd x \, (\cdot)\deq \int_{\Lambda} \dd x \, (\cdot)$ to denote integration over $\Lambda$ with respect to Lebesgue measure. We define the \emph{one-particle space} $\cal H \deq L^2(\Lambda; \C)$. We denote by $\scalar{\cdot}{\cdot}$ the inner product of the space $\cal H$, which is by definition linear in the second argument. Moreover, we denote by $\tr$ the trace on operators acting on $\cal H$.

For $n \in \N$, we denote by $P_n$ the orthogonal projection onto the symmetric subspace of $\cal H^{\otimes n}$; explicitly, for $\Psi_n \in \cal H^{\otimes n}$,
\begin{equation} \label{def_Pn}
P_n \Psi_n(x_1, \dots, x_n) \deq \frac{1}{n!} \sum_{\pi \in S_n} \Psi_n(x_{\pi(1)}, \dots, x_{\pi(n)})\,,
\end{equation}
where $S_n$ is the group of permutations on $\{1, \dots, n\}$. For $n \in \N^*$, we define the $n$-particle space as $\cal H_n \deq P_n \cal H^{\otimes n}$.

We identify a closed operator $\ops$ on $\cal H_n$ with its \emph{Schwartz operator kernel}, denoted by $\ops_{\f x, \tilde {\f x}}$, which satisfies $(\ops f)(\f x) = \int \dd \tilde {\f x} \, \ops_{\f x,\tilde {\f x}} f(\tilde {\f x})$ for $f$ in the domain of $\ops$. The kernel $\ops_{\f x, \tilde {\f x}}$ is in general a tempered distribution (see e.g.\ \cite[Corollary V.4.4]{RS1}).
For $1 \leq r \leq \infty$, we denote by $\norm{\ops}_{L^r}$ the $L^r(\Lambda^n \times \Lambda^n)$-norm of the kernel $\ops_{\f x, \tilde {\f x}}$ in the arguments $\f x, \tilde {\f x}$.

The \emph{one-particle Hamiltonian} $h$ is a positive operator on $\cal H$.  We always assume that $h$ has a compact resolvent and that
\begin{equation} \label{tr_h_assump}
\tr h^{-s} \;<\; \infty
\end{equation}
for some $s > 0$. In the sequel, we shall always impose \eqref{tr_h_assump} for $s = 2$, in which case the condition \eqref{tr_h_assump} will be satisfied by our Assumption \ref{ass:torus} when $\Lambda$ is the torus and Assumption \ref{ass:trap} when $\Lambda$ is Euclidean space.

The \emph{two-body interaction potential} $v \col \Lambda \to \R$ is always assumed to be even, meaning that $v(x) = v(-x)$ for all $x \in \Lambda$, and of positive type, meaning that its Fourier transform is a nonnegative measure.

\subsection{Classical field theory} \label{sec:classical_field}
In this subsection we define the classical interacting field theory and its correlation functions. We begin by defining the classical free field. For $r \in \R$ denote by $\cal H^r$ the Hilbert space of complex-valued tempered distributions on $\Lambda$ with inner product $\scalar{f}{g}_{\cal H^r} \deq \scalar{f}{h^r g}$. In particular, $\cal H^0 = \cal H$.
We define the \emph{classical free field} as the abstract Gaussian process indexed by the Hilbert space $\cal H^{-1}$. 

An explicit construction of the classical free field goes as follows. Let
\begin{equation} \label{h_spectrum}
h = \sum_{k \in \N} \lambda_k u_k u_k^*
\end{equation}
be the spectral decomposition of $h$, with eigenvalues $\lambda_k > 0$ and normalized eigenfunctions $u_k \in \cal H$.
Let $(X_k)_{k \in \N}$ be a family of independent standard complex Gaussian random variables\footnote{We recall that $X$ is a standard complex Gaussian if it is Gaussian and satisfies $\E X = 0$, $\E X^2 = 0$, and $\E \abs{X^2} = 1$, or, equivalently, if it has law $\pi^{-1} \ee^{- \abs{z}^2} \dd z$ on $\C$, where $\dd z$ denotes Lebesgue measure.}.
From \eqref{tr_h_assump} we easily find that the random series
\begin{equation} \label{phi_series}
\phi = \sum_{k \in \N} \frac{X_k}{\sqrt{\lambda_k}} u_k\,,
\end{equation}
converges\footnote{In fact, an application of Wick's rule shows that the convergence holds in $L^m$ for any $m < \infty$.} in $L^2$ in the topology of $\cal H^{1 - s}$, where $s$ is as in \eqref{tr_h_assump}. We call $\phi$ the \emph{classical free field}. We denote the law of $\phi$ by $\mu_{h^{-1}}$, since by construction $\phi$ is a complex Gaussian field with covariance $h^{-1}$. Indeed, for $f,g \in \cal H^{-1}$ we have
\begin{equation} \label{cov_h}
\int \mu_{h^{-1}}(\dd \phi) \, \scalar{f}{\phi} \, \scalar{\phi}{g}  = \scalar{f}{h^{-1} g}\,.
\end{equation}

Next, we define the interacting field theory. In dimensions $d = 2,3$, the interaction has to be renormalized by Wick ordering. To that end, for $K \in \N$ we define the finite-dimensional orthogonal projection
\begin{equation} \label{def_Pk}
\cal P_K \deq \sum_{k = 0}^{K} u_k u_k^*\,.
\end{equation}
Define the classical truncated density
\begin{equation} \label{def_varrho_K}
\varrho_K(x) \deq \int \mu_{h^{-1}}(\dd \phi) \, \abs{\cal P_K \phi(x)}^2
\end{equation}
and the Wick-ordered truncated interaction
\begin{equation} \label{def_W_K}
W_K^v \deq \frac{1}{2} \int \dd x \, \dd \tilde x \, \pB{\abs{\cal P_K \phi(x)}^2 - \varrho_K(x)} \, v(x - \tilde x)\, \pB{\abs{\cal P_K \phi(\tilde x)}^2 - \varrho_K(\tilde x)}\,.
\end{equation}
It is well known (see Lemma \ref{lem:constr_W} (i) below) that if $v$ is bounded and $h$ satisfies \eqref{tr_h_assump} with $s = 2$ then $W_K^v$ converges in $L^2(\mu_{h^{-1}})$, as $K \to \infty$, to a nonnegative random variable which we denote by $W^v \equiv W$.

We define the (relative) partition function of the classical field theory as
\begin{equation} \label{def_z}
\zeta \deq \int \mu_{h^{-1}}(\dd \phi) \, \ee^{-W(\phi)}\,,
\end{equation}
For $p \in \N$, we define the \emph{correlation function}
\begin{equation} \label{def_gamma}
(\gamma_p)_{x_1 \dots x_p, \tilde x_1 \dots \tilde x_p} \deq \frac{1}{\zeta} \int \mu_{h^{-1}}(\dd \phi) \, \ee^{-W(\phi)} \, \bar \phi(\tilde x_1) \cdots \bar \phi(\tilde x_p) \phi(x_1) \cdots \phi(x_p)\,,
\end{equation}
which is the $2p$-th moment of the field $\phi$ under the probability measure $\frac{1}{\zeta} \mu_{h^{-1}}(\dd \phi) \, \ee^{-W(\phi)}$. This measure is sub-Gaussian, and is hence determined by its moments $(\gamma_p)_{p \in \N^*}$. (Note that any moment containing a different number of $\bar \phi$s and $\phi$s vanishes by invariance of the measure $\mu_{h^{-1}}(\dd \phi) \, \ee^{-W(\phi)}$ under the gauge transformation $\phi \mapsto \alpha \phi$, where $\abs{\alpha} = 1$.)

\subsection{Quantum many-body system} \label{sec:quantum_definitions}
In this subsection we define the quantum many-body system and its reduced density matrices.
For $n \in \N^*$ we consider an $n$-body Hamiltonian on $n$-particle space $\cal H_n$ of the form \eqref{Hamiltonian_H_n}.
We define Fock space as the Hilbert space $\cal F \equiv \cal F (\cal H) \deq \bigoplus_{n \in \N} \cal H_n$. We denote by $\Tr_{\cal F}(X)$ the trace of an operator $X$ acting on $\cal F$. On $\cal F$ we define the Hamiltonian $H$ of the system through
\begin{equation} \label{def_H_Fock}
H \deq \bigoplus_{n \in \N} H_n\,.
\end{equation}
We define the \emph{quantum grand canonical density matrix} as the operator $\ee^{-H} / Z$, where
\begin{equation} \label{def_Z}
Z \deq \Tr_{\cal F} (\ee^{-H})
\end{equation}
is the \emph{grand canonical partition function}\footnote{In the physics literature this grand canonical partition function is sometimes denoted by $\Xi$.}.
Throughout the following, we use the notations $H^0 = \bigoplus_{n \in \N} H_n^0$ and $Z^0 = \Tr_{\cal F}(\ee^{-H^0})$ to denote the corresponding free quantities, for which $\lambda = 0$.

On Fock space it is convenient to use creation and annihilation operators. For $f \in \cal H$ we define the bosonic annihilation and creation operators $a(f)$ and $a^*(f)$ on $\cal F$ through their action on a dense set of vectors $\Psi = (\Psi_n)_{n \in \N} \in \mathcal{F}$ as
\begin{align}
\label{def_b2}
\pb{a(f) \Psi}^{(n)}(x_1, \dots, x_n) &= \sqrt{n+1} \int \dd x \, \bar f(x) \, \Psi^{(n+1)} (x,x_1, \dots, x_n)\,,
\\
\label{def_b1}
\pb{a^*(f) \Psi}^{(n)}(x_1, \dots, x_n) &= \frac{1}{\sqrt{n}} \sum_{i = 1}^n f(x_i) \Psi^{(n - 1)}(x_1, \dots, x_{i - 1}, x_{i+1}, \dots, x_n)
\,.
\end{align}
The operators $a(f)$ and $a^*(f)$ are unbounded closed operators on $\cal F$, and are each other's adjoints. They satisfy the canonical commutation relations
\begin{equation} \label{CCR_b}
[a(f), a^*(g)] = \scalar{f}{g} \, 1 \,, \qquad [a(f), a(g)] = [a^*(f), a^*(g)] =0\,,
\end{equation}
where $[X,Y] \deq XY - YX$ denotes the commutator. We regard $a$ and $a^*$ as operator-valued distributions and use the notations
\begin{equation} \label{phi_tau f}
a(f) = \scalar{f}{a} = \int \dd x \, \bar f(x) \, a(x)\,, \qquad
a^*(f) = \scalar{a}{f} = \int \dd x \, f(x) \, a^*(x)\,.
\end{equation}
The distribution kernels $a^*(x)$ and $a(x)$ satisfy the canonical commutation relations
\begin{equation} \label{CCR}
[a(x),a^*(\tilde x)] = \delta(x - \tilde x) \,, \qquad
[a(x),a(\tilde x)] = [a^*(x),a^*(\tilde x)] = 0\,.
\end{equation}
Using the creation and annihilation operators, we may conveniently rewrite $H \equiv H(h)$ from \eqref{def_H_Fock} and \eqref{Hamiltonian_H_n} as
\begin{equation} \label{H_sq}
 H(h)  = \nu \int \dd x \, \dd \tilde x \,  h_{x, \tilde x} \, a^*(x) a(\tilde x) + \frac{\lambda}{2} \int \dd x \, \dd \tilde x \, a^*(x) a(x) \, v(x - \tilde x)\, a^*(\tilde x) a(\tilde x)\,.
\end{equation}

In dimensions $d > 1$, we have to renormalize \eqref{H_sq}, and define the \emph{Wick-ordered Hamiltonian}
\begin{equation} \label{H_ren}
H = \nu \int \dd x \, \dd \tilde x \, h_{x, \tilde x}\, a^*(x) a(\tilde x) + \frac{\lambda}{2} \int \dd x \, \dd \tilde x \, \pbb{a^*(x) a(x) - \frac{\varrho(x)}{\nu}}  \, v(x - \tilde x) \, \pbb{a^*(\tilde x) a(\tilde x) - \frac{\varrho(\tilde x)}{\nu}}\,,
\end{equation}
where
\begin{equation} \label{def_rho}
\varrho(x) \deq \nu \Tr_{\cal F} \pbb{a^*(x) a(x) \, \frac{\ee^{-H^0}}{Z^0}}
\end{equation}
is the particle density in the equilibrium state of the ideal quantum gas. A simple computation shows that, up to an irrelevant additive constant, \eqref{H_ren} is of the form \eqref{H_sq} with the modified one-particle Hamiltonian $\tilde h \deq h - \frac{\lambda}{\nu^2} v * \varrho$: 
\begin{equation} \label{H_ren_H_sq}
H = H(\tilde h) + \text{constant}\,.
\end{equation}
For a detailed discussion on the relationship between \eqref{H_sq} and \eqref{H_ren}, we refer to the discussion of the counterterm problem below, after Remark \ref{rem:decay_hat_gamma} in Section \ref{sec:traps_intro}.  We remark that in first quantized notation, the Wick-ordered Hamiltonian \eqref{H_ren_H_sq} is $H = \bigoplus_{n \in \N} H_n$, with $H_n$ given in \eqref{H_n_first_quantized} below.

For $p \in \N^*$ we define the reduced $p$-particle density matrix $\Gamma_p$ with kernel
\begin{equation} \label{def_gammap}
(\Gamma_p)_{x_1 \dots x_p, \tilde x_1 \dots \tilde x_p} \deq \Tr_{\cal F} \pbb{a^*(\tilde x_1) \cdots a^*(\tilde x_p) a(x_1) \cdots a(x_p) \, \frac{\ee^{-H}}{Z}}\,.
\end{equation}

\subsection{Results I: torus} \label{sec:results1}
Our results are on the \emph{mean-field limit} $\lambda = \nu^2 \to 0$. In this subsection we state our results for the domain $\Lambda$ being a torus. In particular, the problem is translation invariant. We make the following assumption on the one-particle Hamiltonian.

\begin{assumption} \label{ass:torus}
The domain $\Lambda \equiv \Lambda_{L,d} = [-L/2,L/2]^d$ is the $d$-dimensional torus, with $d = 1,2,3$, of side length $L>0$, and $h = \kappa - \Delta/2$, where $\kappa > 0$ and $\Delta$ is the (negative definite) Laplacian on $\Lambda$ with periodic boundary conditions.
\end{assumption}

It is easy to check that under Assumption \ref{ass:torus} the condition \eqref{tr_h_assump} holds for $s = 2$.
We make the following assumption on the interaction potential.

\begin{assumption} \label{ass:v}
The interaction potential $v : \Lambda \to \R$ is even, periodic, of positive type, and continuous.
\end{assumption}

We begin by stating the convergence of the relative partition function and the reduced density matrices, Theorem \ref{thm:L2}. Then we state the convergence of the Wick-ordered particle densities, Theorem \ref{thm:density}. Both are special cases of our main result on Wick-ordered reduced density matrices, Theorem \ref{thm:main}.

\begin{theorem}[Convergence of partition function and reduced density matrices] \label{thm:L2}
Suppose that Assumptions \ref{ass:torus} and \ref{ass:v} hold. Let the Hamiltonian $H$ be given as in \eqref{H_ren} and set $\lambda \deq \nu^2$. Then $Z/Z^0 \to \zeta$ as $\nu \to 0$. Moreover, suppose that $1 \leq r \leq \infty$ if $d = 1$, that $1 \leq r < \infty$ if $d = 2$, and that $1 \leq r < 3$ if $d = 3$. Then for all $p \in \N^*$, as $\nu \to 0$,
\begin{equation} \label{conv_Lr}
\nu^p \, \Gamma_p \overset{L^r}{\longrightarrow} \gamma_p\,.
\end{equation}
The convergence is also locally uniform for distinct arguments $x_1, \dots, x_p, \tilde x_1, \dots, \tilde x_p$.
\end{theorem}

The assumption imposed in Theorem \ref{thm:L2} on the exponent $r$ is optimal. To see this, consider the ideal quantum gas, where $v = 0$. Then \eqref{cov_h} yields the singular behaviour near the diagonal
\begin{equation} \label{gamma_1_free}
(\gamma_1)_{x, \tilde x} = (h^{-1})_{x, \tilde x} = ((\kappa - \Delta/2)^{-1})_{x, \tilde x} \asymp
\begin{cases}
1 & \text{if } d = 1
\\
\log\abs{x - \tilde x}^{-1} &\text{if } d = 2
\\
\abs{x - \tilde x}^{-1} & \text{if } d = 3\,.
\end{cases}
\end{equation}
Moreover, by the Wick theorem (see Lemma \ref{Quantum Wick theorem} below) we have $\nu \Gamma_1 = \nu \frac{\ee^{-\nu h}}{1 - \ee^{-\nu h}} = \nu \sum_{n \geq 1} \ee^{-n \nu h}$, whose operator kernel we can write using the heat kernel $\psi^t(x) \deq (\ee^{t\Delta / 2})_{x,0}$ as
\begin{equation} \label{Gamma_1_free}
\nu (\Gamma_1)_{x, \tilde x} = \nu \sum_{n \geq 1} \ee^{-\kappa \nu n} \psi^{n \nu}(x  - \tilde x)\,,
\end{equation}
which is in fact continuous (but of course not uniformly in $\nu$). We conclude that \eqref{conv_Lr} can only hold for $r$ as in Theorem \ref{thm:L2}.

In the grand canonical ensemble \eqref{gc_density}, the \emph{pressure}\footnote{In the mathematical literature, this quantity is sometimes somewhat incorrectly referred to as the free energy.} of the gas is equal to $p = \frac{1}{\beta \abs{\Lambda}} \log Z$. Thus, setting $\beta = 1$, Theorem \ref{thm:L2} yields the first correction to the pressure of an interacting Bose gas above the critical temperature,
\begin{equation*}
p = \frac{1}{\abs{\Lambda}} \log Z = \frac{1}{\abs{\Lambda}} \log Z^0 + \frac{1}{\abs{\Lambda}} \log \zeta + \oo(1)\,.
\end{equation*}

For $d = 1$, Theorem \ref{thm:L2} implies that correlation functions of \emph{the particle density}
\begin{equation*}
\fra N(x) \deq a^*(x) a(x)
\end{equation*}
are, after scaling by $\nu$, asymptotically given by corresponding correlation functions of the classical \emph{mass density}
\begin{equation*}
\fra n(x) \deq \abs{\phi(x)}^2\,,
\end{equation*}
i.e.
\begin{equation*}
\Tr_{\cal F} \pbb{\nu^p \, \fra N(x_1) \cdots \fra N(x_p) \, \frac{\ee^{-H}}{Z}} \overset{L^\infty}{\longrightarrow}
\frac{1}{\zeta} \int \mu_{h^{-1}}(\dd \phi) \, \ee^{-W(\phi)} \, \fra n(x_1) \cdots \fra n(x_p)\,.
\end{equation*}
However, this result is wrong for $d > 1$. In fact, for $d > 1$, $\fra n(x)$ is ill-defined, as the pointwise product of distributions $\phi$ and $\ol \phi$ of negative regularity does not make sense. This is also apparent in that the diagonal of $\gamma_1$ in the free case is divergent for $d > 1$; see \eqref{gamma_1_free}. On the quantum side, we find from \eqref{Gamma_1_free} that the expected particle density in the free case is
\begin{equation} \label{exp_fraN1}
\Tr_{\cal F} \pbb{\nu \, \fra N(x) \, \frac{\ee^{-H}}{Z}} = \nu (\Gamma_1)_{x, x} = \nu \sum_{k \geq 1} \ee^{-\kappa \nu k} \psi^{k \nu}(0)\,.
\end{equation}
An elementary analysis of the right-hand side with the explicit formula \eqref{heat_kernel} for $\psi^{t}(0)$ shows that, as $\nu \to 0$, \eqref{exp_fraN1} converges to $\int \mu_{h^{-1}}(\dd \phi) \, \fra n(x) = (h^{-1})_{x,x}$ if $d = 1$, diverges logarithmically if $d = 2$, and diverges like $\nu^{-1/2}$ if $d = 3$.

Thus, a major shortcoming of Theorem \ref{thm:L2} is that it says nothing about particle densities for $d > 1$. To overcome this shortcoming, we tame the divergencies by considering \emph{Wick-ordered} densities. The Wick-ordered quantum particle density is by definition
\begin{equation*}
\wick{\fra N(x)} = \fra N(x) - \Tr_{\cal F} \pbb{\fra N(x) \, \frac{\ee^{-H^0}}{Z^0}}\,.
\end{equation*}
For the classical mass density we define, for each $K \in \N^*$,
\begin{equation} \label{def_n_K}
\wick{\fra n_K(x)} = \fra n_K(x) - \int \mu_{h^{-1}}(\dd \phi) \, \fra n_K(x) \,, \qquad \fra n_K(x) = \abs{\cal P_K \phi(x)}^2\,.
\end{equation}
It is not hard to see (see Lemma \ref{lem:conv_n_Wick}) that for any bounded function $f$ the random variable $\scalar{f}{\wick{\fra n_K}}$ converges as $K \to \infty$ in $L^2(\mu_{h^{-1}})$ to a random variable, which we denote by $\scalar{f}{\wick{\fra n}}$. This defines the Wick-ordered classical mass density $\wick{\fra n(x)}$, which is a priori a random distribution.

With these notations, the classical interaction constructed in Lemma \ref{lem:constr_W} (i) can be written as $W = \frac{1}{2} \int \dd x \, \dd \tilde x  \, \wick{\fra n(x)} \, v(x - \tilde x) \, \wick{\fra n(\tilde x)}$ and the quantum interaction from \eqref{H_ren} can be written as $\frac{1}{2} \int \dd x \, \dd \tilde x \, \wick{\fra N(x)} \, v(x - \tilde x) \,\wick{\fra N(\tilde x)}$.

\begin{theorem}[Convergence of Wick-ordered particle densities] \label{thm:density}
Suppose that Assumptions \ref{ass:torus} and \ref{ass:v} hold. Let $H$ be given as in \eqref{H_ren} and set $\lambda \deq \nu^2$. Suppose that $1 \leq r \leq \infty$ if $d = 1$, that $1 \leq r < \infty$ if $d = 2$, and that $1 \leq r < 3/2$ if $d = 3$. Then for all $p \in \N^*$, as $\nu \to 0$,
\begin{equation} \label{conv_Lr_density}
\Tr_{\cal F} \pbb{\nu^p \, \wick{\fra N(x_1)} \cdots \wick{\fra N(x_p)} \, \frac{\ee^{-H}}{Z}} \overset{L^r}{\longrightarrow}
\frac{1}{\zeta} \int \mu_{h^{-1}}(\dd \phi) \, \ee^{-W(\phi)} \, \wick{\fra n(x_1)} \cdots \wick{\fra n(x_p)}\,.
\end{equation}
The convergence is also locally uniform for distinct arguments $x_1, \dots, x_p$.
\end{theorem}

Again, our assumptions on the exponent $r$ are optimal, as can be checked in the free case $v = 0$ and $p = 2$.

\begin{remark}
Define the number of particles $\fra N \deq \int \dd x \, \fra N(x)$ and the total mass $\fra n \deq \int \dd x \, \fra n(x)$ of the classical field. For any $d \leq 3$, using Theorem \ref{thm:density} we obtain the convergence
\begin{equation} \label{moments_conv}
\Tr_{\cal F} \pbb{(\nu \wick{\fra N})^p \, \frac{\ee^{-H}}{Z}} \longrightarrow
\frac{1}{\zeta} \int \mu_{h^{-1}}(\dd \phi) \, \ee^{-W(\phi)} \, (\wick{\fra n})^p
\end{equation}
of the moments of the Wick-ordered number of particles to the corresponding moments of the Wick-ordered mass. We can use this result to characterize the asymptotic distribution of the particle number $\fra N$. To that end, define the probability measures $\theta_\nu, \theta$ on $\R$ as the distributions of $\nu \wick{\fra N}$ and $\wick{\fra n}$ respectively:
\begin{equation*}
\theta_\nu(B) \deq \Tr_{\cal F} \pbb{\ind{\nu :\fra N: \in B} \, \frac{\ee^{-H}}{Z}} \,,
\qquad
\theta(B) \deq
\frac{1}{\zeta} \int \mu_{h^{-1}}(\dd \phi) \, \ee^{-W(\phi)} \,\ind{: \fra n : \in B}\,.
\end{equation*}
By \eqref{moments_conv}, $\theta_\nu$ converges to $\theta$ in the sense of moments. In fact, this convergence holds also in distribution, since $\theta$ is determined by its moments (because $\theta$ has subexponential tails by Lemma \ref{lem:conv_n_Wick} below). Using Lemma \ref{Quantum Wick theorem} below, we conclude that in the state $\frac{1}{Z} \ee^{-H}$ the centred particle number $\nu \fra N - \tr \pb{\frac{\nu}{\ee^{\nu h} - 1}}$ has asymptotically distribution $\theta$. Moreover, we have $\tr \pb{\frac{\nu}{\ee^{\nu h} - 1}} \asymp T_\nu$, where
\begin{equation} \label{def_T_nu}
T_\nu \deq
\begin{cases}
1 & \text{if } d = 1 
\\
\log \nu^{-1} & \text{if } d = 2
\\
\nu^{-1/2} & \text{if } d = 3\,.
\end{cases}
\end{equation}
Since all moments of $\theta_\nu$ are asymptotically of order one for all $d$, we conclude that, in the mean-field limit $\nu \to 0$, the particle number $\fra N$ \emph{concentrates} for $d > 1$ and does not concentrate for $d = 1$. Here by concentration we mean that the standard deviation of $\fra N$ divided by its expectation tends to zero.
\end{remark}

\begin{remark}
More generally, using Theorem \ref{thm:density} we can characterize the joint distribution of the family of operators $(\scalar{f}{\nu\wick{\fra N}})_{f \in L^\infty(\Lambda)}$ in the mean-field limit: for $f_1, \dots, f_p \in L^\infty(\Lambda)$ we have
\begin{equation*}
\Tr_{\cal F} \pbb{\scalar{f_1}{\nu\wick{\fra N}} \cdots \scalar{f_p}{\nu\wick{\fra N}} \, \frac{\ee^{-H}}{Z}} \longrightarrow
\frac{1}{\zeta} \int \mu_{h^{-1}}(\dd \phi) \, \ee^{-W(\phi)} \, \scalar{f_1}{\wick{\fra n}} \cdots \scalar{f_p}{\wick{\fra n}}\,.
\end{equation*}
For example, if $D,\tilde D \subset \Lambda$ are two regions of space we can characterize the limiting behaviour of the covariance of the number of particles in $D$ and the number of particles in $\tilde D$,
\begin{multline*}
\Tr_{\cal F} \pbb{\int_{D} \dd x \, \fra N(x) \, \int_{\tilde D} \dd \tilde x \, \fra N(\tilde x) \, \frac{\ee^{-H}}{Z}}
- \Tr_{\cal F} \pbb{\int_{D} \dd x \, \fra N(x) \, \frac{\ee^{-H}}{Z}}
\Tr_{\cal F} \pbb{\int_{\tilde D} \dd \tilde x \, \fra N(\tilde x) \, \frac{\ee^{-H}}{Z}}
\\
= \frac{1 + \oo(1)}{\nu^2} \op{Cov} \pbb{\int_{D} \dd x \, \wick{\fra n(x)} \, , \, \int_{\tilde D} \dd \tilde x \, \wick{\fra n(\tilde x)}}\,,
\end{multline*}
where $\op{Cov}$ is the covariance with respect to the measure $\frac{1}{\zeta} \mu_{h^{-1}}(\dd \phi) \, \ee^{-W(\phi)}$.
\end{remark}

Both Theorems \ref{thm:L2} and \ref{thm:density} are in fact special cases of our main result, Theorem \ref{thm:main}, which states that the Wick-ordered reduced density matrices converge in $C(\R^{2p})$, the space of continuous functions equipped with the supremum norm, to the Wick-ordered classical correlation functions. To state it, we first recall the definition of Wick ordering with respect to a Gaussian measure.

Let $\mu_{\cal C}$ be a Gaussian measure with mean zero and covariance $\cal C$ on $\R^n$. We define the \emph{Wick ordering} of the polynomial $u_1 \cdots u_n$ with respect to $\mu_{\cal C}$ as
\begin{equation} \label{def_Wick}
\wick{u_1 \cdots u_n} = \frac{\partial^n}{\partial \lambda_1 \dots \partial \lambda_n}  \frac{\ee^{\lambda \cdot u}}{\int \mu_{\cal C}(\dd u) \, \ee^{\lambda \cdot u}} \bigg\vert_{\lambda = 0}\,.
\end{equation}
Writing $\int \mu_{\cal C}(\dd u) \, \ee^{\lambda \cdot u} = \ee^{\frac{1}{2} \lambda \cdot \cal C \lambda}$ and computing the derivatives, we obtain
\begin{equation} \label{wick_expanded}
\wick{u_1 \cdots u_n} = \sum_{\Pi \in \fra M([n])} \prod_{i \in [n] \setminus [\Pi]} u_i \prod_{\{i,j\} \in \Pi} \pbb{- \int \mu_{\cal C}(\dd u) \, u_i u_j}\,,
\end{equation}
where we defined $[n] \deq \{1, \dots, n\}$, and $\fra M([n])$ is the set of partial pairings of the set $[n]$ (i.e.\ a set of disjoint unordered pairs of elements of $[n]$) with $[\Pi] \deq \bigcup_{\{i,j\} \in \Pi} \{i,j\}$. We extend \eqref{wick_expanded} to complex Gaussian random variables by linearity.

In analogy to \eqref{def_gamma}, we define the \emph{Wick-ordered correlation function}
\begin{equation} \label{def_gamma_hat}
(\wh \gamma_p)_{x_1 \dots x_p, \tilde x_1 \dots \tilde x_p} \deq \frac{1}{\zeta} \int \mu_{h^{-1}}(\dd \phi) \, \ee^{-W(\phi)} \, \wick{\bar \phi(\tilde x_1) \cdots \bar \phi(\tilde x_p) \phi(x_1) \cdots \phi(x_p)}\,,
\end{equation}
where the Wick ordering is performed with respect to the measure $\mu_{h^{-1}}$. Defining the free correlation function $\gamma_p^0$ as \eqref{def_gamma} with $v = 0$, we have
\begin{equation} \label{gamma_renormalized}
\wh\gamma_p = \sum_{k = 0}^p \binom{p}{k}^2 (-1)^{p - k} \, P_p \pb{\gamma_{k} \otimes \gamma_{p-k}^0} P_p\,;
\end{equation}
see Lemma \ref{lem:wh_gamma_id} below.
For instance,
\begin{equation} \label{gamma_ex1}
\wh \gamma_1 = \gamma_1 - \gamma_1^0
\end{equation}
and
\begin{multline} \label{gamma_ex2}
(\wh \gamma_2)_{x_1 x_2, \tilde x_1 \tilde x_2} = (\gamma_2)_{x_1 x_2, \tilde x_1 \tilde x_2}
- (\gamma_1)_{x_1, \tilde x_1} (\gamma_1^0)_{x_2, \tilde x_2}
- (\gamma_1)_{x_1, \tilde x_2} (\gamma_1^0)_{x_2, \tilde x_1}
\\
- (\gamma_1)_{x_2, \tilde x_1} (\gamma_1^0)_{x_1, \tilde x_2}
- (\gamma_1)_{x_2, \tilde x_2} (\gamma_1^0)_{x_1, \tilde x_1}
+ (\gamma_2^0)_{x_1 x_2, \tilde x_1 \tilde x_2}\,.
\end{multline}

In analogy to \eqref{gamma_renormalized}, we define the \emph{Wick-ordered\footnote{This formula can be interpreted as the obvious analogue of \eqref{gamma_renormalized}. In addition, in the formal functional integral representation of the quantum many-body theory explained in Section \ref{sec:fct_int}, it is nothing but a correlation function that is Wick ordered with respect to the free field; see \eqref{formal_whGamma} below.} reduced density matrix} through
\begin{equation} \label{def_Gamma_renormalized}
\wh\Gamma_p \deq \sum_{k = 0}^p \binom{p}{k}^2 (-1)^{p - k} \, P_p \pb{ \Gamma_{k} \otimes \Gamma_{p-k}^0} P_p\,,
\end{equation}
where the free reduced density matrix $\Gamma_p^0$ is given by \eqref{def_gammap} with $v = 0$. Direct analogues of the examples \eqref{gamma_ex1} and \eqref{gamma_ex2} hold for $\wh \Gamma_1$ and $\wh \Gamma_2$.

\begin{theorem}[Convergence of Wick-ordered reduced density matrices] \label{thm:main}
Suppose that Assumptions \ref{ass:torus} and \ref{ass:v} hold. Let $H$ be given as in \eqref{H_ren} and set $\lambda \deq \nu^2$. Then for all $p \in \N^*$, as $\nu \to 0$,
\begin{equation*}
\nu^p \, \wh \Gamma_p \overset{C}{\longrightarrow} \wh \gamma_p\,,
\end{equation*}
where $\overset{C}{\longrightarrow}$ denotes convergence in the space of continuous functions with respect to the supremum norm.
\end{theorem}

\begin{remark} \label{rem:rho}
The parameter $\kappa$ corresponds to a negative chemical potential, and can be used to adjust the density of particles of the quantum system. Decreasing $\kappa$ increases the particle density. We always assume that $\kappa > 0$ for convenience, but in fact this is not important for our result and we can make $\kappa$ as negative as we wish, also depending on $\nu$. This is possible thanks to the interaction term, through which we can, at no extra cost, introduce a very large chemical potential. We refer to Appendix \ref{sec:rho} for a more detailed explanation.
\end{remark}

\subsection{Results II: trapping potentials} \label{sec:traps_intro}
In this subsection we formulate the analogue of Theorem \ref{thm:main} in infinite volume, $\Lambda = \R^d$, under the presence of a trapping potential. In particular, the translation invariance of Section \ref{sec:results1} is broken and replaced with the potential well of a trap. In addition to the small scale (ultraviolet) behaviour of the reduced density matrices, which is the main subject of this paper, here we also have to address the large scale (infrared) behaviour. We shall do so by establishing convergence in supremum norm with respect to a decaying weight function, given in \eqref{def_Upsilon} below. This will allow us to obtain convergence in $L^r$ spaces with an optimal range of exponents $r$ (see Remark \ref{rem:trap_Lr} below).

\begin{assumption} \label{ass:trap}
The domain $\Lambda = \Lambda_{\infty, d} = \R^d$ is the $d$-dimensional Euclidean space, with $d = 1,2,3$. We take a possibly $\nu$-dependent external potential $U \equiv U_\nu \col \Lambda \to \R$, which is continuous and satisfies
\begin{equation} \label{U_lower_bound}
U(x) \geq b \abs{x}^\theta
\end{equation}
for some constants $b > 0$ and $\theta \geq 2$. The one-particle Hamiltonian is
\begin{equation}
h \equiv h^U \deq \kappa - \Delta/2 + U\,,
\end{equation}
where $\kappa > 0$. We suppose that \eqref{tr_h_assump} holds for $s = 2$, uniformly in $\nu$.
\end{assumption}

By the Lieb-Thirring inequality, \cite[Theorem 1]{BFLP06}, for $h$ satisfying Assumption \ref{ass:trap} the condition \eqref{tr_h_assump} holds for $s > \frac{d}{2} + \frac{d}{\theta}$. In particular, Assumption \ref{ass:trap} holds provided that \eqref{U_lower_bound} holds with
\begin{equation} \label{theta_bound}
\theta > \frac{2d}{4-d}\,.
\end{equation}

\begin{assumption} \label{ass:U_conv}
We suppose that $U_\nu$ converges locally uniformly, as $\nu \to 0$, to a function $U_0$, and that $U_\nu \geq U_0 / C$ for all $\nu > 0$, for some fixed $C > 0$.
\end{assumption}

For instance, we may simply take $U = U_0$ not to depend on $\nu$. We include the possible $\nu$-dependence to allow the combination of our main result with the solution of the associated counterterm problem where $U$ depends on $\nu$. The counterterm problem is discussed after Theorem \ref{thm:trap} below, and the relevant results, proved in \cite{frohlich2017gibbs}, are summarized in Proposition \ref{prop:counterterm} below. In particular, Assumption \ref{ass:U_conv} is sufficiently general to be verified by the solution of the counterterm problem; see Corollary \ref{cor:trap} below.

\begin{assumption} \label{ass:v_trap}
The interaction potential $v : \Lambda \to \R$ is even, of positive type, uniformly continuous, and bounded.
\end{assumption}

We construct the classical field exactly as in Section \ref{sec:classical_field}. The free field measure $\mu_{h^{-1}}$ is constructed in terms of $h^{U_0}$ with the limiting external potential $U_0$. Under Assumptions \ref{ass:trap} and \ref{ass:v_trap}, it is easy to check that the Wick-ordered interaction \eqref{def_W_K} has a limit $W$, and we define the partition function and the correlation functions as in \eqref{def_z} and \eqref{def_gamma}. For $c > 0$ and $\theta > 0$ we define the weight function
\begin{equation} \label{def_Upsilon}
\Upsilon_{\theta,c}(x, \tilde x) \deq (1 + \abs{x} + \abs{\tilde x})^{-\theta (2 - d/2)} \, \ee^{-c \abs{x - \tilde x}}\,,
\end{equation}
and for $\f x, \tilde {\f x} \in \Lambda^p$ we define the symmetrized weight
\begin{equation} \label{def_Upsilon_p}
\Upsilon_{\theta,c}(\f x, \tilde {\f x}) \deq \sum_{\pi \in S_p} \prod_{i = 1}^p \Upsilon_{\theta,c}(x_i, \tilde x_{\pi(i)})\,,
\end{equation}
which we shall use to control the decay of Wick-ordered correlation functions.

\begin{theorem} \label{thm:trap}
Suppose that Assumptions \ref{ass:trap}, \ref{ass:U_conv}, and \ref{ass:v_trap} hold. Let $H$ be as in \eqref{H_ren} and set $\lambda = \nu^2$. Then $Z/Z^0 \to \zeta$ as $\nu \to 0$. Moreover, for any $p \in \N^*$ there exists a constant $c > 0$ such that for any $\epsilon > 0$ we have
\begin{equation} \label{Gamma_hat_conv_trap}
\lim_{\nu \to 0} \sup_{\f x, \tilde {\f x}} \frac{\abs{\nu^p (\wh \Gamma_p)_{\f x, \tilde {\f x}} - (\wh \gamma_p)_{\f x, \tilde {\f x}}}}{\Upsilon_{\theta - \epsilon,c}(\f x, \tilde {\f x})} = 0\,.
\end{equation}
\end{theorem}

\begin{remark} \label{rem:trap_Lr}
Suppose that \eqref{theta_bound} holds. An immediate consequence of Theorem \ref{thm:trap} and the fact that $\Upsilon_{\theta - \epsilon,c} \in L^1(\R^{2dp})$ for small enough $\epsilon > 0$ is that $\nu^p \wh \Gamma_p$ converges to $\wh \gamma_p$ in $L^1 \cap L^\infty$. Here the space $L^1$ controls the infrared (large scale) behaviour and $L^\infty$ the ultraviolet (small scale) behaviour.

Moreover, if \eqref{theta_bound} holds then the analogues of Theorems \ref{thm:L2} and \ref{thm:density} hold under the assumptions of Theorem \ref{thm:trap}. Note that this means convergence in $L^1 \cap L^r$.
This follows easily from Theorem \ref{thm:trap} as well as the observation that $\nu \Gamma_1^0$ converges to $\gamma^0_1$ in $L^1 \cap L^r$ with $r$ satisfying the assumptions of Theorem \ref{thm:L2}, as can be seen from a simple estimate using Lemma \ref{lem:prop_est} below.
\end{remark}

\begin{remark} \label{rem:decay_hat_gamma}
From Theorem \ref{thm:trap}, more precisely from Proposition \ref{prop:trap_bound} below, we find that $\abs{\nu^p (\wh \Gamma_{p})_{\f x, \tilde {\f x}}} \leq C \Upsilon_{\theta,c}(\f x, \tilde {\f x})$. Recalling the form \eqref{def_Upsilon}, \eqref{def_Upsilon_p} and comparing to standard estimates on the decay of the free density matrix $\nu^p \Gamma^0_p$, this implies that for $d > 1$ the Wick-ordered reduced density matrix $\nu^p \wh \Gamma_p$ decays much faster than the free density matrix $\nu^p \Gamma^0_p$ for large arguments $\f x, \tilde {\f x}$. In other words, $\nu^p \wh \Gamma_p$ has a much better infrared behaviour than $\nu^p \Gamma^0_p$. This observation complements the observation, already made in Theorem \ref{thm:main}, that $\nu^p \wh \Gamma_p$ has much more regular ultraviolet behaviour than $\nu^p \Gamma^0_p$. These estimates are all uniform in $\nu$ and hence hold also for $\wh \gamma_p$ and $\gamma^0_p$. A simple instance of this phenomenon is that, for $d >1$, the condition \eqref{theta_bound} which ensures that $\gamma_1^0$ is in $L^2$ implies that $\wh \gamma_1$ is even in $L^1$ (as follows from \eqref{def_Upsilon} and Proposition \ref{prop:trap_bound} below).
\end{remark}

We conclude this subsection with a discussion on the \emph{counterterm problem}, which relates the Wick-ordered Hamiltonian \eqref{H_ren} to the Hamiltonian without Wick ordering \eqref{H_sq}. Under the assumptions of Theorem \ref{thm:trap}, the Hamiltonian \eqref{H_ren} reads
\begin{equation} \label{Hc_ren}
H = \nu \int \dd x \, a^*(x) \pbb{\kappa - \frac{\Delta}{2} + U(x)} a(x) 
+ \frac{1}{2} \int \dd x \, \dd \tilde x  \, \p{\nu a^*(x) a(x) - \varrho^U(x)}  \, v(x - \tilde x) \, \p{\nu a^*(\tilde x) a(\tilde x) - \varrho^U(\tilde x)}\,,
\end{equation}
where we can write the free quantum density $\varrho^U(x)$ from \eqref{def_rho} using the Wick theorem (see Lemma \ref{Quantum Wick theorem} below) as
\begin{equation*}
\varrho^U(x) = \pbb{\frac{\nu}{\exp(\nu h^U) - 1}}_{x,x}\,.
\end{equation*}
We compare the Wick-ordered Hamiltonian \eqref{Hc_ren} to the Hamiltonian without Wick ordering
\begin{equation} \label{Hc_phys}
\tilde H = \nu \int \dd x \, a^*(x) \pbb{\tilde \kappa - \frac{\Delta}{2} + V(x)} a(x) 
+ \frac{\nu^2}{2} \int \dd x \, \dd \tilde x  \,  a^*(x) a(x)  \, v(x - \tilde x)\, a^*(\tilde x) a(\tilde x)\,,
\end{equation}
where $\tilde \kappa \in \R$ and $V \col \Lambda \to \R$ is a \emph{bare} external potential. A simple calculation shows that $H$ and $\tilde H$ differ by an irrelevant constant $H - \tilde H = \frac{1}{2} \int \dd x \, \dd \tilde x \, \varrho^U(x) v(x - \tilde x) \varrho^U(\tilde x)$ under the conditions
\begin{subequations}
\begin{align} \label{kappa_relation}
\tilde \kappa &= \kappa - \hat v(0) \varrho^0
\\ \label{counterterm-problem}
V &= U - v * (\varrho^U - \varrho^0)\,.
\end{align}
\end{subequations}
Here
\begin{equation*}
\varrho^0 = \pbb{\frac{\nu}{\exp(\nu h^0) - 1}}_{0,0}
\end{equation*}
is the (homogeneous) free particle density in the absence of the external potential, and $\wh v(0) = \int \dd x \, v(x)$. It is easy to verify the asymptotic behaviour
\begin{equation} \label{rho0_behaviour}
\varrho^0  \asymp T_\nu\,,
\end{equation}
where $T_\nu$ was defined in \eqref{def_T_nu}.

The equation \eqref{counterterm-problem} is the \emph{counterterm problem}, which relates the bare external potential $V$ to the dressed, or renormalized, external potential $U$. It is a nonlinear integral equation for the potential $U$, which is solved in \cite[Section 5]{frohlich2017gibbs}. We quote the relevant result.

\begin{proposition}[Theorem 5.2 and Remark 5.3 of \cite{frohlich2017gibbs}] \label{prop:counterterm}
Let $V \col \Lambda \to \R$ be a locally bounded nonnegative function satisfying $V(x) \geq c \abs{x}^\theta$ for some constants $c > 0$ and $\theta$ satisfying \eqref{theta_bound}. Suppose moreover that there exists a constant $C> 0$ such that
\begin{equation*}
V(x+y) \leq C (1 + V(x)) (1 + V(y))
\end{equation*}
as well as
\begin{equation*}
\int \dd x \, v(x) \, (1 + V(x)^2) < \infty\,.
\end{equation*}
Then the following holds for large enough $\kappa > 0$. 
\begin{enumerate}[label=(\roman*)]
\item
The counterterm problem \eqref{counterterm-problem} has a unique solution $U_\nu$ for each $\nu > 0$.
\item
The equation
\begin{equation} \label{ct_limit}
V(x) = U_0(x) - \int \dd \tilde x \, v(x - \tilde x) \pbb{\frac{1}{\kappa - \Delta/2 + U_0} - \frac{1}{\kappa - \Delta/2}}_{\tilde x, \tilde x}
\end{equation}
has a unique solution $U_0$, and we have
\begin{equation*}
\lim_{\nu \to 0} \norm{(U_\nu - U_0)/V}_{L^\infty} = 0\,.
\end{equation*}
\item
There is a constant $C > 0$ such that $V / C \leq U_\nu \leq C V$ for all $\nu \geq 0$.
\end{enumerate}
\end{proposition}

The assumptions of Proposition \ref{prop:counterterm} are for instance satisfied for the external potential $V(x) = \abs{x}^\theta$, with $\theta$ satisfying \eqref{theta_bound}, and $\abs{v(x)} \leq C (1 + \abs{x})^{-2 \theta -d  - \delta}$ for some $C,\delta > 0$. In this case, it is easy to check that Assumptions \ref{ass:trap} and \ref{ass:U_conv} hold for $U_\nu$ constructed in Proposition \ref{prop:counterterm}, so that Theorem \ref{thm:trap} can be applied to the grand canonical Gibbs state defined by the physical Hamiltonian \eqref{Hc_phys}. More precisely, we have the following result.

\begin{corollary} \label{cor:trap}
Suppose that $\Lambda = \R^d$ for $d = 1,2,3$ and suppose that $V(x) = \abs{x}^\theta$, with $\theta$ satisfying \eqref{theta_bound}, and $\abs{v(x)} \leq C (1 + \abs{x})^{-2 \theta -d  - \delta}$ for some $C,\delta > 0$. Let $U_0$ be the unique solution of the limiting counterterm problem \eqref{ct_limit} associated with $V$. Fix $\kappa > 0$ and let $\tilde \kappa$ be given by \eqref{kappa_relation}. Let $\wh \Gamma_p$ be the Wick-ordered reduced density matrix \eqref{def_Gamma_renormalized}, \eqref{def_gammap} associated with the Hamiltonian \eqref{Hc_phys} with external potential $V$, and let $\wh \gamma_p$ be the Wick-ordered correlation function \eqref{def_gamma_hat}, \eqref{def_gamma} associated with the external potential $U_0$. Then \eqref{Gamma_hat_conv_trap} holds for some $c > 0$ and any $\epsilon > 0$.
\end{corollary}

We remark that, by \eqref{rho0_behaviour}, the relation \eqref{kappa_relation} means that $\tilde \kappa \to -\infty$ as $\nu \to 0$ for $d >1$. The physical interpretation is that, in order to obtain a well-defined field limit of the quantum many-body  system, the chemical potential of the Hamiltonian \eqref{Hc_phys} has to be chosen very large to compensate the very large repulsive interaction energy of the particles.

\subsection{The functional integral and strategy of proof} \label{sec:fct_int}
We conclude the introduction by explaining the main ideas behind our proof. Our strategy is based on a functional integral representation of quantum many-body theory. It is particularly transparent if one is willing to forgo mathematical rigour, as we shall do in this subsection.  In fact, an important contribution of our paper is to put the following formal discussion on a rigorous footing.

Let us first consider the many-body Hamiltonian \eqref{H_sq} without Wick ordering.
Let $\Phi \col [0,\nu] \times \Lambda \to \C$ be a field, i.e.\ a function of the time variable $\tau \in [0,\nu]$ and the spatial variable $x \in \Lambda$. We define the \emph{free action}
\begin{equation*}
\cal S^0(\Phi) \deq \int_0^\nu \dd \tau \int_\Lambda \dd x \, \bar \Phi(\tau,x) \pb{\partial_\tau + \kappa - \Delta/2} \Phi(\tau,x)\,,
\end{equation*}
which is complex valued (because $\partial_\tau$ is not self-adjoint) with a positive real part. Moreover, we define the \emph{interaction}
\begin{equation} \label{def_W_quantum}
\cal W(\Phi) \deq \frac{\lambda}{2 \nu} \int_0^\nu \dd \tau \int_\Lambda \dd x \, \dd \tilde x \, \abs{\Phi(\tau,x)}^2 \, v(x - \tilde x) \, \abs{\Phi(\tau, \tilde x)}^2\,,
\end{equation}
and write $\cal S(\Phi) \deq \cal S^0(\Phi) + \cal W(\Phi)$ for the action of the interacting field theory.
Then the statistical mechanics of the grand canonical ensemble $\ee^{-H}/Z$ associated with the quantum many-body problem can be expressed in terms of the field theory
\begin{equation*}
\frac{1}{C} \, \ee^{-\cal S(\Phi)} \, \mathrm D \Phi\,,
\end{equation*}
where $\mathrm D \Phi = \prod_{\tau \in [0,\nu]} \prod_{x \in \Lambda} \dd \Phi(\tau, x)$ is the (nonexisting) uniform measure over the space of fields, and $C$ is an (infinite) normalization constant. More precisely, the relative partition function has the representation
\begin{equation} \label{formal_Z}
\cal Z \deq \frac{Z}{Z^0} = \frac{\int \mathrm D \Phi \, \ee^{-\cal S(\Phi)}}{\int \mathrm D \Phi \, \ee^{-\cal S^0(\Phi)}}\,,
\end{equation}
where the partition function $Z$ was defined in \eqref{def_Z} and $Z^0$ is the free partition function. Moreover, the reduced density matrices defined in \eqref{def_gammap} are given as correlation functions of the field $\Phi$ at time zero:
\begin{equation} \label{formal_Gamma}
(\Gamma_p)_{x_1 \dots x_p, \tilde x_1 \dots \tilde x_p} = \frac{1}{\int \mathrm D \Phi \, \ee^{-\cal S(\Phi)}} \int \mathrm D \Phi \, \ee^{-\cal S(\Phi)} \, 
\bar \Phi(0,\tilde x_1) \cdots \bar \Phi(0, \tilde x_p) \Phi(0, x_1) \cdots \Phi(0, x_p)\,.
\end{equation}
The formulas \eqref{formal_Z} and \eqref{formal_Gamma} are the functional integral representation underlying our proofs. We remark that in this functional integral formulation, the Wick-ordered reduced density matrix from \eqref{def_Gamma_renormalized} can be written as
\begin{equation} \label{formal_whGamma}
(\wh \Gamma_p)_{x_1 \dots x_p, \tilde x_1 \dots \tilde x_p} = \frac{1}{\int \mathrm D \Phi \, \ee^{-\cal S(\Phi)}} \int \mathrm D \Phi \, \ee^{-\cal S(\Phi)} \, 
\wick{\bar \Phi(0,\tilde x_1) \cdots \bar \Phi(0, \tilde x_p) \Phi(0, x_1) \cdots \Phi(0, x_p)}\,,
\end{equation}
where the Wick ordering is performed with respect to the free field $\ee^{-\cal S^0(\Phi)} \mathrm D\Phi$. The representations \eqref{formal_Gamma} and \eqref{formal_whGamma} are the quantum functional integral counterparts of the classical field expressions \eqref{def_gamma} and \eqref{def_gamma_hat}.

Analogously, the classical field theory from Section \ref{sec:classical_field}, here without Wick ordering, can be written formally as
\begin{equation*}
\frac{1}{c} \, \ee^{-s(\phi)} \, \mathrm D \phi\,,
\end{equation*}
where $\phi \col \Lambda \to \C$ is the classical field which depends on the spatial variable $x \in \Lambda$ only, $\mathrm D \phi = \prod_{x \in \Lambda} \dd \phi(x)$ is the (nonexisting) uniform measure on the space of classical fields, $c$ is an (infinite) normalization constant, and $s(\phi) \deq s^0(\phi) + w(\phi)$ is the classical action, with the classical free action
\begin{equation*}
s^0(\phi) \deq \int_\Lambda \dd x \, \bar \phi(x) (\kappa -\Delta/2) \phi(x)
\end{equation*}
and the classical interaction
\begin{equation*}
w(\phi) \deq \frac{1}{2} \int_\Lambda \dd x \, \dd \tilde x \, \abs{\phi(x)}^2 \, v(x - \tilde x) \, \abs{\phi(\tilde x)}^2\,.
\end{equation*}
Analogously to the identities \eqref{formal_Z} and \eqref{formal_Gamma}, we can write the (relative) partition function \eqref{def_z} as
\begin{equation} \label{formal_z}
\zeta = \frac{\int \mathrm D \phi \, \ee^{-s(\phi)}}{\int \mathrm D \phi \, \ee^{-s^0(\phi)}}
\end{equation}
and the correlation functions \eqref{def_gamma}
\begin{equation} \label{formal_gamma}
(\gamma_p)_{x_1 \dots x_p, \tilde x_1 \dots \tilde x_p} = \frac{1}{\int \mathrm D \phi \, \ee^{-s(\phi)}} \int \mathrm D \phi \, \ee^{-s(\phi)} \, \bar \phi(\tilde x_1) \cdots \bar \phi(\tilde x_p) \phi(x_1) \cdots \phi(x_p)\,.
\end{equation}

To analyse the mean-field limit $\lambda = \nu^2 \to 0$ of \eqref{formal_Z} and \eqref{formal_Gamma}, it is convenient to introduce the rescaled field $\Phi'(t,x) \deq \sqrt{\nu} \Phi(\nu t,x)$, with $t\in [0,1]$. Thus we have $\cal S(\Phi) = \cal S'(\Phi')$, where
\begin{equation*}
\cal S'(\Phi') \deq \int_0^1 \dd t \int_\Lambda \dd x \, \bar \Phi'(t,x) \pb{\partial_t / \nu+ \kappa - \Delta/2} \Phi'(t,x)
+ \frac{1}{2} \int_0^1 \dd t \int_\Lambda \dd x \, \dd \tilde x \, \abs{\Phi'(t,x)}^2 \, v(x - \tilde x) \, \abs{\Phi'(t, \tilde x)}^2\,.
\end{equation*}
Moreover, \eqref{formal_Gamma} becomes
\begin{equation*}
\nu^p (\Gamma_p)_{x_1 \dots x_p, \tilde x_1 \dots \tilde x_p}
= \frac{1}{\int \mathrm D \Phi' \, \ee^{-\cal S'(\Phi')}} \int \mathrm D \Phi' \, \ee^{-\cal S'(\Phi') } \, 
\bar \Phi'(0,\tilde x_1) \cdots \bar \Phi'(0, \tilde x_p) \Phi'(0, x_1) \cdots \Phi'(0, x_p)\,.
\end{equation*}
In terms of the field $\Phi'$, we see that the $\nu$-dependence appears only in front of the time derivative $\partial_t$ in the free part of $\cal S'$. Hence, as $\nu \to 0$, the time-dependence of $\Phi'$ is suppressed by the strong oscillations in the factor $\ee^{-\cal S'(\Phi')}$, so that, by a stationary phase argument, we expect the dominant contribution in the integral over $\Phi'$ to arise from fields that are constant in time. This is an indication that
the expressions \eqref{formal_Z} and $\nu^p \cdot$\eqref{formal_Gamma} will converge to \eqref{formal_z} and \eqref{formal_gamma}, respectively. This gives an appealing heuristic for the emergence of the mean-field limit from the quantum many-body problem. Thus, our proof can be regarded as a rigorous implementation of an infinite-dimensional stationary phase argument for such ill-defined functional integrals.

The discussion above was performed for the bare theory without Wick ordering of the interaction. We implement the Wick ordering in the quantum many-body problem by replacing the right-hand side of \eqref{def_W_quantum} with
\begin{equation} \label{Wick_formal}
\frac{\lambda}{2 \nu} \int_0^\nu \dd \tau \int_\Lambda \dd x \, \dd \tilde x \, \pb{\abs{\Phi(\tau,x)}^2 - \varrho(x)} \, v(x - \tilde x) \, \pb{\abs{\Phi(\tau, \tilde x)}^2 - \varrho(\tilde x)}\,, \qquad \varrho(x) =
\frac{\int \mathrm D \Phi \, \ee^{-\cal S^0(\Phi)} \,  \abs{\Phi(0,x)}^2}{\int \mathrm D \Phi \, \ee^{-\cal S^0(\Phi)}} \,,
\end{equation}
and an analogous procedure for the classical field $\phi$, which was already carefully explained in Section \ref{sec:classical_field}.

Next, we outline how we make the above formal construction rigorous and establish convergence to the mean-field limit. We start by noting that already the Gaussian measure associated with the free theory presents a major obstacle. In the classical case, the formal Gaussian measure $\ee^{-s^0(\phi)} \mathrm D \phi$ can be easily constructed using standard arguments, and we gave such a construction in Section \ref{sec:classical_field}. For this construction, it is crucial that the covariance $(\kappa - \Delta/2)^{-1}$ is self-adjoint. In contrast, in the quantum case, any attempt to construct the free Gaussian measure $\ee^{-\cal S^0(\Phi)} \mathrm D\Phi$ with covariance $(\partial_\tau  + \kappa - \Delta /2)^{-1}$ is doomed to fail. This is because the covariance is not self-adjoint, although it is a well-defined normal operator with strictly positive real part. In finite dimensions, it is easy to construct explicitly a Gaussian measure with a non-self-adjoint covariance, as long as the real part of the covariance is positive. In infinite dimensions, such a measure does in general not exist, and this is in particular the case for $\ee^{-\cal S^0(\Phi)} \mathrm D\Phi$. Indeed, as explained in \cite{balaban2008functional}, a ``no-go'' theorem from \cite{cameron1962ilstow} shows that $\ee^{-\cal S^0(\Phi)} \mathrm D\Phi$ \emph{cannot} lead to a well-defined complex measure on the space of fields. The problem is that the imaginary part of the exponent $\cal S^0(\Phi)$ is unbounded and it gives rise to uncontrollable oscillations.

We remark that the formal integral representation \eqref{formal_Z}, \eqref{formal_Gamma} is also used as the starting point of a major and ongoing programme \cite{balaban2008functional, balaban2008functional2} (see also \cite{froh_houches,Moshe_Zinn-Justin}) with the goal of establishing Bose-Einstein condensation for an interacting Bose gas. The approach of \cite{balaban2008functional, balaban2008functional2} to make \eqref{formal_Z}, \eqref{formal_Gamma} rigorous is a coherent state functional integral and a discretization of the time direction $\tau \in [0,\nu]$. As we now explain, in this paper we take very different approach.

Our solution to the construction of the measure $\ee^{-\cal S^0(\Phi)} \mathrm D \Phi$ is not to attempt to find an actual measure (which, as explained above, does not exist), but to define it as a linear functional on a sufficiently large class of functions of the field $\Phi$. Let us first explain our construction for the relative partition function without Wick ordering, \eqref{formal_Z}. Our starting point is a Hubbard-Stratonovich transformation using an auxiliary real field $\sigma : [0,\nu] \times \Lambda \to \R$, with law $\mu_{\cal C}$, which is centred and has covariance
\begin{equation*}
\int \mu_{\cal C}(\dd \sigma) \, \sigma(\tau,x) \, \sigma(\tilde \tau, \tilde x) = \frac{\lambda}{\nu}\,\delta(\tau - \tilde \tau) \, v(x - \tilde x)\,.
\end{equation*}
In practice, as such a field is white noise in the time direction, we need to introduce a regularization into our covariance which ensures that $\sigma$ is almost surely continuous (see \eqref{Fourier_regularization2} and \eqref{cov_C_eta} below). By the Hubbard-Stratonovich transformation (see \eqref{Hubbard_Stratonovich_formula} below) we obtain from \eqref{formal_Z} that
\begin{equation} \label{DPhi_formal_int}
\cal Z = \frac{\int \mathrm D \Phi \, \ee^{-\cal S(\Phi)}}{\int \mathrm D \Phi \, \ee^{-\cal S^0(\Phi)}}
= \int \mu_{\cal C}(\dd \sigma) \, \frac{\int \mathrm D \Phi \, \exp\pb{- \scalarb{\Phi}{(\partial_\tau -  \Delta/2 + \kappa - \ii \sigma) \Phi}}}{\int \mathrm D \Phi \, \exp\pb{- \scalarb{\Phi}{(\partial_\tau - \Delta/2 + \kappa) \Phi}}}\,.
\end{equation}
For any function $u \col [0,\nu] \times \Lambda \to \C$ whose real part is always negative, we introduce the operator $K(u) \deq \partial_\tau  - \Delta/2 - u$. Thus, we obtain
\begin{equation} \label{calZ_formal}
\cal Z = \int \mu_{\cal C}(\dd \sigma) \frac{\int \mathrm D \Phi \, \exp\pb{- \scalarb{\Phi}{K(-\kappa + \ii \sigma) \Phi}}}{\int \mathrm D \Phi \, \exp\pb{-\scalarb{\Phi}{K(-\kappa) \Phi}}} = \int \mu_{\cal C}(\dd \sigma) \frac{\det K(-\kappa + \ii \sigma)^{-1}}{\det K(-\kappa)^{-1}}
=
\int \mu_{\cal C}(\dd \sigma) \, \ee^{F_1(\sigma)}\,,
\end{equation}
where we defined
\begin{equation}
F_1(\sigma) \deq - \Tr \pB{\log K(-\kappa + \ii \sigma) - \log K(-\kappa)} = \int_0^\infty \dd t \, \Tr \pbb{\frac{1}{t + K(-\kappa + \ii \sigma)}-  \frac{1}{t + K(-\kappa)}}\,,
\end{equation}
and the last step follows by an integral representation of the logarithm (see \eqref{integral_identity} below).

Next, we observe that for any function $u : [0,\nu] \times \Lambda \to \C$ the Green function $(K(u)^{-1})^{\tau, \tilde \tau}_{x, \tilde x}$ of $K(u)$ can be expressed in terms of the propagator $W^{\tau, \tilde \tau}(u)$ of a heat flow driven by a periodic time-dependent potential $u([\tau]_\nu)$, where $[\tau]_\nu$ is the $\nu$-periodic representative of $\tau \in \R$ in $[0,\nu]$. More precisely, we have the relation
\begin{equation} \label{K_inv_intro}
(K(u)^{-1})^{\tau, \tilde \tau}_{x, \tilde x} = \sum_{r \in \nu \N} \ind{\tau  + r > \tilde \tau}  \, W^{\tau + r, \tilde \tau}_{x,\tilde x}(u)\,,
\end{equation}
where
\begin{equation} \label{def_W_intro}
\partial_\tau W^{\tau, \tilde \tau}(u) = \pbb{\frac{1}{2} \Delta + u([\tau]_\nu)} W^{\tau, \tilde \tau}(u) \,, \qquad W^{\tau, \tau}(u) = 1\,.
\end{equation}
Using the Feynman-Kac formula, we represent the propagator $W^{\tau, \tilde \tau}(u)$ as
\begin{equation} \label{FK_intro}
W^{\tau, \tilde \tau}_{x, \tilde x}(u) = \int  \bb W^{\tau, \tilde \tau}_{x, \tilde x}(\dd \omega)  \, \ee^{\int_{\tilde \tau}^\tau \dd t \, u([t]_\nu, \omega(t))}\,,
\end{equation}
where $\bb W^{\tau, \tilde \tau}_{x, \tilde x}$ is the usual (unnormalized) path measure of the Brownian bridge from the space-time point $(\tilde \tau, \tilde x)$ to $(\tau, x)$. Putting all of these ingredients together, we then show that the right-hand side of \eqref{calZ_formal} is well defined and equals precisely $Z/Z^0$ with $Z$ is given by \eqref{def_Z} and $Z^0$ is its free version.

Thus, our rigorous construction of the formal functional integrals in \eqref{formal_Z} is summarized by the relations \eqref{DPhi_formal_int}--\eqref{FK_intro}, along with an appropriate regularization of the measure $\mu_{\cal C}(\dd \sigma)$. This allows us to represent the left-hand side of \eqref{DPhi_formal_int} rigorously in terms of an integral over the field $\sigma$ and Brownian loops $\omega$.

In the above formal discussion, the Wick ordering from \eqref{Wick_formal} is very easy to implement: we simply introduce a phase $\exp \pb{-\frac{\ii}{\nu}\int_0^\nu \dd \tau \int \dd x \, \sigma(\tau, x) \, \varrho(x)}$
in the integral on the right-hand side of \eqref{DPhi_formal_int}, and find that this gives rise to the correctly Wick-ordered interaction on the left-hand side. After introducing this phase, a calculation (see Lemma \ref{lem:quantum_density} below) shows that the representation \eqref{calZ_formal} becomes, upon Wick ordering,
\begin{equation*}
\cal Z 
=
\int \mu_{\cal C}(\dd \sigma) \, \ee^{F_2(\sigma)}\,, \qquad F_2(\sigma) \deq F_1(\sigma) - \int_0^\infty \dd t \, \Tr \pbb{\frac{1}{t + K(-\kappa)} \, \ii \sigma \, \frac{1}{t + K(-\kappa)}}\,.
\end{equation*}
Here we observe the regularizing effect of Wick ordering: by a simple resolvent expansion we notice a strong cancellation between the two terms of $F_2(\sigma)$, which will ensure the boundedness of $F_2(\sigma)$ as $\nu \to 0$ in all dimensions $d = 1,2,3$, unlike $F_1(\sigma)$, which is well behaved only for $d = 1$.

Next, using the above representation we can write (for the non-Wick-ordered interaction for simplicity, as above Wick-ordering the interaction amounts to replacing $F_1$ by $F_2$) the reduced density matrices \eqref{formal_Gamma} as
\begin{equation} \label{Gamma_intro_formal}
\Gamma_p = \frac{1}{\int \mu_{\cal C}(\dd \sigma) \, \ee^{F_1(\sigma)}} \, p! P_p \int \mu_{\cal C}(\dd \sigma)\, \ee^{F_1(\sigma)} \, \pB{\pb{K(-\kappa + \ii \sigma)^{-1}}^{0,0}}^{\otimes p}\,.
\end{equation}
To see why \eqref{Gamma_intro_formal} indeed reproduces \eqref{formal_Gamma}, we apply a Hubbard-Stratonovich transformation to \eqref{formal_Gamma}, which yields
\begin{multline*}
(\Gamma_p)_{x_1 \dots x_p, \tilde x_1 \dots \tilde x_p} = \frac{1}{\cal Z}\, \frac{1}{\int \mathrm D \Phi \, \exp\pb{-\scalarb{\Phi}{K(-\kappa) \Phi}}}
\\
\times \int \mu_{\cal C}(\dd \sigma) \int \mathrm D \Phi \, \exp\pb{- \scalarb{\Phi}{K(-\kappa + \ii \sigma) \Phi}} \, \bar \Phi(0,\tilde x_1) \cdots \bar \Phi(0, \tilde x_p) \Phi(0, x_1) \cdots \Phi(0, x_p)\,.
\end{multline*}
Now applying Wick's rule for the Gaussian measure $\exp\pb{- \scalarb{\Phi}{K(-\kappa + \ii \sigma) \Phi}} \, \mathrm D \Phi$ yields the kernel of \eqref{Gamma_intro_formal}. In order to make a rigorous link between \eqref{Gamma_intro_formal} and the actual definition \eqref{def_gammap}, we use the Wick theorem for quasi-free bosonic states (see Lemma \ref{Quantum Wick theorem} below). We remark that Wick-ordering of the reduced density matrices, as in \eqref{def_Gamma_renormalized}, is particularly transparent in the form \eqref{Gamma_intro_formal}:
\begin{equation} \label{wh_Gamma_intro}
\wh \Gamma_p = \frac{1}{\int \mu_{\cal C}(\dd \sigma) \, \ee^{F_1(\sigma)}} \, p! P_p \int \mu_{\cal C}(\dd \sigma)\, \ee^{F_1(\sigma)} \, \pB{\pb{K(-\kappa + \ii \sigma)^{-1}}^{0,0} - \pb{K(-\kappa)^{-1}}^{0,0}}^{\otimes p}\,,
\end{equation}
as can be seen by expanding the product over the $p$ operators and applying \eqref{Gamma_intro_formal} to each resulting term.

Having derived a rigorous version of the functional integral representation \eqref{formal_Z}, \eqref{formal_Gamma}, we derive an analogous representation of the classical field theory. The starting point is again the Hubbard-Stratonovich transformation, except that this time the auxiliary field $\xi$ depends on spatial variables only, and has covariance $\int \mu_{v}(\dd \xi) \, \xi(x) \, \xi(\tilde x) = v(x - \tilde x)$. A simple but important observation is that the time-averaged field $\ang{\sigma}(x) = \frac{1}{\nu} \int_0^\nu \dd \tau \, \sigma(\tau,x)$ has the same law as $\xi$. As it turns out, the expressions we obtain for the classical field theory mirror those of the quantum theory described above, except that the time-dependence of the field $\sigma$ is suppressed through the time averaging $\ang{\sigma}$. Similarly, the time-dependent propagator $W^{\tau, \tilde \tau}(-\kappa + \ii \sigma)$ is replaced with the time-homogeneous propagator $\ee^{-(\tau - \tilde \tau) (\kappa - \Delta / 2 - \ii \xi)}$, and sums of the form \eqref{K_inv_intro}, with $u = -\kappa + \ii \sigma$ and after multiplication by $\nu$, are given by the corresponding Riemann integrals
\begin{equation} \label{Riemann_integral}
\int_0^\infty \dd r \, \ee^{-r (\kappa-  \Delta/2 - \ii \xi)} = \frac{1}{\kappa - \Delta / 2 - \ii \xi}\,.
\end{equation}

Once the rigorous functional integral representation in terms of the auxiliary field $\sigma$ and the Brownian paths $\omega$ (see \eqref{FK_intro}) is set up, the proof of convergence to the mean-field limit entails establishing convergence of Riemann sums of the kind \eqref{K_inv_intro} to integrals of the form \eqref{Riemann_integral}, as well as quantitative continuity properties of the propagator $W^{\tau, \tilde \tau}$ in the time variables $\tau, \tilde \tau$. These estimates represent the main analytical work of our proof. All of these estimates are seriously complicated by the fact that the field $\sigma$ is singular, white noise in time after removal of the regularization in the measure $\mu_{\cal C}$.

In all of these estimates, the representation in terms of Brownian loops proves very useful. Indeed, we can use quantitative continuity properties of Brownian motion to control the convergence to the mean-field limit. Moreover, when considering particles in Euclidean space confined by an external potential (as in Section \ref{sec:traps_intro}), the infrared (i.e.\ long-range) properties of the Wick ordered correlation functions can be very effectively estimated using basic excursion estimates for Brownian bridges.

\subsection{Outline of the paper}

The rest of this paper is devoted to the proofs our our main results, Theorems \ref{thm:L2}, \ref{thm:density}, \ref{thm:main}, and \ref{thm:trap}. The main argument is the proof of Theorem \ref{thm:main}, which is given in Sections \ref{sec:preliminaries}--\ref{sec:proofs_conclusion}, in which we work on the torus and make Assumptions \ref{ass:torus} and \ref{ass:v}. In Section \ref{sec:preliminaries} we collect a variety of standard tools and notations that we use throughout the proof. In Section \ref{sec:fct_quantum} we derive the functional integral representation for the quantum many-body system, and in Section \ref{sec:classical} we do the same for the classical field theory. In Section \ref{sec:mf}, we prove the convergence of the quantum many-body theory to the classical theory in the mean-field limit. In Section \ref{sec:proofs_conclusion}, we conclude the proof of Theorem \ref{thm:main}, and use it to deduce Theorems \ref{thm:L2} and \ref{thm:density} as corollaries. Finally, in Section \ref{sec:traps} we explain how to extend the analysis of Sections \ref{sec:fct_quantum}--\ref{sec:proofs_conclusion} to the case of an external trapping potential in Euclidean space, and prove Theorem \ref{thm:trap}.

\section{Preliminaries} \label{sec:preliminaries}
In this section we collect various tools and notations that we shall use throughout the paper.

\subsection{Basic notations}
We use $1$ to denote the identity operator on a Hilbert space, and we sometimes write $1/\ops$ for the inverse of an operator $\ops$. We use $\ind{A}$ to denote the indicator function of a set $A$. We use the notation $\mu_{\cal C}$ for a Gaussian measure with covariance $\cal C$. We use the letters $C,c >0$ to denote generic positive constants, which can change from one expression to the next ($C$ should be thought of as being large enough and $c$ small enough). If $C$ depends on some parameter $\alpha$, we indicate this by writing $C \equiv C_\alpha$. For vectors $\f r \in [0,\infty)^n$ with nonnegative entries we use the notation $\abs{\f r} \deq \sum_{i = 1}^n r_i$. For points $x \in \Lambda$ we use $\abs{x}$ to denote the Euclidean norm on $\Lambda$.

For a separable Hilbert space $\cal H$ and $p \in [1,\infty]$, the Schatten space $\fra S^p(\cal H)$ is the set of bounded operators $\ops$ on $\cal H$ satisfying $\|\ops\|_{\fra S^p} < \infty$, where
\begin{equation*}
\|\ops\|_{\fra S^p} \deq
\begin{cases}
( \tr \, \abs{\ops}^p)^{1/p}  &\mbox{if }p<\infty\\
\sup \spec \, \abs{\ops} &\mbox{if } p=\infty\,,
\end{cases}
\end{equation*}
and $\abs{\ops} \deq \sqrt{\ops^* \ops}$. With these notations, $\norm{\ops}_{\fra S^2} = \norm{\ops}_{L^2}$, and $\fra S^1(\cal H)$ is the space of trace class operators.

\subsection{Brownian paths}

Let $0 \leq \tilde \tau < \tau$ and denote by $\Omega^{\tau, \tilde \tau}$ the space of continuous paths $\omega : [\tilde \tau, \tau] \to \Lambda$.
For $\tilde x \in \Lambda$ and $0 \leq \tilde \tau < \tau$, let $\P^{\tau, \tilde \tau}_{\tilde x}$ denote the law on $\Omega^{\tau, \tilde \tau}$ of standard Brownian motion equal to $\tilde x$ at time $\tilde \tau$, with periodic boundary conditions on $\Lambda$. For $\tilde x,x \in \Lambda$ and $0 \leq \tilde \tau < \tau$, let $\P^{\tau, \tilde \tau}_{x, \tilde x}$ denote the law on $\Omega^{\tau, \tilde \tau}$ of the Brownian bridge equal to $\tilde x$ at time $\tilde \tau$ and equal to $x$ and time $\tau$, with periodic boundary conditions on $\Lambda$.

We can characterize the measures $\P_{\tilde x}^{\tau, \tilde \tau}$ and $\P^{\tau, \tilde \tau}_{x, \tilde x}$ explicitly through their finite-dimensional distributions. To that end, let
\begin{equation} \label{heat_kernel}
\psi^t(x) \deq (\ee^{\Delta t/2})_{x,0} = \sum_{n \in \Z^d} (2 \pi t)^{-d/2} \ee^{-\abs{x - Ln}^2 / 2t}
\end{equation}
be the periodic heat kernel on $\Lambda \equiv \Lambda_{L,d}$. (We allow $L = \infty$ for the case $\Lambda = \R^d$.) Let $n \in \N^*$, $\tilde \tau < t_1 < \cdots < t_n < \tau$, and $f : \Lambda^n \to \R$ be a continuous function. Then the law $\P^{\tau, \tilde \tau}_{\tilde x}$ is characterized through its finite-dimensional distribution
\begin{multline*}
\int \bb P^{\tau, \tilde \tau}_{\tilde x}(\dd \omega)\, f(\omega(t_1), \dots, \omega(t_n))
\\
= \int \dd x_1 \cdots \dd x_n \, \psi^{t_1 - \tilde \tau}(x_1 - \tilde x) \psi^{t_2 - t_1} (x_2 - x_1) \cdots \psi^{t_n - t_{n - 1}} (x_n - x_{n-1}) \, f(x_1, \dots, x_n)\,.
\end{multline*}
Similarly, the positive measure
\begin{equation} \label{def_W}
\bb W^{\tau, \tilde \tau}_{x,\tilde x}(\dd \omega) \deq \psi^{\tau - \tilde \tau}(x - \tilde x) \, \bb P^{\tau, \tilde \tau}_{x,\tilde x}(\dd \omega)
\end{equation}
is characterized by its finite-dimensional distribution
\begin{multline} \label{W_dist}
\int \bb W^{\tau, \tilde \tau}_{x,\tilde x}(\dd \omega)\, f(\omega(t_1), \dots, \omega(t_n))
\\
= \int \dd x_1 \cdots \dd x_n \, \psi^{t_1 - \tilde \tau}(x_1 - \tilde x) \psi^{t_2 - t_1} (x_2 - x_1) \cdots \psi^{t_n - t_{n - 1}} (x_n - x_{n-1}) \psi^{\tau - t_n} (x - x_n) \, f(x_1, \dots, x_n)\,.
\end{multline}
Hence,
\begin{equation} \label{W_P_relation}
\bb W^{\tau, \tilde \tau}_{x,\tilde x}(\dd \omega) =
\P^{\tau, \tilde \tau}_{\tilde x}(\dd \omega) \, \delta\pb{\omega(\tau) - x} \,.
\end{equation}

The following result is well known.

\begin{lemma}[Feynman-Kac] \label{FK_continuous}
Let $I \subset \R$ be a closed bounded interval and $V \col I \times \Lambda \to \C$ be continuous and bounded from below. Let $(W^{\tau, \tilde \tau})_{\tilde \tau \leq \tau \in I}$ be the propagator satisfying
\begin{equation*}
\partial_\tau W^{\tau, \tilde \tau} = \pbb{\frac{1}{2} \Delta - V(\tau)} W^{\tau, \tilde \tau} \,, \qquad W^{\tau, \tau} = 1\,.
\end{equation*}
Then $W^{\tau, \tilde \tau}$ has the operator kernel
\begin{equation*}
W^{\tau,\tilde \tau}_{x,\tilde x} = \int \bb W^{\tau, \tilde \tau}_{x,\tilde x} (\dd \omega) \, \ee^{-\int_{\tilde \tau}^\tau \dd s \, V(s, \omega(s))}\,.
\end{equation*}
\end{lemma}

We recall the following elementary heat kernel estimate.
\begin{lemma} \label{lem:heat_estimate}
There is a constant $C$ such that
\begin{equation*}
\sup_{x, \tilde x} \int \bb W^{\tau,\tilde \tau}_{x, \tilde x}(\dd \omega) \leq C \pb{L^{-d} + (\tau - \tilde \tau)^{-d/2}}\,.
\end{equation*}
\end{lemma}
\begin{proof}
The left-hand side is equal to $\sup_x \psi^{\tau - \tilde \tau}(x)$, and the claim is an easy consequence of \eqref{heat_kernel} and a Riemann sum approximation.
\end{proof}

We shall also need the following quantitative $L^2$-continuity result for $\P^{\tau, \tilde \tau}_{x, \tilde x}$. Denote by $\abs{x}_L \deq \min_{n \in \Z^d} \abs{x - L n}$ the periodic Euclidean norm of $x \in \Lambda$.

\begin{lemma} \label{lem:P_cont}
There exists a constant $C > 0$ such that
\[ \int \mathbb{P}_{x, \tilde{x}}^{\tau, \tilde{\tau}} \, ({\rm d} \omega)  |\omega (t) - \omega (s)|_L^2 \leq C \qbb{(t-s) + \abs{x - \tilde x}_L^2 \frac{(t - s)^2}{(\tau - \tilde \tau)^2}}\]
for any $\tilde{\tau} \leq s \leq t \leq \tau$.
\end{lemma} 

The proof of Lemma \ref{lem:P_cont} is given in Appendix \ref{Proof_of_lem:P_cont}.

\subsection{Gaussian integration}
We record the following generalization of Wick's rule for a real Gaussian measure. 

\begin{lemma} \label{lem:Gaussian}
Let $\cal C > 0$ be a positive real $n \times n$ matrix. Define the Gaussian probability measure on $\R^n$ with covariance $\cal C$ through
\begin{equation*}
\mu_{\cal C}(\dd u) \deq \frac{1}{\sqrt{(2 \pi)^{n} \det \cal C }} \, \ee^{-\frac{1}{2} \scalar{u}{\cal C^{-1} u}} \, \dd u\,.
\end{equation*}
For any $k$ and $f, f_1, \dots, f_k \in \R^n$ we have
\begin{equation*}
\int \mu_{\cal C}(\dd u) \prod_{i = 1}^k \scalar{f_i}{u} \, \ee^{\ii \scalar{f}{u}} = \ee^{-\frac{1}{2} \scalar{f}{\cal C f}} \sum_{\Pi \in \fra M([k])} \prod_{\{i,j\} \in \Pi} \scalar{f_i}{\cal C f_j} \prod_{i \in [k] \setminus [\Pi]} \ii \scalar{f_i}{\cal C f}\,,
\end{equation*}
where we recall the notations introduced after \eqref{wick_expanded}.
\end{lemma}

\begin{proof}
By completing the square, we find
\begin{equation*}
\int \mu_{\cal C}(\dd u) \prod_{i = 1}^k \scalar{f_i}{u} \, \ee^{\ii \scalar{f}{u}}
= \frac{1}{\sqrt{(2 \pi)^{n} \det \cal C }}
\int \dd u \, \ee^{-\frac{1}{2} \scalar{u - \ii \cal C f}{\cal C^{-1} (u - \ii \cal C f)}} \ee^{-\frac{1}{2} \scalar{f}{\cal C f}} \prod_{i = 1}^k \scalar{f_i}{u}\,.
\end{equation*}
By a change of variables $u \mapsto u + \ii \cal C f$ and using Cauchy's theorem to deform the integration path in $\R^n$, we obtain
\begin{equation*}
\int \mu_{\cal C}(\dd u) \prod_{i = 1}^k \scalar{f_i}{u} \, \ee^{\ii \scalar{f}{u}} =
\ee^{-\frac{1}{2} \scalar{f}{\cal C f}} \int \mu_\cal C(\dd u) \, \prod_{i = 1}^k \scalar{f_i}{u + \ii \cal C f}\,.
\end{equation*}
Using Wick's rule for $\mu_{\cal C}$, we therefore obtain
\begin{align*}
\int \mu_{\cal C}(\dd u) \prod_{i = 1}^k \scalar{f_i}{u} \, \ee^{\ii \scalar{f}{u}} &= \ee^{-\frac{1}{2} \scalar{f}{\cal C f}} \sum_{I \subset [k] \text{ even}} \int \mu_{\cal C}(\dd u)\prod_{i \in I} \scalar{f_i}{u} \prod_{i \in [k] \setminus I} \ii \scalar{f_i}{\cal C f}
\\
&=
\ee^{-\frac{1}{2} \scalar{f}{\cal C f}} \sum_{I \subset [k] \text{ even}} \sum_{\Pi \in \fra M_c(I)} \prod_{\{i,j\} \in \Pi} \scalar{f_i}{\cal C f_j} \prod_{i \in [k] \setminus I} \ii \scalar{f_i}{\cal C f}\,,
\end{align*}
where $\fra M_c(I)$ denotes the set of complete pairings of the set $I$. The claim now follows.
\end{proof}

Next, we record some basic results on complex Gaussians. We use the notation $\scalar{z}{w} = \sum_i \ol z_i w_i$ for the complex inner product and $\dd z$ for the Lebesgue measure on $\C^n$. For a complex $n \times n$ matrix $\cal C$ we define $\re \cal C \deq \frac{1}{2}(\cal C + \cal C^*)$.

\begin{lemma}
\label{complex_Gaussian}
Let $\cal C$ be a complex $n \times n$ matrix with $\re \cal C>0$. Then
\begin{equation*}
\int_{\C^n} \dd z \, \ee^{- \scalar{z}{\cal C^{-1} z}} = \pi^n \det \cal C\,.
\end{equation*}
\end{lemma}
In light of Lemma \ref{complex_Gaussian}, for a complex $n \times n$ matrix $\cal C$ with $\re \cal C >0$, we define the Gaussian probability measure on $\C^n$ with covariance $\cal C$ through 
\begin{equation}
\label{mu_A_complex}
\mu_{\cal C}(\dd z) \deq \frac{1}{\pi^{n} \det \cal C } \, \ee^{-\scalar{z}{\cal C^{-1} z}} \, \dd z\,.
\end{equation}
We state Wick's rule for a complex Gaussian measure.

\begin{lemma}
\label{complex_Wick_theorem}
Let $\cal C$ be a complex $n \times n$ matrix with $\re \cal C >0$ and let $\mu_{\cal C}$ be given by \eqref{mu_A_complex}. For $f_1, \dots, f_p, g_1, \dots, g_p \in \C^n$ we have
\begin{equation*}
\int_{\mathbb{C}^n} \mu_{\cal C}(\dd z) \, \prod_{i = 1}^p \scalar{f_i}{z} \scalar{z}{g_i} 
= \sum_{\pi \in S_p} \prod_{i=1}^{p} \scalar{f_i}{\cal C g_{\pi(i)}}\,.
\end{equation*} 
\end{lemma}
\begin{proof}
Using the notation $\bar \partial_i = \frac{\partial}{\partial \bar z_i}$, the claim follows from the identity
\begin{equation*}
\scalar{f}{z} \,\ee^{-\scalar{z}{\cal C^{-1} z}} = - \scalar{f}{\cal C \bar \partial} \, \ee^{- \scalar{z}{\cal C^{-1} z}}
\end{equation*}
and integration by parts.
\end{proof}

\subsection{Fourier transform} \label{sec:FT}
Here we summarize the conventions of the Fourier transform that we shall use.
For a function $f \col \R^d \to \C$ we use the continuous Fourier transform
\begin{equation*}
f(x) = \int_{\R^d} \dd p \, (\cal F f)(p) \, \ee^{2 \pi \ii x \cdot p} \,, \qquad
(\cal F f)(p) = \int_{\R^d} \dd x \, f(x) \, \ee^{- 2 \pi \ii x \cdot p}\,.
\end{equation*}
For a function $f \col \Lambda_{L,d} \to \C$ we use the Fourier series
\begin{equation*}
f(x) = \frac{1}{L^d} \sum_{p \in \Z^d} (\cal F_L f)(p) \, \ee^{2 \pi \ii p \cdot x / L}\,, \qquad (\cal F_L f)(p) = \int_{\Lambda} \dd x \, f(x) \, \ee^{- 2 \pi \ii x \cdot p / L}\,.
\end{equation*}
On $\Lambda$ we define the convolution
\begin{equation*}
(f * g)(x) \deq \int_\Lambda \dd y \, f(x - y) \, g(y)\,, \qquad (\cal F_L (f * g))(p) = (\cal F_L f)(p) (\cal F_L g)(p)\,.
\end{equation*}
For a function $v \col \Lambda \to \C$ and a function $\varphi \col \R^d \to \C$ we have for any $\eta > 0$, by Poisson summation,
\begin{equation} \label{Fourier_regularization}
\frac{1}{L^d} \sum_{p \in \Z^d} \varphi(\eta p) (\cal F_L v)(p) \, \ee^{2 \pi \ii x \cdot p / L} = (\delta_{\eta, L, \varphi} * v)(x)\,, \qquad \delta_{\eta, L, \varphi}(x) \deq \frac{1}{\eta^d} \sum_{y \in \Z^d} (\cal F^{-1} \varphi) \pbb{\frac{x - L y}{\eta}}\,.
\end{equation}
The function $\delta$ has the interpretation of an approximate delta function on $\Lambda$. More precisely, if $\varphi(0) = 1$ then $\delta_{\eta, L, \varphi}$ is a periodic function on $\Lambda$ satisfying $\int_\Lambda \dd x \, \delta_{\eta, L, \varphi}(x) = 1$, which converges (in distribution, i.e.\ tested against continuous functions) as $\eta \to 0$ to the periodic delta function on $\Lambda$.

\subsection{Properties of quasi-free states} \label{sec:qf}
In this subsection we review some standard facts about bosonic quasi-free states.

\begin{definition}
We lift any bounded operator $\ops$ on $\cal H$ to an operator $\Gamma(\ops)$ on Fock space $\cal F(\cal H)$ through
$\Gamma(\ops) \deq \bigoplus_{n \in \N} \ops^{\otimes n}$ on $\cal F(\cal H)$.
\end{definition}

\begin{lemma} \label{lem:trace_formulas}
Let $\ops$ be a trace-class operator on $\cal H$ with spectral radius strictly less than one. Then
\begin{equation} \label{trace_bosons}
\Tr_{\cal F} (\Gamma(\ops)) = \sum_{n \in \N} \tr (P_n \ops^{\otimes n}) = \det (1 - \ops)^{-1} = \ee^{- \tr \log (1 - \ops)}\,.
\end{equation}
\end{lemma}

\begin{proof}
By a standard approximation argument, it suffices to establish the claim for finite-dimensional $\cal H$. Moreover, by density of diagonalizable matrices, we may assume that $\ops$ is diagonalizable. Hence, by cyclicity of trace, it suffices to show the claim for diagonal $\ops = \diag(\ops_1, \dots, \ops_k)$, where $k = \dim \cal H$.

We use the occupation state basis of the Fock space $\cal F(\cal H)$. For $\f m \in \N^k$ we define the vector $E_{\f m} \deq P_{\abs{\f m}} \, \f e_1^{\otimes m_1} \otimes \f e_2^{\otimes m_2} \otimes \cdots \otimes \f e_k^{\otimes m_k}$, where $\abs{\f m} \deq \sum_{i = 1}^k m_i$ and $(\f e_i)_{i = 1}^k$ is the standard basis of $\cal H$. It is easy to check that $(E_{\f m})_{\f m \in \N^k}$ is an orthonormal basis of $\cal F(\cal H)$. Thus we find
\begin{equation*}
\Tr_{\cal F}(\Gamma(\ops)) = \sum_{\f m \in \N^k} \scalar{E_{\f m}}{\ops^{\otimes \abs{\f m}} E_{\f m}} = \sum_{\f m \in \N^k} \ops_1^{m_1} \cdots \ops_k^{m_k} = \prod_{i = 1}^k \frac{1}{1 - \ops_i} = \det(1 - \ops)^{-1}\,.
\qedhere
\end{equation*}
\end{proof}

Let $\ops$ be a trace-class operator on $\cal H$ with spectral radius strictly less than one. We define the \emph{quasi-free state} associated with $\ops$ through
\begin{equation} \label{quasi_free}
q_\ops (X) \deq \frac{\Tr_{\cal F} (X \Gamma(\ops))}{ \Tr_{\cal F} (\Gamma(\ops))}\,,
\end{equation}
where $X$ is an operator on the Fock space $\cal F(\cal H)$. The following result is well known; for a proof see e.g.\ \cite[Lemma B.1]{frohlich2017gibbs}. (Note that the argument of \cite[Lemma B.1]{frohlich2017gibbs} was given for the case that $\ops$ is a positive operator, but this assumption is irrelevant for the proof.)

\begin{lemma}[Wick theorem]
\label{Quantum Wick theorem}
Let $\ops$ be a trace-class operator on $\cal H$ with spectral radius strictly less than one.
Then the following holds.
\begin{enumerate}
\item
We have
\begin{equation}
\label{QWT_i_1}
q_\ops \pb{a^*(g) \, a(f)} = \scalarbb{f}{\frac{\ops}{1 - \ops}\, g}
\end{equation}
and
\begin{equation}
\label{QWT_i_2}
q_\ops\pb{a(f) \, a(g)} = q_\ops\pb{a^*(f) \, a^*(g)} = 0
\end{equation}
for all $f,g \in \cal H$.
\item
Let $X_1, \dots, X_n$ be operators of the form $X_i = a(f_i)$ or $X_i = a^*(f_i)$, where $f_1, \dots, f_n \in \cal H$. Then we have
\begin{equation}
\label{QWT_ii}
q_\ops (X_1 \cdots X_n) = \sum_{\Pi \in \fra M_c([n])} \prod_{(i,j) \in \Pi} q_\ops(X_i X_j)\,,
\end{equation}
where the sum ranges over all complete pairings of $[n]$, and we label the edges of $\Pi$ using ordered pairs $(i,j)$ with $i < j$.
\end{enumerate}
\end{lemma}

\section{Functional integral representation of quantum many-body systems} \label{sec:fct_quantum}

In this section we derive the functional integral representation for the partition function $Z = \Tr_{\cal F}(\ee^{-H})$ and the reduced density matrices \eqref{def_gammap}, under Assumptions \ref{ass:torus} and \ref{ass:v}. The main results of this section are Proposition \ref{prop:fif_renorm} for the partition function as well as Propositions \ref{prop:red_dens_matrix} and \ref{prop:gamma_hat} for the reduced density matrices.

For clarity of the exposition, in Subsections \ref{intro_sigma}--\ref{sec:space-time}, we focus on an interaction that is not Wick-ordered, i.e.\ we consider a Hamiltonian of the form \eqref{Hamiltonian_H_n}. Then, in Section \ref{sec:Wick}, we explain how Wick ordering is incorporated into our functional integral representation.

\subsection{Random field representation} \label{intro_sigma}
In this subsection we suppose that $H_n$ is of the form \eqref{Hamiltonian_H_n}.
By the Feynman-Kac formula, Lemma \ref{FK_continuous}, applied to dimension $dn$ instead of $d$, we obtain
\begin{equation} \label{e_H_n}
(\ee^{- H_n})_{\f x, \tilde{\f x}} = \ee^{-\nu \kappa n} \int \prod_{i = 1}^n \bb W^{\nu,0}_{x_i, \tilde x_i}(\dd \omega_i) \, \ee^{-\frac{\lambda}{2 \nu} \sum_{i,j = 1}^n \int_0^\nu \dd t \, v(\omega_i(t) - \omega_j(t))}\,.
\end{equation}
The starting point of our representation is a Hubbard-Stratonovich transformation to the exponential in \eqref{e_H_n}. Formally, we introduce a real Gaussian field $\sigma$ on the space-time torus $[0,\nu] \times \Lambda$ with law $\mu_{\cal C_0}$, whose covariance is given by
\begin{equation*}
\int \mu_{\cal C_0}(\dd \sigma) \, \sigma(\tau, x) \, \sigma(\tilde \tau, \tilde x) = \frac{\lambda}{\nu} \, \delta(\tau - \tilde \tau) v(x - \tilde x) \eqd  (\cal C_0)_{x, \tilde x}^{\tau, \tilde \tau} \,.
\end{equation*}
Then we use that for a real Gaussian measure $\mu_{\cal C}$ with covariance $\cal C$
\begin{equation}
\label{Hubbard_Stratonovich_formula}
\int \mu_{\cal C} (\dd \sigma) \, \ee^{\ii \langle f, \sigma \rangle}
= e^{-\frac{1}{2} \langle f , \cal C f \rangle}\,;
\end{equation}
see Lemma \ref{lem:Gaussian}. Using \eqref{Hubbard_Stratonovich_formula} with $f(\tau, x) = \sum_{i = 1}^n \delta(x - \omega_i(\tau))$ we obtain
\begin{equation}
(\ee^{- H_n})_{\f x, \tilde{\f x}} = \ee^{-\nu \kappa n}\int \mu_{\cal C_0}(\dd \sigma) \int \prod_{i = 1}^n \bb W^{\nu,0}_{x_i, \tilde x_i}(\dd \omega_i)  \prod_{i = 1}^n \ee^{\ii \int_0^\nu \dd t \, \sigma(t, \omega_i(t))}\,.
\end{equation}

For the validity of many of the following arguments, we shall need that $\sigma$ is a continuous function of $\tau$ and $x$. Unfortunately, under $\mu_{\cal C_0}$, the field $\sigma$ is white noise in the time direction (and therefore singular in $t$), and potentially discontinuous even in $x$ if $v$ is not smooth enough. We remedy this by regularizing the covariance $\cal C_0$. Let $\varphi$, respectively $\tilde \varphi$, be a smooth, nonnegative, even, compactly supported function on $\R^d$ satisfying $\varphi(0) = 1$, respectively on $\R$ satisfying $\tilde \varphi(0) = 1$. Moreover, we choose\footnote{This condition can be easily achieved by taking the convolution of an appropriate nonnegative function with itself.} $\varphi$ such that $\cal F^{-1} \varphi \geq 0$. 
Recalling the definition \eqref{Fourier_regularization}, we abbreviate
\begin{equation} \label{Fourier_regularization2}
\delta_{\eta, \nu} \deq \delta_{\eta, \nu, \tilde \varphi}\,, \qquad v_\eta \deq v * \delta_{\eta,L,\varphi}\,.
\end{equation}

For $\eta > 0$ we define the real Gaussian measure $\mu_{\cal C_\eta}$ with mean zero and covariance
\begin{equation} \label{cov_C_eta}
\int \mu_{\cal C_\eta}(\dd \sigma)\, \sigma(\tau, x) \, \sigma(\tilde \tau, \tilde x) = \frac{\lambda}{\nu} \, \delta_{\eta, \nu}(\tau - \tilde \tau) \, v_\eta (x - \tilde x) \eqd (\cal C_\eta)_{x, \tilde x}^{\tau, \tilde \tau}\,.
\end{equation}
For $\eta \geq 0$, we can represent $\mu_{\cal C_\eta}$ as the law of the Gaussian field
\begin{equation} \label{sigma_X_rep}
\sigma_\eta(\tau,x) \deq \sqrt{\frac{\lambda}{\nu}} \frac{1}{\sqrt{\nu L^d}} \sum_{k \in \Z} \sum_{p \in \Z^d} \frac{X_{k,p} + \ol X_{-k,-p}}{\sqrt{2}} \, \sqrt{(\cal F_L v)(p) \, \tilde \varphi(\eta k)\, \varphi(\eta p)} \, \ee^{2 \pi \ii \tau k / \nu + 2 \pi \ii p \cdot x / L} \,,
\end{equation}
where $(X_{k,p} \col k\in \Z, p \in \Z^d)$ is family of independent standard complex Gaussians. For $\eta = 0$, the series converges almost surely in a Sobolev space of sufficiently negative index. For $\eta > 0$, the series is a finite sum, and therefore, under the law $\mu_{\cal C_\eta}$, the field $\sigma$ is almost surely a smooth periodic function on $[0,\nu] \times \Lambda$, and the integration over $\mu_{\cal C_\eta}$ is a finite-dimensional Gaussian integral. 

We define the following regularized versions of $\ee^{-H}$ and $Z = \Tr_{\cal F} (\ee^{-H})$. 
\begin{definition} \label{def:R}
Let $\eta > 0$. On $\cal F$ we define the operator $\wt R_\eta \deq \bigoplus_{n \in \N} \wt R_{n, \eta}$, where $\wt R_{n, \eta}$ has operator kernel
\begin{equation} \label{def_R_kernel}
(\wt R_{n, \eta})_{\f x, \tilde {\f x}} \deq \ee^{-\nu \kappa n} \int \mu_{\cal C_\eta}(\dd \sigma)  \int \prod_{i = 1}^n \bb W^{\nu,0}_{x_i, \tilde x_i}(\dd \omega_i) \, \prod_{i = 1}^n \ee^{\ii \int_0^\nu \dd t \, \sigma(t, \omega_i(t))}\,.
\end{equation}
Moreover, define $\wt Z_\eta \deq \Tr_{\cal F} (\wt R_\eta)$.
\end{definition}

As a guide to the reader, we note that the tilde on $\tilde R_{n,\eta}$ and $\tilde Z_\eta$ indicates the absence of Wick-ordering. In Section \ref{sec:Wick} we define the corresponding Wick-ordered quantities, which are denoted by $R_{n,\eta}$ and $Z_\eta$.

Throughout the following we shall use simple pointwise bounds of the form
\begin{equation} \label{kernel_est}
\abs{(\wt R_{n, \eta})_{\f x, \tilde {\f x}}} \leq  \ee^{-\nu \kappa n} \int \prod_{i = 1}^n \bb W^{\nu,0}_{x_i, \tilde x_i}(\dd \omega_i) = (\ee^{-H_n^0})_{\f x, \tilde {\f x}}\,.
\end{equation}
For the actual proof, we shall not need the following result, but we give it to illustrate the argument used to prove convergence as $\eta \to 0$, which we shall use repeatedly later on.
\begin{lemma} \label{lem:lim_eta}
We have $\lim_{\eta \to 0} \wt Z_\eta = Z$.
\end{lemma}

\begin{proof}
We first claim that for all $n \in \N$ and all $\f x, \tilde {\f x} \in \Lambda^n$ we have the pointwise\footnote{Similarly to Proposition \ref{prop:fif_renorm}, this convergence also holds in $L^\infty$, which we however neither need nor prove here.} convergence
\begin{equation} \label{R_eta_pointwise}
\lim_{\eta \to 0} (\wt R_{n, \eta})_{\f x, \tilde {\f x}} = (\ee^{-H_n})_{\f x, \tilde {\f x}}\,.
\end{equation}
Using Fubini's theorem and \eqref{Hubbard_Stratonovich_formula} with $f(\tau,x)=\sum_{i=1}^{n} \delta (x-\omega_i(t))$ for the finite-dimensional Gaussian integral over $\mu_{\cal C_\eta}$ in \eqref{def_R_kernel}, we obtain
\begin{equation*}
(\wt R_{n, \eta})_{\f x, \tilde {\f x}} = \ee^{-\nu \kappa n} \int \prod_{i = 1}^n \bb W^{\nu,0}_{x_i, \tilde x_i}(\dd \omega_i)
\exp \pbb{-\frac{\lambda}{2 \nu} \sum_{i,j = 1}^n \int_0^\nu \dd t \int_0^\nu \dd s \, \delta_{\eta, \nu}(t - s) \, v_\eta (\omega_i(t) - \omega_j(s))}\,.
\end{equation*}
Moreover, from \eqref{def_R_kernel} we find that the exponent has nonpositive real part. Since the Brownian paths $\omega_1, \dots, \omega_n$ are almost surely continuous, using dominated convergence and recalling \eqref{e_H_n}, we find that to prove \eqref{R_eta_pointwise} it suffices to show, for any pair of continuous paths $\omega, \tilde \omega : [0,\nu] \to \Lambda$, that
\begin{equation} \label{eta_conv_omega}
\lim_{\eta \to 0} \int_0^\nu \dd t \int_0^\nu \dd s \, \delta_{\eta, \nu}(t - s) \, v_\eta (\omega(t) - \tilde \omega(s)) = \int_0^\nu \dd t \, v(\omega(t) - \tilde \omega(t))\,.
\end{equation}
To prove this, write
\begin{multline*}
\int_0^\nu \dd t \int_0^\nu \dd s \, \delta_{\eta, \nu}(t - s) \, v_\eta (\omega(t) - \tilde \omega(s))
- \int_0^\nu \dd t \int_0^\nu \dd s \, \delta(t - s) \, v (\omega(t) - \tilde \omega(s))
\\
=
\int_0^\nu \dd t \int_0^\nu \dd s \, \delta_{\eta, \nu}(t - s) \, \pb{v_\eta - v}(\omega(t) - \tilde \omega(s))
+ \int_0^\nu \dd t \int_0^\nu \dd s \, (\delta_{\eta, \nu} - \delta) (t - s) \, v(\omega(t) - \tilde \omega(s))\,.
\end{multline*}
The first term converges to zero as $\eta \to 0$ because $\lim_{\eta \to 0} \norm{v_\eta - v}_{L^\infty} = 0$ by continuity of $v$, and the second converges to zero as $\eta \to 0$ because $(t,s) \mapsto v(\omega(t) - \tilde \omega(s))$ is a continuous map on $[0,\nu]^2$. This concludes the proof of \eqref{R_eta_pointwise}.

Now write
\begin{equation*}
Z = \Tr_{\cal F}(\ee^{-H}) =  \sum_{n \in \N} \tr (P_n \ee^{-H_n}) = \sum_{n \in \N} \frac{1}{n!} \sum_{\pi \in S_{n}} \int_{\Lambda^n} \dd \f u \, (\ee^{-H_n})_{\pi \f u, \f u}\,,
\end{equation*}
where $\pi \f u = (u_{\pi(1)} \dots u_{\pi(n)})$. Analogously,
\begin{equation*}
\wt Z_\eta = \Tr_{\cal F} (\wt R_\eta) = \sum_{n \in \N} \tr (P_n \wt R_{n,\eta}) = \sum_{n \in \N} \frac{1}{n!} \sum_{\pi \in S_{n}} \int_{\Lambda^n} \dd \f u \, (\wt R_{n, \eta})_{\pi \f u, \f u}\,.
\end{equation*}
Then the claim follows from \eqref{R_eta_pointwise} using \eqref{kernel_est},
\begin{equation} \label{sum_free_Z}
\sum_{n \in \N} \frac{1}{n!} \sum_{\pi \in S_{n}} \int_{\Lambda^n} \dd \f u \, (\ee^{-H_n^0})_{\pi \f u, \f u} = Z^0 < \infty\,,
\end{equation}
and dominated convergence.
\end{proof}

\subsection{The periodic propagator} \label{sec:propagator}
As in the previous subsection, we assume that $H_n$ is of the form \eqref{Hamiltonian_H_n}.
We now introduce a fundamental object of our representation -- a propagator of a heat flow driven by a periodic time-dependent potential.
\begin{definition} \label{def:Gamma_u}
For $t \in \R$ abbreviate $[t]_\nu \deq (t \, \op{mod}\, \nu) \in [0,\nu)$. For any continuous periodic function $u \col [0,\nu] \times \Lambda \to \C$, we define the propagator $(W^{\tau, \tilde \tau}(u))_{\tilde \tau \leq \tau}$ through
\begin{equation*}
\partial_\tau W^{\tau, \tilde \tau}(u) = \pbb{\frac{1}{2} \Delta + u([\tau]_\nu)} W^{\tau, \tilde \tau}(u) \,, \qquad W^{\tau, \tau}(u) = 1\,,
\end{equation*}
where we regard $u([\tau]_\nu)$ as a multiplication operator.
\end{definition}

By the Feynman-Kac formula, Lemma \ref{FK_continuous}, $W^{\tau, \tilde \tau}(u)$ has kernel
\begin{equation} \label{W_FK}
W^{\tau, \tilde \tau}_{x, \tilde x}(u) = \int  \bb W^{\tau, \tilde \tau}_{x, \tilde x}(\dd \omega)  \, \ee^{\int_{\tilde \tau}^\tau \dd t \, u([t]_\nu, \omega(t))}\,.
\end{equation}

\begin{lemma} \label{lem:Gamma}
The propagator $(W^{\tau, \tilde \tau})_{\tilde \tau \leq \tau}$ satisfies the following properties.
\begin{enumerate}
\item
It is a groupoid in the sense that for $\tau_1 \leq \tau_2 \leq \tau_3$ we have
\begin{equation*}
W^{\tau_3, \tau_2}(u) \, W^{\tau_2, \tau_1}(u) = W^{\tau_3, \tau_1}(u)\,.
\end{equation*}
\item
It is $\nu$-periodic in the sense that for $\tilde \tau \leq \tau$ we have
\begin{equation*}
W^{\tau + \nu, \tilde \tau + \nu}(u) = W^{\tau, \tilde \tau}(u)\,.
\end{equation*}
\item
Suppose that $\re u \leq - c$ for some constant $c > 0$. Then $\norm{W^{\tau, \tilde \tau}(u)}_{\fra S^\infty} \leq \ee^{- c(\tau - \tilde \tau)}$, and moreover for any $\tau > \tilde \tau$, $W^{\tau, \tilde \tau}(u)$ is trace class.
\end{enumerate}
\end{lemma}

\begin{proof}
The properties (i) and (ii) are obvious from Definition \ref{def:Gamma_u} and \eqref{W_FK}. For property (iii), we use the Schur test for the operator norm and \eqref{W_FK} to obtain $\norm{W^{\tau, \tilde \tau}(u)}_{\fra S^\infty} \leq \ee^{- c(\tau - \tilde \tau)}$. Finally, for any $\tau > \tilde \tau$, $W^{\tau, \tilde \tau}(u)$ is trace class by \eqref{W_FK} and Lemma \ref{lem:heat_estimate}.
\end{proof}

Using Definition \ref{def:Gamma_u}, we can express $\wt R_\eta$ and $\wt Z_\eta$ from Definition \ref{def:R} as follows.
\begin{proposition} \label{prop:ZR_eta}
Let $H$ be defined as in \eqref{H_sq} and suppose that Assumptions \ref{ass:torus} and \ref{ass:v} hold.
\begin{enumerate}[label=(\roman*)]
\item
We have
\begin{equation} \label{E_eta_Gamma}
\wt R_\eta = \int \mu_{\cal C_\eta}(\dd \sigma) \, \Gamma \pb{W^{\nu,0}(-\kappa + \ii \sigma)}\,.
\end{equation}
\item
We have
\begin{equation} \label{wt_Z_eta}
\wt Z_\eta =  \int \mu_{\cal C_\eta}(\dd \sigma) \, \ee^{F_0(\sigma)}\,,
\end{equation}
where
\begin{equation} \label{Z_int_rep}
F_0(\sigma) \deq - \tr \log (1 - W^{\nu,0}(-\kappa + \ii \sigma))\,.
\end{equation}
\end{enumerate}
\end{proposition}

\begin{proof}
From Definition \ref{def:R} we find $(\wt R_{n, \eta})_{\f x, \tilde {\f x}} = \int \mu_{\cal C_\eta}(\dd \sigma) \prod_{i = 1}^n W^{\nu,0}_{x_i, \tilde x_i}(-\kappa + \ii \sigma)$, and (i) follows.
Moreover, thanks to Lemma \ref{lem:Gamma} (iii), we may apply Lemma \ref{lem:trace_formulas} to obtain (ii).
\end{proof}

Next, expanding the logarithm in \eqref{Z_int_rep} as a power series, using Lemma \ref{lem:Gamma} (iii) and $\kappa > 0$, we get
\begin{equation} \label{Z_log_expanded}
F_0(\sigma) = \sum_{\ell \in \N^*} \frac{1}{\ell} \tr W^{\ell \nu,0}(-\kappa + \ii \sigma)\,,
\end{equation}
where we used Lemma \ref{lem:Gamma} (i) and (ii). We have the following result.
\begin{lemma} \label{lem:F_1}
We have
\begin{equation} \label{cal_Z_F_1}
\wt Z_\eta = Z^0 \int \mu_{\cal C_\eta}(\dd \sigma) \, \ee^{F_1(\sigma)}\,,
\end{equation}
where
\begin{equation} \label{def_F_1}
F_1(\sigma) \deq F_0(\sigma) - F_0(0)\,.
\end{equation}
Moreover, $F_1(\sigma)$ has a nonpositive real part.
\end{lemma}
\begin{proof}
The identity  \eqref{cal_Z_F_1} follows immediately from \eqref{wt_Z_eta}, \eqref{def_F_1}, and the definition of $Z^0$.
From \eqref{Z_log_expanded} we get
\begin{equation*}
F_1(\sigma) = \sum_{\ell \in \N^*} \frac{1}{\ell} \tr \pb{W^{\ell \nu,0}(-\kappa + \ii \sigma) - W^{\ell \nu,0}(-\kappa)}
\end{equation*}
By \eqref{W_FK}, we find $\abs{\tr W^{\ell \nu,0}(-\kappa + \ii \sigma)} \leq \tr W^{\ell \nu, 0}(-\kappa)$, and the claim follows.
\end{proof}

\subsection{Space-time representation} \label{sec:space-time}

As in the previous subsections, we assume that $H_n$ is of the form \eqref{Hamiltonian_H_n}.
In this subsection we derive a space-time representation for the function $F_1(\sigma)$ from \eqref{def_F_1}. We begin with a pedagogical interlude that illustrates the main point of this subsection. Let $\kappa > 0$ and consider the operator $K \deq \partial_\tau + \kappa$ on the space $L^2([0,\nu])$ with periodic boundary conditions. Clearly, $K$ is invertible with bounded inverse, and we claim that its inverse has operator kernel
\begin{equation*}
(K^{-1})^{\tau,\tilde \tau} = \sum_{r \in \nu \N} \ind{\tau + r > \tilde \tau} \, \ee^{-\kappa(\tau - \tilde \tau + r)}\,.
\end{equation*}
There are several ways of arriving at this formula. One is by Fourier series on $[0,\nu]$, which diagonalizes $K$. Another is by direct inspection: for $\tau, \tilde \tau \in (0,\nu)$ we have
\begin{align*}
(K K^{-1})^{\tau, \tilde \tau} &= (\partial_\tau + \kappa) \sum_{r \in \nu \N} \ind{\tau + r > \tilde \tau} \, \ee^{-\kappa(\tau - \tilde \tau + r)}
\\
&= \sum_{r \in \nu \N} \pb{\delta(\tau - \tilde \tau + r) + (\kappa - \kappa) \ind{\tau  + r > \tilde \tau}} \, \ee^{-\kappa(\tau - \tilde \tau + r)}
\\
&= \delta(\tau - \tilde \tau)\,.
\end{align*}

We use the following notations for operators acting on functions depending on space and time.

\begin{definition}
We denote by $\Tr$ the trace on operators on $L^2([0,\nu] \times \Lambda)$.
Let $\ops$ be an operator on $L^2([0,\nu] \times \Lambda)$, acting on functions $f \in L^2(I \times \Lambda)$ of time $\tau \in I$ and space $x \in \Lambda$. We use the notation $\ops_{x, \tilde x}^{\tau, \tilde \tau}$ for the operator kernel of $\ops$; thus, $(\ops f)(\tau,x) = \int \dd \tilde \tau \int \dd \tilde x \, \ops^{\tau, \tilde \tau}_{x, \tilde x} \,f(\tilde \tau, \tilde x)$.
\end{definition}

The main object of this subsection is the following operator.

\begin{definition}
Let $u$ be as in Definition \ref{def:Gamma_u}, and suppose that $\re u \leq -c$ for some positive constant $c$. Abbreviate $h_\tau \deq \frac{1}{2} \Delta + u(\tau)$. On $L^2([0,\nu] \times \Lambda)$ define the operator
\begin{equation} \label{def_K}
K(u) \deq \partial_\tau - h_\tau\,,
\end{equation}
with periodic boundary conditions in $[0,\nu]$.
\end{definition}
The Green function of $K(u)$ can be easily expressed in terms of the propagator $W^{\tau, \tilde \tau}(u)$ from Definition \ref{def:Gamma_u}.

\begin{lemma} \label{lem:G_K}
The Green function of $K(u)$ has operator kernel
\begin{equation} \label{K_inv}
(K(u)^{-1})^{\tau, \tilde \tau}_{x, \tilde x} = \sum_{r \in \nu \N} \ind{\tau  + r > \tilde \tau}  \, W^{\tau, \tilde \tau - r}_{x,\tilde x}(u)
= \sum_{r \in \nu \N} \ind{\tau  + r > \tilde \tau} \, W^{\tau + r, \tilde \tau}_{x,\tilde x}(u)\,.
\end{equation}
\end{lemma}
\begin{proof}
We drop the argument $u$ from our notation. Denote by $K^{-1}$ the operator with kernel given in \eqref{K_inv}. Note first that by Lemma \ref{lem:Gamma} (iii) both series in \eqref{K_inv} converge in operator norm on $L^2(\Lambda)$, for fixed $\tau, \tilde \tau$. Moreover, by the periodicity of $W^{\tau, \tilde \tau}$ from Lemma \ref{lem:Gamma} (ii), both series coincide, and $K^{-1}$ is periodic in the sense that $(K^{-1})^{\tau + \nu, \tilde \tau + \nu} = (K^{-1})^{\tau, \tilde \tau}$.

We have to verify that $K K^{-1} = 1$.
To that end, we compute, for $\tau, \tilde \tau \in (0,\nu)$,
\begin{align*}
(K K^{-1})^{\tau, \tilde \tau}_{x, \tilde x} &=  \sum_{r \in \nu \N} (\partial_\tau - h_\tau) \ind{\tau + r > \tilde \tau} \, W^{\tau, \tilde \tau - r}_{x,\tilde x}
\\
&= \sum_{r \in \nu \N} \pB{\delta(\tau - \tilde \tau + r) + (h_\tau - h_\tau) \ind{\tau + r > \tilde \tau}} W^{\tau, \tilde \tau - r}_{x,\tilde x}
\\
&=\sum_{r \in \nu \N} \delta(\tau - \tilde \tau + r) W^{\tau, \tilde \tau - r}_{x,\tilde x} = \delta(\tau - \tilde \tau) W^{\tau, \tau}_{x,\tilde x} = \delta(\tau - \tilde \tau) \delta(x - \tilde x)\,,
\end{align*}
where we used that $\abs{\tau - \tilde \tau} < \nu$.
\end{proof}

Next, by cyclicity of the trace, and Lemma \ref{lem:Gamma} (i) and (ii), we note that for any $\tau \in [0,\nu]$ and $r \in \nu \N^*$ we have (omitting the arguments $u$ for brevity)
\begin{equation*}
\tr W^{\tau+r,\tau} = \tr W^{\tau+r,r} \,W^{r,\tau} = \tr W^{r,\tau} \, W^{\tau+r,r} = \tr W^{r,\tau} \, W^{\tau,0} = \tr W^{r,0}\,.
\end{equation*}
Hence,
\begin{equation*}
\Tr K(u)^{-1} = \nu \sum_{r \in \nu \N^*} \tr W^{r,0}(u)\,.
\end{equation*}
Again by Lemma \ref{lem:Gamma} (i) and (ii), we have $W^{k \nu,0}(u) = (W^{\nu, 0}(u))^k$, so that by Lemma \ref{lem:Gamma} (iii) we have
\begin{equation} \label{tr_K_inv}
\Tr K(u)^{-1} = \nu \sum_{k \in \N^*} \tr (W^{\nu, 0}(u))^k = \nu \tr \pbb{\frac{1}{1 - W^{\nu, 0}(u)} - 1}\,.
\end{equation}

Next, we note that for any $a,b \in \C$ with strictly positive real parts we have the integral representation
\begin{equation} \label{integral_identity}
\log a - \log b = - \int_0^\infty \dd t \pbb{\frac{1}{t + a} - \frac{1}{t + b}}\,,
\end{equation}
as can be seen by computing the integral on the right-hand side from $0$ to $M > 0$ and then taking the limit $M \to \infty$.
We apply the integral representation \eqref{integral_identity} to \eqref{Z_int_rep} and \eqref{def_F_1}, and obtain
\begin{equation} \label{calZ_integral_rep}
F_1(\sigma) =   \int_0^\infty \dd t \, \tr \pbb{\frac{1}{t + 1 - W^{\nu,0}(-\kappa + \ii \sigma)} - \frac{1}{t + 1 - W^{\nu,0}(-\kappa)}}\,.
\end{equation}
Since
\begin{equation*}
\frac{1}{1 + t - W^{\nu,0}(u)} = \frac{1}{1 + t} \frac{1}{1 - W^{\nu,0}(u - \log(1+t) / \nu)}\,,
\end{equation*}
we find from \eqref{tr_K_inv} that
\begin{equation} \label{K_W_link}
\frac{1}{1+t} \Tr \frac{1}{K(u - \log(1+t) / \nu)} = \nu \tr \pbb{\frac{1}{1 + t - W^{\nu,0}(u)} - \frac{1}{1 + t}}\,.
\end{equation}
Plugging \eqref{K_W_link} into \eqref{calZ_integral_rep} yields
\begin{equation}
F_1(\sigma) = \int_0^\infty \frac{\dd t}{(1 + t) \nu} \, \Tr \pbb{\frac{1}{K(-\kappa + \ii \sigma - \log(1+t) / \nu)} - \frac{1}{K(-\kappa - \log(1+t) / \nu)}}\,.
\end{equation}
Doing the change of variables $\log(1 + t) / \nu \mapsto t$ and recalling Lemma \ref{lem:lim_eta}, we get the following result.

\begin{proposition}[Functional integral representation of partition function] \label{prop:fif_bare}
Let $H$ be defined as in \eqref{H_sq} and suppose that Assumptions \ref{ass:torus} and \ref{ass:v} hold. Then we have
\begin{equation*}
Z = \lim_{\eta \to 0} \wt Z_\eta\,,
\end{equation*}
where
\begin{equation*}
\wt Z_\eta  = Z^0 \int \mu_{\cal C_\eta}(\dd \sigma) \, \ee^{F_1(\sigma)}\,,
\end{equation*}
and
\begin{equation} \label{Z_space-time_integral}
F_1(\sigma) = \int_0^\infty \dd t \, \Tr \pbb{\frac{1}{t + K(-\kappa + \ii \sigma)} - \frac{1}{t + K(-\kappa)}}\,.
\end{equation}
\end{proposition}

\subsection{Wick ordering} \label{sec:Wick}
In this subsection we explain how the Wick ordering from \eqref{H_ren}, \eqref{def_rho} can be incorporated in the functional integral formulation from Proposition \ref{prop:fif_bare}. Thus, we take $H$ to be of the form \eqref{H_ren}, \eqref{def_rho}. Throughout the following we abbreviate
\begin{equation*}
\scalar{\sigma}{\varrho} \deq \frac{1}{\nu}\int_0^\nu \dd \tau \int \dd x \, \sigma(\tau, x) \, \varrho(x)\,.
\end{equation*}
The following definition is the Wick-ordered version of Definition \ref{def:R}.

\begin{definition} \label{def:R_Wick}
Let $\eta > 0$. On $\cal F$ we define the operator $R_\eta \deq \bigoplus_{n \in \N} R_{n, \eta}$, where $R_{n, \eta}$ has operator kernel
\begin{equation} \label{R_eta_expanded}
(R_{n, \eta})_{\f x, \tilde {\f x}} \deq \ee^{-\nu \kappa n} \int \mu_{\cal C_\eta}(\dd \sigma) \, \ee^{-\ii \scalar{\sigma}{\varrho}} \int \prod_{i = 1}^n \bb W^{\nu,0}_{x_i, \tilde x_i}(\dd \omega_i) \, \prod_{i = 1}^n \ee^{\ii \int_0^\nu \dd t \, \sigma(t, \omega_i(t))}\,.
\end{equation}
Moreover, we define
\begin{equation*}
Z_\eta \deq \Tr_{\cal F} (R_\eta) \,, \qquad \cal Z_\eta \deq Z_\eta / Z^0\,,
\end{equation*}
\end{definition}
We also use the notation
\begin{equation*}
\cal Z \deq Z / Z^0
\end{equation*}
for the relative partition function.

\begin{proposition}[Functional integral representation under Wick ordering] \label{prop:fif_renorm}
Let $H$ be defined as in \eqref{H_ren}, \eqref{def_rho}, and suppose that Assumptions \ref{ass:torus} and \ref{ass:v} hold.
\begin{enumerate}[label=(\roman*)]
\item
We have
\begin{equation}
R_\eta = \int \mu_{\cal C_\eta}(\dd \sigma)\, \ee^{-\ii \scalar{\sigma}{\varrho}} \, \Gamma \pb{W^{\nu,0}(-\kappa + \ii \sigma)}\,.
\end{equation}
\item
For all $n \in \N$ we have $\lim_{\eta \to 0} \norm{R_{n, \eta} - \ee^{-H_n}}_{L^\infty} = 0$.
\item
We have
\begin{equation}
\cal Z = \lim_{\eta \to 0} {\cal Z}_\eta\,,
\end{equation}
where
\begin{equation} \label{wtZ_F2}
{\cal Z}_\eta  = \int \mu_{\cal C_\eta}(\dd \sigma) \, \ee^{F_2(\sigma)}
\end{equation}
and
\begin{equation} \label{def_F2}
F_2(\sigma) \deq \int_0^\infty \dd t \, \Tr \pbb{\frac{1}{t + K(-\kappa + \ii \sigma)} - \frac{1}{t + K(-\kappa)} - \frac{1}{t + K(-\kappa)} \, \ii \sigma \, \frac{1}{t + K(-\kappa)}}\,,
\end{equation}
where we regard $\sigma$ as a multiplication operator on $L^2([0,\nu] \times \Lambda)$.
 \item
 The function $F_2(\sigma)$ has a nonpositive real part.
\end{enumerate}
\end{proposition}

The rest of this subsection is devoted to the proof of Proposition \ref{prop:fif_renorm}. We start with a preliminary calculation, which relates the Wick-ordering phase $\scalar{\sigma}{\varrho}$ to the last term of \eqref{def_F2}.

\begin{lemma} \label{lem:quantum_density}
Recall the free quantum density $\varrho(x)$ from \eqref{def_rho}. Under Assumption \ref{ass:torus} we have
\begin{equation*}
\scalar{\sigma}{\varrho} = \int_0^\infty \dd t \, \Tr \pbb{\frac{1}{t + K(-\kappa)} \, \sigma \, \frac{1}{t + K(-\kappa)}}
\end{equation*}
\end{lemma}
\begin{proof}
By the definition \eqref{def_K} we have $t + K(-\kappa) = K(-\kappa - t)$, and hence by Lemma \ref{lem:G_K}, we have
\begin{align*}
&\quad \int_0^\infty \dd t \, \Tr \pbb{\frac{1}{t + K(-\kappa)} \, \sigma \, \frac{1}{t + K(-\kappa)}}
\\
&= \int_0^\infty \dd t \, \sum_{r, \tilde r \in \nu \N} \int_0^\nu \dd \tau\, \dd \tilde \tau \, \ind{\tau  + r > \tilde \tau} \, \ind{\tilde \tau + \tilde r >  \tau} \int \dd x \, \dd \tilde x \, W^{\tau + r, \tilde \tau}_{x, \tilde x}(-\kappa - t) \, W^{\tilde \tau + \tilde r, \tau}_{\tilde x, x}(-\kappa - t) \, \sigma(\tau,x)
\\
&=
 \sum_{r, \tilde r \in \nu \N} \frac{\ind{r + \tilde r > 0}}{r + \tilde r} \int_0^\nu \dd \tau\, \dd \tilde \tau \, \ind{\tau - \tilde r < \tilde \tau < r + \tau } \int \dd x  \, \pb{\ee^{(\Delta / 2 - \kappa) (r + \tilde r)}}_{x, x} \, \sigma(\tau,x)\,.
\end{align*}
In the last equality, we used that $W^{\tau, \tilde \tau}(-\kappa - t) = \ee^{(\tau - \tilde \tau) (\Delta/2 - \kappa - t)}$ by Definition \ref{def:Gamma_u}.
Next, we split the sum over $r, \tilde r$ as $r = 0, \tilde r > 0$, $r > 0, \tilde r =0$, and $r > 0, \tilde r > 0$, and use that, for $r + \tilde r > 0$,
\begin{equation*}
\int_0^\nu \dd \tilde \tau \, \ind{\tau - \tilde r < \tilde \tau < r + \tau } = 
\begin{cases}
\tau & \text{if }r = 0
\\
\nu - \tau & \text{if } \tilde r = 0
\\
\nu & \text{otherwise}\,.
\end{cases}
\end{equation*}
This gives
\begin{align*}
&\quad \int_0^\infty \dd t \, \Tr \pbb{\frac{1}{t + K(-\kappa)} \, \sigma \, \frac{1}{t + K(-\kappa)}}
\\
&= \int_0^\nu \dd \tau \int \dd x \, \sigma(\tau,x) \, \qBB{\tau \sum_{\tilde r \in \nu \N^*} \frac{\ee^{(\Delta / 2 - \kappa) \tilde r}}{\tilde r}  + (\nu - \tau) \sum_{r \in \nu \N^*} \frac{\ee^{(\Delta / 2 - \kappa) r}}{r}  + \nu \sum_{r, \tilde r \in \nu \N^*} \frac{\ee^{(\Delta / 2 - \kappa) (r + \tilde r)}}{r + \tilde r}}_{x,x}
\end{align*}
Using the identities
\begin{equation*}
\sum_{n \in \N^*} \frac{1}{n} \, \ops^n = - \log (1 - \ops)\,, \qquad \sum_{n,m \in \N^*} \frac{1}{n+m} \, \ops^{n+m} = \frac{\ops}{1 - \ops} + \log (1 - \ops)\,,
\end{equation*}
we therefore obtain
\begin{equation*}
\int_0^\infty \dd t \, \Tr \pbb{\frac{1}{t + K(-\kappa)} \, \sigma \, \frac{1}{t + K(-\kappa)}} = \frac{1}{\nu}\int_0^\nu \dd \tau \int \dd x \, \sigma(\tau,x)\, \varrho(x)\,,
\end{equation*}
where
\begin{equation*}
\varrho(x) = \nu \pbb{\frac{\ee^{\nu(\Delta/2 - \kappa)}}{1 - \ee^{\nu(\Delta/2 - \kappa)}}}_{x,x}\,,
\end{equation*}
by the definition \eqref{def_rho} and Lemma \ref{Quantum Wick theorem} (i).
\end{proof}

\begin{proof}[Proof of Proposition \ref{prop:fif_renorm}]
Part (i) is immediate from \eqref{W_FK}.

Next, we prove (ii). Let $\eta \geq 0$ and abbreviate $\delta_{0, \nu} \equiv \delta$.  For $f(\tau, x) = \sum_{i = 1}^n \delta(x - \omega_i(\tau)) - \varrho(x) / \nu$ we have $\frac{1}{2} \scalar{f}{\cal C_\eta f} = - f_\eta(\omega_1, \dots, \omega_n)$, where
\begin{equation*}
f_\eta(\omega_1, \dots, \omega_n) \deq - \frac{\lambda}{2 \nu} \sum_{i,j = 1}^n \int_0^\nu \dd t \, \dd s \, \delta_{\eta, \nu}(t - s) v_\eta(\omega_i(t) - \omega_j(s)) + \frac{\lambda}{\nu} \varrho \hat v(0) n - \frac{\lambda}{2 \nu^2} \hat v(0) \varrho^2 \abs{\Lambda}\,.
\end{equation*}
(Here we used that, by Assumption \ref{ass:torus}, $\varrho(x) = \varrho$.) Since $\cal C_\eta$ is a positive operator, we conclude that $f_\eta(\omega_1, \dots, \omega_n) \leq 0$.
Moreover, \eqref{Hubbard_Stratonovich_formula} yields
\begin{equation*}
(R_{n, \eta})_{\f x, \tilde {\f x}} = \ee^{-\nu \kappa n} \int \prod_{i = 1}^n \bb W^{\nu,0}_{x_i, \tilde x_i}(\dd \omega_i) \, \ee^{f_\eta(\omega_1, \dots, \omega_n)}
\,.
\end{equation*}
Next, a calculation analogous to \eqref{H_ren_H_sq} yields
\begin{equation} \label{H_n_first_quantized}
H_n = \nu \sum_{i = 1}^n h_i + \frac{\lambda}{2} \sum_{i,j = 1}^n v(x_i - x_j) - \frac{\lambda}{\nu} \varrho \hat v(0) n + \frac{\lambda}{2 \nu^2} \hat v(0) \varrho^2 \abs{\Lambda}\,.
\end{equation}
Hence,
\begin{align}
\absb{(R_{n, \eta})_{\f x, \tilde {\f x}} - (\ee^{-H_n})_{\f x, \tilde {\f x}}} &= \ee^{-\nu \kappa n} \absbb{\int \prod_{i = 1}^n \bb W^{\nu,0}_{x_i, \tilde x_i}(\dd \omega_i) \, \pB{\ee^{f_\eta(\omega_1, \dots, \omega_n)}
- \ee^{f_0(\omega_1, \dots, \omega_n)}}}
\notag \\
&\leq
\int \prod_{i = 1}^n \bb W^{\nu,0}_{x_i, \tilde x_i}(\dd \omega_i) \, \absB{f_\eta(\omega_1, \dots, \omega_n) - f_0(\omega_1, \dots, \omega_n)}
\notag\\
&\leq
\frac{\lambda}{2 \nu} \sum_{i,j = 1}^n \int \prod_{i = 1}^n \bb W^{\nu,0}_{x_k, \tilde x_k}(\dd \omega_k) \absBB{\int_0^\nu \dd t \, \dd s \, (\delta_{\eta, \nu} - \delta)(t - s) \, v(\omega_i(t) - \omega_j(s))}
\notag \\ \label{R_H_estimate}
&\quad + 
\frac{\lambda}{2 \nu} \sum_{i,j = 1}^n \int \prod_{k = 1}^n \bb W^{\nu,0}_{x_k, \tilde x_k}(\dd \omega_k) \absBB{\int_0^\nu \dd t \, \dd s \, \delta_{\eta, \nu}(t - s)  (v_\eta - v)(\omega_i(t) - \omega_j(s))}\,.
\end{align}
We estimate the two terms on the right-hand side of \eqref{R_H_estimate} separately.

The second term on the right-hand side of \eqref{R_H_estimate} is easy: by Lemma \ref{lem:heat_estimate}, it is estimated by
\begin{equation*}
\frac{\lambda}{2\nu} n^2 \pbb{C_d (1 + \nu^{-d/2})}^n \, \nu \, \norm{v_\eta - v}_{L^\infty}
\end{equation*}
uniformly in $\f x, \tilde {\f x}$, which tends to zero as $\eta \to 0$ by continuity of $v$.

To estimate the first term on the right-hand side of \eqref{R_H_estimate}, we recall \eqref{def_W} and note that $\psi^{\nu}(x - \tilde x)$ is bounded uniformly in $x, \tilde x$ by Lemma \ref{lem:heat_estimate}. Hence, since $\int_0^\nu \dd \tau\, \delta_{\eta,\nu}(\tau)=1$, it suffices to show that
\begin{equation*}
\int \P^{\nu,0}_{x_1, \tilde x_1}(\dd \omega_1) \, \P^{\nu,0}_{x_2, \tilde x_2}(\dd \omega_2) \,
\int_0^\nu \dd t \, \dd s \, \delta_{\eta, \nu} (t - s) \, \absb{v(\omega_1(t) - \omega_2(s)) - v(\omega_1(s) - \omega_2(s))}
\end{equation*}
converges to $0$ as $\eta \to 0$, uniformly in $x_1, x_2, \tilde x_1, \tilde x_2$. Since $v$ is uniformly continuous on $\Lambda$, we conclude that it suffices to show that for every $\epsilon > 0$
\begin{equation} \label{PP_est}
\int \P^{\nu,0}_{x_1, \tilde x_1}(\dd \omega_1) \,
\int_0^\nu \dd t \, \dd s \, \delta_{\eta, \nu} (t - s) \, \ind{\abs{\omega_1(t) - \omega_1(s)}_L > \epsilon}
\end{equation}
converges to $0$ as $\eta \to 0$, uniformly in $x_1, \tilde x_1$. We estimate this using Chebyshev's inequality and Lemma \ref{lem:P_cont}:
\begin{align*}
\abs{\eqref{PP_est}} &\leq 
\frac{1}{\epsilon^2} \int_0^\nu \dd t \, \dd s \, \delta_{\eta, \nu} (t - s) \, \int \P^{\nu,0}_{x_1, \tilde x_1}(\dd \omega_1) \abs{\omega_1(t) - \omega_1(s)}_L^2
\\
&\leq \frac{C_{d,L,\nu}}{\epsilon^2} \int_0^\nu \dd t \, \dd s \, \delta_{\eta, \nu} (t - s) \, \pb{\abs{t - s} + \abs{t - s}^2} \,,
\end{align*}
which tends to zero as $\eta \to 0$ as desired, since $\delta_{\eta,\nu}$ converges to the $\delta$ function in distribution. This concludes the proof of (ii).

Next, we prove (iii). Note that the bound \eqref{kernel_est} trivially also holds for $R_{n, \eta}$, so that by \eqref{sum_free_Z}, dominated convergence, and part (ii), we obtain $Z = \lim_{\eta \to 0} {Z}_\eta$ and hence also $\cal Z = \lim_{\eta \to 0} {\cal Z}_\eta$. Moreover, from part (i) we get
\begin{equation*}
{Z}_\eta = \int \mu_{\cal C_\eta}(\dd \sigma)\, \ee^{-\ii \scalar{\sigma}{\varrho}} \, \Tr_{\cal F} \pb{\Gamma \pb{W^{\nu,0}(-\kappa + \ii \sigma)}}\,.
\end{equation*}
Recalling Lemmas \ref{lem:trace_formulas} and \ref{lem:quantum_density}, we well as the definitions \eqref{Z_int_rep}, \eqref{def_F_1}, and \eqref{def_F2}, we obtain \eqref{wtZ_F2}, and (iii) is proved.

Finally, (iv) follows from Lemma \ref{lem:F_1}, \eqref{Z_space-time_integral}, and \eqref{def_F2}.
\end{proof}

\subsection{Reduced density matrices}
Let $\rho = \bigoplus_{n \in \N} \rho_n$ be a nonnegative trace class operator on $\cal F$, and for $p \in \N^*$ denote the $p$-particle reduced density matrix with kernel
\begin{equation*}
(\Gamma_p(\rho))_{x_1 \dots x_p, \tilde x_1 \dots \tilde x_p} \deq \Tr_{\cal F} \pb{a^*(\tilde x_1) \cdots a^*(\tilde x_p) a(x_1) \cdots a(x_p) \, \rho}\,.
\end{equation*}

\begin{proposition}[Functional integral representation of reduced density matrices] \label{prop:red_dens_matrix}
Let $H$ be defined as in \eqref{H_ren}, \eqref{def_rho}, and $R_\eta$ be as in Definition \ref{def:R_Wick}. Then for any $p \in \N^*$ the following holds.
\begin{enumerate}[label=(\roman*)]
\item
We have the convergence
\begin{equation*}
\lim_{\eta \to 0} \normb{\Gamma_p(R_\eta) - \Gamma_p(\ee^{-H})}_{L^\infty} = 0\,.
\end{equation*}
\item
We have
\begin{equation*}
\Gamma_p\pb{R_\eta}
= Z^0 \,p! P_p \int \mu_{\cal C_\eta}(\dd \sigma)\, \ee^{F_2(\sigma)} \, \pB{\pb{K(-\kappa + \ii \sigma)^{-1}}^{0,0}}^{\otimes p}\,.
\end{equation*}
\end{enumerate}
\end{proposition}
\begin{proof}
A standard calculation with creation and annihilation operators shows that
\begin{equation*}
\Gamma_p(\rho) = \sum_{n \geq p} \frac{n!}{(n - p)!} \tr_{p+1, \dots, n} \pb{P_n \rho_{n}}\,,
\end{equation*}
where $\tr_{p+1, \dots, n} (\cdot)$ denotes partial trace over the variables $x_{p+1}, \dots, x_{n}$.
Thus,
\begin{multline*}
\abs{\Gamma_p(R_\eta)_{\f x, \tilde {\f x}}} \leq 
\sum_{n \geq p} \frac{n!}{(n - p)!} \absb{\pb{\tr_{p+1, \dots, n} \pb{P_n R_{n,\eta}}}_{\f x, \tilde {\f x}}}
\\
\leq
\sum_{n \geq p} \frac{n!}{(n - p)!} \pb{\tr_{p+1, \dots, n} \pb{P_n \ee^{-H_n^0}}}_{\f x, \tilde {\f x}}
= \Gamma_p(\Gamma(\ee^{-h}))_{\f x, \tilde {\f x}}
\,,
\end{multline*}
where in the second step we used \eqref{def_Pn}, \eqref{R_eta_expanded}, and Lemma \ref{FK_continuous}. By Lemma \ref{Quantum Wick theorem},
\begin{equation} \label{}
\Gamma_p(\Gamma(\ee^{-h}))_{\f x, \tilde {\f x}} = \sum_{\pi \in S_p} \prod_{i = 1}^p \pbb{\frac{\ee^{-h}}{1 - \ee^{-h}}}_{x_i, \tilde x_{\pi(i)}} = \sum_{\pi \in S_p} \sum_{\f k \in (\N^*)^p} \prod_{i = 1}^p (\ee^{-k_i h})_{x_i ,\tilde x_{\pi(i)}}\,.
\end{equation}
By multiplying $\ee^{-h}$ with a parameter $t > 0$, we obtain the identity of power series
\begin{equation*}
\sum_{n \geq p} \frac{n! \, t^n}{(n - p)!} \pb{\tr_{p+1, \dots, n} \pb{P_n \ee^{-H_n^0}}}_{\f x, \tilde {\f x}} =
\sum_{\pi \in S_p} \sum_{\f k \in (\N^*)^p} t^{\abs{\f k}} \prod_{i = 1}^p (\ee^{-k_i h})_{x_i ,\tilde x_{\pi(i)}}\,,
\end{equation*}
where $\abs{\f k} = k_1 + \cdots + k_p$. Note that all terms are positive. In particular, for any $N \geq p$ we have
\begin{align*}
\sum_{n \geq N} \frac{n!}{(n - p)!} \absb{\pb{\tr_{p+1, \dots, n} \pb{P_n R_{n,\eta}}}_{\f x, \tilde {\f x}}} &\leq 
\sum_{\pi \in S_p} \sum_{\f k \in (\N^*)^p} \ind{\abs{\f k} \geq N} \prod_{i = 1}^p (\ee^{-k_i h})_{x_i ,\tilde x_{\pi(i)}}
\\
&\leq p! \sum_{\f k \in (\N^*)^p} \ind{\abs{\f k} \geq N} \prod_{i = 1}^p C_d (L^{-d} + (\nu k_i)^{-d/2}) \, \ee^{-\nu \kappa k_i}
\\
&\leq C_{d,p,L,\nu} \sum_{\f k \in (\N^*)^p} \ind{\abs{\f k} \geq N} \, \ee^{-\nu \kappa \abs{\f k}}\,,
\end{align*}
which tends to zero as $N \to \infty$, uniformly in $\f x, \tilde {\f x}$ and $\eta$.

We conclude that to prove part (i) it suffices to prove, for each $n \geq p$, that
\begin{equation*}
\lim_{\eta \to 0} \normb{\tr_{p+1, \dots, n} \pb{P_n R_{n,\eta}} - \tr_{p+1, \dots, n} \pb{P_n \ee^{-H_n}}}_{L^\infty} = 0\,,
\end{equation*}
which itself follows immediately from Proposition \ref{prop:fif_renorm} (ii) and the definition of the partial trace. This concludes the proof of part (i).

Next, we prove (ii).
By Proposition \ref{prop:fif_renorm} (i), we have
\begin{equation*}
(\Gamma_p\pb{R_\eta})_{\f x, \tilde {\f x}} = \int \mu_{\cal C_\eta}(\dd \sigma)\, \ee^{-\ii \scalar{\sigma}{\varrho}} \, \Tr_{\cal F} \pB{a^*(\tilde x_1) \cdots a^*(\tilde x_p) a(x_1) \cdots a(x_p) \Gamma \pb{W^{\nu,0}(-\kappa + \ii \sigma)}}\,.
\end{equation*}
By Lemma \ref{Quantum Wick theorem} we get
\begin{equation*}
(\Gamma_p\pb{R_\eta})_{\f x, \tilde {\f x}} =  \int \mu_{\cal C_\eta}(\dd \sigma)\, \ee^{-\ii \scalar{\sigma}{\varrho}} \Tr_{\cal F} \pB{\Gamma \pb{W^{\nu,0}(-\kappa + \ii \sigma)}} \sum_{\pi \in S_p} \prod_{i = 1}^p \pbb{\frac{W^{\nu,0}(-\kappa + \ii \sigma)}{1 - W^{\nu,0}(-\kappa + \ii \sigma)}}_{x_i, \tilde x_{\pi(i)}}\,.
\end{equation*}
Recalling Lemmas \ref{lem:trace_formulas} and \ref{lem:G_K} and the definition \eqref{Z_int_rep}, we obtain
\begin{align*}
(\Gamma_p\pb{R_\eta})_{\f x, \tilde {\f x}}
&= \int \mu_{\cal C_\eta}(\dd \sigma)\, \ee^{-\ii \scalar{\sigma}{\varrho}} \ee^{F_0(\sigma)} \sum_{\pi \in S_p} \prod_{i = 1}^p \pB{K(-\kappa + \ii \sigma)^{-1}}^{0,0}_{x_i, \tilde x_{\pi(i)}}
\\
&= Z^0 \int \mu_{\cal C_\eta}(\dd \sigma)\, \ee^{F_2(\sigma)} \sum_{\pi \in S_p} \prod_{i = 1}^p \pB{K(-\kappa + \ii \sigma)^{-1}}^{0,0}_{x_i, \tilde x_{\pi(i)}}\,,
\end{align*}
where in the second step we used the definition \eqref{def_F2} and Lemma \ref{lem:quantum_density}. This is (ii).
\end{proof}

In analogy to the definition \eqref{def_gammap}, which reads
\begin{equation*}
\Gamma_p = \frac{1}{Z} \, \Gamma_p \p{\ee^{-H}}\,,
\end{equation*}
we define
\begin{equation*}
\Gamma_{p, \eta} \deq \frac{1}{Z_\eta} \, \Gamma_p \p{R_\eta}\,.
\end{equation*}
As a consequence of Propositions \ref{prop:fif_renorm} and \ref{prop:red_dens_matrix}, we have
\begin{equation} \label{conv_gamma_p_eta}
\lim_{\eta \to 0} \norm{\Gamma_p - \Gamma_{p,\eta}}_{L^\infty} = 0\,.
\end{equation}

By Proposition \ref{prop:red_dens_matrix} (ii) we have
\begin{equation} \label{Gamma_eta}
\Gamma_{p,\eta} = \frac{p!}{{\cal Z}_\eta} P_p \int \mu_{\cal C_\eta}(\dd \sigma)\, \ee^{F_2(\sigma)} \, \pB{\pb{K(-\kappa + \ii \sigma)^{-1}}^{0,0}}^{\otimes p}\,.
\end{equation}
We define the Wick-ordered version of \eqref{Gamma_eta} as
\begin{equation} \label{def_eh_gamma_eta}
\wh \Gamma_{p,\eta} \deq \frac{p!}{{\cal Z}_\eta} P_p \int \mu_{\cal C_\eta}(\dd \sigma)\, \ee^{F_2(\sigma)} \,\pB{\pb{K(-\kappa + \ii \sigma)^{-1}}^{0,0} - \pb{K(-\kappa)^{-1}}^{0,0}}^{\otimes p}\,.
\end{equation}

\begin{proposition} \label{prop:gamma_hat}
For any $\nu > 0$ and $p \in \N^*$ we have
\begin{equation*}
\lim_{\eta \to 0}\norm{\wh \Gamma_{p,\eta} - \wh \Gamma_p}_{L^\infty} = 0\,,
\end{equation*}
where $\wh \Gamma_p$ was defined in \eqref{def_Gamma_renormalized}.
\end{proposition}
\begin{proof}
By expanding the product in \eqref{def_eh_gamma_eta} and using that $P_p = (P_p)^2$ commutes with any tensor power, we obtain
\begin{equation*}
\wh \Gamma_{p,\eta} = \frac{p!}{{\cal Z}_\eta} \sum_{k = 0}^p \binom{p}{k} (-1)^{p - k} \int \mu_{\cal C_\eta}(\dd \sigma)\, \ee^{F_2(\sigma)} \, P_p\pB{\pb{K(-\kappa + \ii \sigma)^{-1}}^{0,0}}^{\otimes k} \otimes \pB{\pb{K(-\kappa)^{-1}}^{0,0}}^{\otimes (p - k)} P_p\,.
\end{equation*}
By Lemmas \ref{lem:G_K} and \ref{Quantum Wick theorem} we find
\begin{equation*}
\pb{K(-\kappa)^{-1}}^{0,0} = \sum_{n \in \N^*} W^{\nu n,0}(-\kappa) = \sum_{n \in \N^*} \ee^{\nu n(\Delta/2 - \kappa)} = \Gamma_1^0\,.
\end{equation*}
Moreover, by Lemma \ref{Quantum Wick theorem}, we find
\begin{equation*}
\Gamma_p^0 = p! P_p (\Gamma_1^0)^{\otimes p}\,.
\end{equation*}
Hence,
\begin{equation*}
\wh \Gamma_{p,\eta} = \sum_{k = 0}^p \binom{p}{k}^2 (-1)^{p - k} \, P_p \pb{\Gamma_{k,\eta} \otimes \Gamma_{p-k}^0} P_p\,,
\end{equation*}
and the claim follows from \eqref{conv_gamma_p_eta}.
\end{proof}

\section{Functional integral representation of classical field theory} \label{sec:classical}

In this section, we derive the functional integral representation of the classical theory. The analysis of the partition function is given in Section \ref{functional_integral_classical_Z}.  The correlation functions are analysed in Section \ref{functional_integral_classical_gamma}. Throughout this section we make Assumptions \ref{ass:torus} and \ref{ass:v}. The main results of this section are Proposition \ref{prop:Z_field_Wick_ordered} for the partition function and Proposition \ref{prop:gamma_regularization} for the correlation functions.

\subsection{Partition function}
\label{functional_integral_classical_Z}

In this subsection we derive a functional integral representation for the partition function of the classical field theory, defined in \eqref{def_z}. Let $\varphi$ be the smooth cutoff function from Section \ref{intro_sigma}, and recall the approximate delta function from \eqref{Fourier_regularization}. Let $\mu_{v_\eta}$ be the real Gaussian measure with mean zero and covariance
\begin{equation} \label{cov_v_eta}
\int \mu_{v_\eta}(\dd \xi)\, \xi(x) \, \xi(\tilde x) = v_\eta (x - \tilde x)\,.
\end{equation}
For $\eta > 0$, a representation analogous to \eqref{sigma_X_rep} expresses $\xi$ as a finite sum, and therefore under the law $\mu_{v_\eta}$, the field $\xi$ is almost surely a smooth periodic function on $\Lambda$, and the integration over $\mu_{v_\eta}$ is a finite-dimensional Gaussian integral.

\begin{proposition}
\label{prop:Z_field_Wick_ordered}
Suppose that Assumptions \ref{ass:torus} and \ref{ass:v} hold.
\begin{enumerate}[label=(\roman*)]
\item
The partition function \eqref{def_z} is given by
\begin{equation*}
\zeta = \lim_{\eta \to 0} \zeta_\eta\,,
\end{equation*}
where
\begin{equation} \label{def_z_nu}
\zeta_\eta \deq \int \mu_{v_{\eta}} (\dd \xi) \, \ee^{f_2(\xi)}\,,
\end{equation}
and
\begin{equation} \label{def_f2}
f_2(\xi) \deq \int_0^{\infty} \dd t\, \tr \biggl(\frac{1}{t-\Delta/2+\kappa-\ii \xi}-\frac{1}{t-\Delta/2+\kappa}
-\frac{1}{t-\Delta/2+\kappa} \,\ii\xi \,\frac{1}{t-\Delta/2+\kappa}\biggr)\,.
\end{equation}
Here we regard $\xi$ as a multiplication operator.
\item
The function $f_2(\xi)$ is finite and has nonpositive real part.
\end{enumerate}
\end{proposition}

The rest of this subsection is devoted to the proof of Proposition \ref{prop:Z_field_Wick_ordered}. For any operator $\ops$ on $\cal H$ we denote by $\ops \vert_K$ the restriction of $\ops$ to the range of $\cal P_K$ from \eqref{def_Pk}, i.e.\ $\ops\vert_K$ is the $(K+1) \times (K+1)$ matrix $(\scalar{u_k}{\ops u_l})_{k,l = 0}^K$. Moreover, we abbreviate $\ops_K \deq \cal P_K \ops \cal P_K$. Thus, $\tr \ops\vert_K = \tr \ops_K = \tr (\cal P_K \ops)$, and more generally $\tr (f(\ops \vert_K)) = \tr (\cal P_K f(\ops_K))$ for a function $f$. Recall that $h = \kappa - \Delta / 2$ and that $\cal P_K$ commutes with $h$.

We also recall H\"older's inequality for Schatten spaces: if $1 \leq p,q,r \leq \infty$ satisfy $\frac{1}{r} = \frac{1}{p} + \frac{1}{q}$ then $\norm{\ops \tilde \ops}_{\fra S^r} \leq \norm{\ops}_{\fra S^p} \norm{\tilde \ops}_{\fra S^q}$.

We record the following estimates.

\begin{lemma} \label{lem:resolvent_tools}
\begin{enumerate}[label=(\roman*)]
\item
For any $s < -1/2$ we have $\tr h^{s-1} < \infty$.
\item
For any self-adjoint operator $\ops$ on $\cal H$ and any $p \in [1,\infty]$ we have
\begin{equation*}
\normbb{\frac{1}{t + h - \ii \ops}}_{\fra S^p} \leq \normbb{\frac{1}{t + h}}_{\fra S^p}\,.
\end{equation*}
\end{enumerate}
\end{lemma}
\begin{proof}
Part (i) follows immediately from the definition of $h$. For part (ii), abbreviate $\alpha = t + h$ and write
\begin{equation*}
\frac{1}{\alpha - \ii \ops} = \alpha^{-1/2} \frac{1}{1 - \ii \alpha^{-1/2} \ops \alpha^{-1/2}} \alpha^{-1/2}\,.
\end{equation*}
Using H\"older's inequality yields
\begin{equation*}
\normbb{\frac{1}{\alpha - \ii \ops}}_{\fra S^p} = \normb{\alpha^{-1/2}}_{\fra S^{2p}} \normbb{\frac{1}{1 - \ii \alpha^{-1/2} \ops \alpha^{-1/2}}}_{\fra S^\infty} \normb{\alpha^{-1/2}}_{\fra S^{2p}} \leq \norm{\alpha^{-1}}_{\fra S^p}\,. \qedhere
\end{equation*}
\end{proof}

\begin{proof}[Proof of Proposition \ref{prop:Z_field_Wick_ordered}]
Let us start with the claim (ii). Since $\xi$ is $\mu_{v_\eta}$-almost surely bounded, a simple resolvent expansion combined with Lemma \ref{lem:resolvent_tools} (ii) implies that $\frac{1}{t + h -\ii \xi}-\frac{1}{t + h}$ is trace class with trace
\begin{multline*}
\tr \pbb{\frac{1}{t + h -\ii \xi}-\frac{1}{t + h}} = \int_0^\infty \dd r \, \ee^{-(t + \kappa) r} \tr \pb{\ee^{(\Delta / 2 + \ii \xi) r} - \ee^{\Delta r / 2}}
\\
= \int_0^\infty \dd r \, \ee^{-(t + \kappa) r} \int \dd x \int \bb W^{r,0}_{x,x}(\dd \omega) \, \pB{\ee^{\ii \int_0^r \dd t \, \xi(\omega(t))} - 1}\,,
\end{multline*}
which has a nonpositive real part. Since the third term of \eqref{def_f2} is purely imaginary, we conclude that the integrand of \eqref{def_f2} has a nonpositive real part. Moreover, by a resolvent expansion, Lemma \ref{lem:resolvent_tools} (ii), and H\"older's inequality, \eqref{def_f2} is bounded in absolute value by
\begin{equation} \label{res_exp_3}
\int_0^\infty \dd t \, \norm{\xi}_{L^\infty}^2\normbb{\frac{1}{t + h}}_{\fra S^3}^3
= \int_0^\infty \dd t \,  \norm{\xi}_{L^\infty}^2 \tr \pbb{\frac{1}{t+h}}^3 = \frac{1}{2} \norm{\xi}_{L^\infty}^2 \tr h^{-2} < \infty\,.
\end{equation}

Next, recall that $W^v$ and $W^{v_\eta}$ are nonnegative, since $v$ and $v_\eta$ are of positive type (recall \eqref{Fourier_regularization} and that the function $\varphi$ is nonnegative). Since $\lim_{\eta \to 0} \norm{v - v_\eta}_{L^\infty} = 0$ by continuity of $v$, we deduce from Lemma \ref{lem:constr_W} (ii) that
\begin{equation*}
\zeta = \lim_{\eta \to 0} \int \mu_{h^{-1}}(\dd \phi) \, \ee^{-W^{v_\eta}}\,.
\end{equation*}
Moreover, by Lemma \ref{lem:constr_W} (i), we have for any $\eta > 0$,
\begin{equation*}
\int \mu_{h^{-1}}(\dd \phi) \, \ee^{-W^{v_\eta}} = \lim_{K \to \infty} \int \mu_{h^{-1}}(\dd \phi) \, \ee^{-W_K^{v_\eta}}\,.
\end{equation*}

Next, we apply the Hubbard-Stratonovich transformation \eqref{Hubbard_Stratonovich_formula} to the Gaussian measure $\mu_{v_\eta}$ with $f(x) = \abs{\cal P_K \phi(x)}^2 - \varrho_K(x)$, which yields
\begin{equation*}
\int \mu_{h^{-1}}(\dd \phi) \, \ee^{-W_K^{v_\eta}} = 
\int \mu_{v_\eta}(\dd \xi) \, \ee^{- \ii \scalar{\xi}{\varrho_K}}
\int \mu_{h^{-1}}(\dd \phi) \, \ee^{\ii \scalar{\cal P_K \phi}{\xi \cal P_K \phi}}\,.
\end{equation*}
We recall from the definitions \eqref{phi_series} and \eqref{def_Pk} that the integral over $\cal P_K \phi$ is a finite-dimensional complex Gaussian integral with covariance $(h^{-1})_K$.
Thus, Lemma \ref{complex_Gaussian} yields
\begin{equation} \label{41_1}
\int \mu_{h^{-1}}(\dd \phi) \, \ee^{\ii \scalar{\cal P_K \phi}{\xi \cal P_K \phi}} = \frac{\det \pb{(h - \ii \xi)\vert_K}^{-1}}{\det \pb{h\vert_K}^{-1}}
= \exp \qBB{- \tr \pBB{ \log \pb{\p{h - \ii \xi}\vert_K} - \log \p{h\vert_K}}}
\,.
\end{equation}
Using the integral representation \eqref{integral_identity} we therefore get
\begin{equation} \label{41_2}
\int \mu_{h^{-1}}(\dd \phi) \, \ee^{\ii \scalar{\cal P_K \phi}{\xi \cal P_K \phi}}
= \exp \hBB{\int_0^\infty \dd t \,  \tr \qBB{\frac{1}{(t + h - \ii \xi)\vert_K} - \frac{1}{(t + h)\vert_K}}}\,.
\end{equation}

What remains, therefore, is to show that for every continuous $\xi \col \Lambda \to \R$ we have
\begin{equation}
f_2(\xi) = \lim_{K \to \infty} \pBB{\int_0^\infty \dd t \,  \tr \qBB{\frac{1}{(t + h - \ii \xi)\vert_K} - \frac{1}{(t + h)\vert_K}} - \ii \scalar{\xi}{\varrho_K}}\,.
\end{equation}
By \eqref{def_varrho_K} we have
\begin{equation} \label{41_3}
\scalar{\xi}{\varrho_K} = \tr \pbb{ \xi \frac{\cal P_K}{h}} = \int_0^\infty \dd t \, \tr \pbb{ \xi \, \frac{\cal P_K}{(t + h)^2}}
= 
\int_0^\infty \dd t \, \tr \pbb{\frac{1}{t + h} \,  \xi_K \, \frac{1}{t + h}}\,.
\end{equation}
We conclude that it remains to show that
\begin{equation} \label{41_4}
f_2(\xi) = \lim_{K \to \infty} \int_0^\infty \dd t \,  \tr \qBB{ \cal P_K \pBB{\frac{1}{(t + h - \ii \xi_K)} - \frac{1}{t + h} - \frac{1}{t + h} \,  \ii \xi_K \, \frac{1}{t + h}}}\,.
\end{equation}
To that end, we perform a resolvent expansion to get
\begin{equation} \label{res_exp}
\frac{1}{(t + h - \ii \xi_K)} - \frac{1}{t + h} - \frac{1}{t + h} \,  \ii \xi_K \, \frac{1}{t + h}
= \frac{1}{t + h} \,  \ii \xi_K \, \frac{1}{t + h} \ii \xi_K \, \frac{1}{t + h - \ii \xi_K}\,.
\end{equation}

By H\"older's inequality, using that $\norm{\xi_K}_{\fra S^\infty} \leq \norm{\xi}_{L^\infty}$, we find from \eqref{res_exp} and Lemma \ref{lem:resolvent_tools} (ii) that
\begin{equation*}
\normbb{\frac{1}{(t + h - \ii \xi_K)} - \frac{1}{t + h} - \frac{1}{t + h} \,  \ii \xi_K \, \frac{1}{t + h}}_{\fra S^1} \leq \norm{\xi}_{L^\infty}^2 \normbb{\frac{1}{t + h}}_{\fra S^3}^3\,,
\end{equation*}
which is integrable by \eqref{res_exp_3}.

By dominated convergence and a resolvent expansion of the integrand of \eqref{def_f2}, to show \eqref{41_4} it therefore suffices to show that for all $t \geq 0$ we have
\begin{equation} \label{trace_conv}
\lim_{K \to \infty} \tr \pBB{\frac{1}{t + h} \,  \ii \xi_K \, \frac{1}{t + h} \ii \xi_K \, \frac{1}{t + h - \ii \xi_K}} = \tr  \pBB{\frac{1}{t + h} \,  \ii \xi \, \frac{1}{t + h} \ii \xi \, \frac{1}{t + h - \ii \xi}}\,.
\end{equation}
To that end, we estimate
\begin{equation} \label{trace_conv1}
\normbb{\frac{1}{t + h} \xi_K - \frac{1}{t+h} \xi}_{\fra S^2} \leq \norm{\xi}_{L^\infty} \normbb{\frac{1 - \cal P_K}{t + h}}_{\fra S^2} = \norm{\xi}_{L^\infty} \pBB{\sum_{k > K} \frac{1}{\lambda_k^2}}^{1/2} \to 0
\end{equation}
as $K \to \infty$. Moreover,
\begin{multline} \label{trace_conv2}
\normbb{\frac{1}{t + h - \ii \xi_K} - \frac{1}{t + h - \ii \xi}}_{\fra S^\infty} = \normbb{\frac{1}{t + h - \ii \xi_K} (\xi - \xi_K) \frac{1}{t + h - \ii \xi}}_{\fra S^\infty}
\\
\leq \|\xi\|_{L^{\infty}} \normbb{\frac{1}{t + h - \ii \xi_K} (1 - \cal P_K)}_{\fra S^\infty} \normbb{\frac{1}{t + h - \ii \xi}}_{\fra S^\infty} +
\|\xi\|_{L^{\infty}}  \normbb{\frac{1}{t + h - \ii \xi_K}}_{\fra S^\infty} \normbb{(1 - \cal P_K) \frac{1}{t + h - \ii \xi}}_{\fra S^\infty}\,,
\end{multline}
which tends to zero as $K \to \infty$, since by a resolvent expansion
\begin{equation*}
\normbb{\frac{1}{t + h - \ii \xi_K} (1 - \cal P_K)}_{\fra S^\infty} \leq \pbb{1 + \normbb{\frac{1}{t + h - \ii \xi_K} \xi_K }_{\fra S^\infty}}\normbb{\frac{1 - \cal P_K}{t + h}}_{\fra S^\infty} \to 0
\end{equation*}
as $K \to \infty$, by Lemma \ref{lem:resolvent_tools} (ii). By analogous arguments, the second term in \eqref{trace_conv2} tends to zero as $K \to \infty$.

Using \eqref{trace_conv1} and \eqref{trace_conv2}
it is now easy to deduce \eqref{trace_conv}.
\end{proof}

\subsection{Correlation functions}
\label{functional_integral_classical_gamma}

In this subsection we derive the functional integral representation for the correlation functions \eqref{def_gamma} and \eqref{gamma_renormalized}. We take over the notations from Section \ref{functional_integral_classical_Z}.
Analogously to \eqref{Gamma_eta} and \eqref{def_eh_gamma_eta}, for $\eta > 0$ we define
\begin{equation} \label{def_gamma_eta}
\gamma_{p,\eta} \deq \frac{p!}{\zeta_\eta} P_p \int \mu_{v_\eta}(\dd \xi) \, \ee^{f_2(\xi)}\, \pbb{\frac{1}{\kappa - \Delta/2 - \ii \xi}}^{\otimes p}\,.
\end{equation}
as well as its Wick-ordered version
\begin{equation} \label{def_gamma_eta_hat}
\wh \gamma_{p,\eta} \deq \frac{p!}{\zeta_\eta} P_p \int \mu_{v_\eta}(\dd \xi) \, \ee^{f_2(\xi)}\, \pbb{\frac{1}{\kappa - \Delta/2 - \ii \xi} - \frac{1}{\kappa - \Delta/2}}^{\otimes p}\,.
\end{equation}
The main result of this subsection is the following.

\begin{proposition} \label{prop:gamma_regularization}
For any $p \in \N^*$ we have
\begin{equation*}
\lim_{\eta \to 0} \norm{\wh \gamma_{p,\eta} - \wh \gamma_p}_{L^\infty} = 0\,,
\end{equation*}
where $\wh \gamma_p$ was defined in \eqref{def_gamma_hat}.
\end{proposition}

The rest of this subsection is devoted to the proof of Proposition \ref{prop:gamma_regularization}.

We recall the following notion of convergence of operators.
Let $p \in \N^*$ and $(\ops_\eta)_{\eta > 0}$ be a family of operators on $P_p \cal H^{\otimes p}$. We say that $\ops_\eta$ converges to the operator $\ops$ in the \emph{weak operator topology} if for all $F,G \in P_p \cal H^{\otimes p}$ we have
\begin{equation*}
\lim_{\eta \to 0} \scalarb{F}{(\ops_\eta - \ops) G} = 0\,.
\end{equation*}

\begin{proposition} \label{prop:gamma_weak}
For all $p \in \N^*$, as $\eta \to 0$, the operator $\gamma_{p, \eta}$ converges to $\gamma_{p}$ in the weak operator topology.
\end{proposition}
\begin{proof}
The proof is analogous to that of Proposition \ref{prop:Z_field_Wick_ordered}. Note first that, by \eqref{def_gamma_eta}, Lemma \ref{lem:resolvent_tools} (ii), Proposition \ref{prop:Z_field_Wick_ordered} (ii), and $\|P_p\|_{\fra S^\infty} \leq 1$, we have
\begin{equation*}
\norm{\gamma_{p,\eta}}_{\fra S^\infty} \leq \frac{p!}{\zeta_\eta} \frac{1}{\kappa^p}\,,
\end{equation*}
so that, by Proposition \ref{prop:Z_field_Wick_ordered} (i), $\norm{\gamma_{p,\eta}}_{\fra S^\infty}$ is uniformly bounded for all $\eta > 0$. Thus, by a density argument, it suffices to prove that for all $K_* \in \N^*$ and all $f_1, \dots, f_p, g_1, \dots, g_p \in \cal P_{K_*} \cal H$, the quantity $\scalarb{f_1 \otimes \cdots \otimes f_p}{(\gamma_{p,\eta} - \gamma_p)\,  g_1 \otimes \cdots \otimes g_p}$ tends to $0$ as $\eta \to 0$.

By the definition \eqref{def_gamma}, for $f_1, \dots, f_p, g_1, \dots, g_p \in \cal P_{K_*} \cal H$, we have
\begin{equation}
\scalarb{f_1 \otimes \cdots \otimes f_p}{\zeta \gamma_p\,  g_1 \otimes \cdots \otimes g_p} = \int \mu_{h^{-1}}(\dd \phi) \, \ee^{-W^v} \, \prod_{i = 1}^p \scalar{f_i}{ \phi} \scalar{\phi}{g_i}\,.
\end{equation}
Using that for $x,y \geq 0$ we have $\abs{\ee^{-x} - \ee^{-y}} \leq \abs{x - y}$ and Cauchy-Schwarz, we obtain
\begin{multline*}
\absbb{\int \mu_{h^{-1}}(\dd \phi) \, \pb{\ee^{-W^v} - \ee^{-W^{v_\eta}_K}} \, \prod_{i = 1}^p \scalar{f_i}{ \phi} \scalar{\phi}{g_i}}
\\
\leq \pbb{\int \mu_{h^{-1}}(\dd \phi) \abs{W^v - W^{v_\eta}_K}^2}^{1/2}
\pbb{\int \mu_{h^{-1}}(\dd \phi) \prod_{i = 1}^p \abs{\scalar{f_i}{ \phi}}^2 \abs{\scalar{\phi}{g_i}}^2}^{1/2}
\end{multline*}
The second factor is finite -- in fact equal to $\scalar{\psi}{\gamma_{2p}^0 \psi}$ with $\psi = \bigotimes_{i = 1}^p f_i \otimes g_i$. By Lemma \ref{lem:constr_W}, we therefore deduce that
\begin{equation} \label{zgamma_conv}
\scalarb{f_1 \otimes \cdots \otimes f_p}{\zeta \gamma_p\,  g_1 \otimes \cdots \otimes g_p} = \lim_{\eta \to 0} \lim_{K \to \infty} \int \mu_{h^{-1}}(\dd \phi) \, \ee^{-W^{v_\eta}_K} \, \prod_{i = 1}^p \scalar{f_i}{ \phi} \scalar{\phi}{g_i}\,.
\end{equation}

Next, for $K \geq K_*$, the Hubbard-Stratonovich transformation \eqref{Hubbard_Stratonovich_formula} with $f(x) = \abs{\cal P_K \phi(x)}^2 - \varrho_K(x)$ yields
\begin{multline*}
\int \mu_{h^{-1}}(\dd \phi) \, \ee^{-W^{v_\eta}_K} \, \prod_{i = 1}^p \scalar{f_i}{ \phi} \scalar{\phi}{g_i}
\\
= \int \mu_{v_\eta}(\dd \xi) \, \ee^{- \ii \scalar{\xi}{\varrho_K}}
\int \mu_{h^{-1}}(\dd \phi) \, \ee^{\ii \scalar{\cal P_K \phi}{\xi \cal P_K \phi}} \prod_{i = 1}^p \scalar{f_i}{\cal P_K \phi} \scalar{\cal P_K \phi}{g_i} \,,
\end{multline*}
where we used that $f_i = \cal P_{K} f_i$ (and similarly for $g_i$), since by assumption $f_i = \cal P_{K_*} f_i$ and $\cal P_K \cal P_{K_*} = \cal P_{K_*}$. The integral over $\cal P_K \phi$ is a finite-dimensional complex Gaussian integral with covariance $(h^{-1})_K$, and using the Wick rule from Lemma \ref{complex_Wick_theorem} yields
\begin{equation*}
\int \mu_{h^{-1}}(\dd \phi) \, \ee^{\ii \scalar{\cal P_K \phi}{\xi \cal P_K \phi}} \prod_{i = 1}^p \scalar{f_i}{\cal P_K \phi} \scalar{\cal P_K \phi}{g_i}  = \frac{\det \pb{(h - \ii \xi)\vert_K}^{-1}}{\det \pb{h\vert_K}^{-1}} \sum_{\pi \in S_p} \scalarbb{f_i}{\frac{1}{(h - \ii \xi)_K} g_{\pi(i)}}\,.
\end{equation*}
Proceeding as in the proof of Proposition \ref{prop:Z_field_Wick_ordered} (see \eqref{41_1}, \eqref{41_2}, \eqref{41_3}), we therefore find
\begin{multline*}
\int \mu_{h^{-1}}(\dd \phi) \, \ee^{-W^{v_\eta}_K} \, \prod_{i = 1}^p \scalar{f_i}{ \phi} \scalar{\phi}{g_i}
= \int \mu_{v_\eta}(\dd \xi) \, \sum_{\pi \in S_p} \scalarbb{f_i}{\frac{1}{(h - \ii \xi)_K} g_{\pi(i)}}
\\
\times \exp \hBB{\int_0^\infty \dd t \,  \tr \qBB{ \cal P_K \pBB{\frac{1}{(t + h - \ii \xi_K)} - \frac{1}{t + h} - \frac{1}{t + h} \,  \ii \xi_K \, \frac{1}{t + h}}}} \,,
\end{multline*}
where the factor on the last line is bounded in absolute value by one. From the proof of Proposition \ref{prop:Z_field_Wick_ordered}, see \eqref{41_4}, we know that as $K \to \infty$ the last line converges to $\ee^{f_2(\xi)}$ for all continuous $\xi$.

Next, we note that, by assumption on $K$ and $f_i,g_i$, we have
\begin{equation*}
\scalarbb{f_i}{\frac{1}{(h - \ii \xi)_K} g_{\pi(i)}} = \scalarbb{f_i}{\frac{1}{h - \ii \xi_K} g_{\pi(i)}}\,,
\end{equation*}
so that Lemma \ref{lem:resolvent_tools} (ii) implies
\begin{equation*}
\absbb{\scalarbb{f_i}{\frac{1}{(h - \ii \xi)_K} g_{\pi(i)}}} \leq \norm{f_i}_{L^2} \norm{g_{\pi(i)}}_{L^2} \frac{1}{\kappa}\,.
\end{equation*}
Using \eqref{trace_conv2} and dominated convergence, we therefore deduce that
\begin{align*}
\lim_{K \to \infty} \int \mu_{h^{-1}}(\dd \phi) \, \ee^{-W^{v_\eta}_K} \, \prod_{i = 1}^p \scalar{f_i}{ \phi} \scalar{\phi}{g_i}
&= \int \mu_{v_\eta}(\dd \xi) \, \ee^{f_2(\xi)} \sum_{\pi \in S_p} \scalarbb{f_i}{\frac{1}{h - \ii \xi} g_{\pi(i)}}
\\
&= \scalarb{f_1 \otimes \cdots \otimes f_p}{\zeta_\eta \gamma_{p,\eta}\,  g_1 \otimes \cdots \otimes g_p}\,.
\end{align*}
The claim now follows from \eqref{zgamma_conv} and Proposition \ref{prop:Z_field_Wick_ordered} (i).
\end{proof}

\begin{corollary} \label{cor:hat_gamma}
For all $p \in \N^*$, as $\eta \to 0$, the operator $\wh \gamma_{p, \eta}$ converges to $\wh \gamma_{p}$ in the weak operator topology.
\end{corollary}

\begin{proof}
A straightforward calculation using the binomial theorem, \eqref{def_gamma_eta}--\eqref{def_gamma_eta_hat}, and $\gamma_p^0 = p! P_p (h^{-1})^{\otimes p}$ shows that
\begin{equation*}
\wh \gamma_{p,\eta} = \sum_{k = 0}^p \binom{p}{k}^2 (-1)^{p - k} \, P_p \pb{\gamma_{k,\eta} \otimes \gamma_{p-k}^0} P_p\,,
\end{equation*}
and the claim follows from Proposition \ref{prop:gamma_weak}.
\end{proof}

\begin{lemma} \label{lem:gamma_cauchy}
The family $(\wh \gamma_{p,\eta})_{\eta > 0}$ is Cauchy with respect to $\norm{\cdot}_{L^\infty}$.
\end{lemma}
\begin{proof}
Let $\epsilon > 0$. Using Proposition \ref{prop:main_conv} below, choose $\nu > 0$ such that $\normb{\nu^p \, \wh \Gamma_{p,\eta,\nu} - \wh \gamma_{p,\eta}}_{L^\infty} \leq \epsilon$ for all $\eta$. Then for $\eta, \tilde \eta > 0$ we estimate
\begin{align*}
\norm{\wh \gamma_{p,\eta} - \wh \gamma_{p, \tilde \eta}}_{L^\infty} &\leq \norm{\wh \gamma_{p,\eta} - \nu^p \, \wh \Gamma_{p, \eta, \nu}}_{L^\infty} + \norm{\nu^p \, \wh \Gamma_{p,\eta, \nu} - \nu^p \, \wh \Gamma_{p, \tilde \eta, \nu}}_{L^\infty} + \norm{\nu^p \, \wh \Gamma_{p, \tilde \eta, \nu} - \wh \gamma_{p, \tilde \eta}}_{L^\infty}
\\
&\leq 2 \epsilon + \norm{\nu^p \, \wh \Gamma_{p,\eta, \nu} - \nu^p \, \wh \Gamma_{p, \tilde \eta, \nu}}_{L^\infty}\,.
\end{align*}
For the above $\nu > 0$, by Proposition \ref{prop:gamma_hat} we have
\begin{equation*}
\lim_{\eta, \tilde \eta \to 0} \norm{\nu^p \, \wh \Gamma_{p,\eta, \nu} - \nu^p \, \wh \Gamma_{p, \tilde \eta, \nu}}_{L^\infty} = 0\,,
\end{equation*}
and the claim follows.
\end{proof}

Now Proposition \ref{prop:gamma_regularization} follows immediately from  Corollary \ref{cor:hat_gamma} and Lemma \ref{lem:gamma_cauchy}.

\begin{remark}
Note that our proof of Lemma \ref{lem:gamma_cauchy} (and hence also of Proposition \ref{prop:gamma_regularization}) we used Proposition \ref{prop:main_conv}, which involves the quantum many-body system. We do this because it provides the shortest argument (since Proposition \ref{prop:main_conv} is needed anyway for our main result), effectively using the quantum many-body system as a convenient regularization of the classical field theory. If Lemma \ref{lem:gamma_cauchy} were our only goal, a direct approach without invoking the quantum many-body system would also work, but its proof would share some ideas with that of Proposition \ref{prop:main_conv}.
\end{remark}

\section{Mean-field limit} \label{sec:mf}

Throughout this section we work in the mean-field scaling
\begin{equation} \label{mean_field_scaling}
\lambda = \nu^2\,.
\end{equation}
The main result of this section is the following.

\begin{proposition} \label{prop:main_conv}
Suppose that Assumptions \ref{ass:torus} and \ref{ass:v} hold. Suppose \eqref{mean_field_scaling}. Let $p \in \N^*$ and recall the definitions \eqref{def_eh_gamma_eta} and \eqref{def_gamma_eta_hat}. Then we have, uniformly in $\eta > 0$,
\begin{equation*}
\lim_{\nu \to 0} \normb{\nu^p \, \wh \Gamma_{p,\eta,\nu} - \wh \gamma_{p,\eta}}_{L^\infty} = 0\,.
\end{equation*}
\end{proposition}

The rest of this section is devoted to the proof of Proposition \ref{prop:main_conv}. We always make the Assumptions \ref{ass:torus} and \ref{ass:v} and we suppose \eqref{mean_field_scaling}.

A simple but important observation used throughout our proof is that the field $\xi$ from \eqref{cov_v_eta} can be represented as a time average of the field $\sigma$ from \eqref{cov_C_eta}. To that end, we define
\begin{equation}
\label{sigma(x)_definition}
\ang{\sigma}(x) \deq \frac{1}{\nu} \int_0^\nu \dd \tau \, \sigma(\tau,x)\,.
\end{equation}
We record the following result, which is an immediate consequence of \eqref{mean_field_scaling}.
\begin{lemma} \label{lem:sigma_zeta}
Let $\eta > 0$. If $\sigma$ has law $\mu_{\cal C_\eta}$ then $\ang{\sigma}$ has law $\mu_{v_\eta}$.
\end{lemma}

\subsection{Partition function}
In this subsection we show the following result.
\begin{proposition} \label{prop:conv_Z}
Uniformly in $\eta>0$ we have $\lim_{\nu \to 0} \cal Z_{\eta} = \zeta_\eta$.
\end{proposition}
The rest of this subsection is devoted to the proof of Proposition \ref{prop:conv_Z}. 
Throughout, we adopt the following conventions. For a path $\omega \in \Omega^{\tau, \tilde \tau}$ we abbreviate 
\begin{equation*}
\int_{\tilde \tau}^\tau \dd s \, f(\omega(s)) \equiv \int \dd s \,  f(\omega(s))\,
\end{equation*} 
since the integration bounds are determined by the path $\omega$. Moreover, we set $\bb W^{\tau, \tilde \tau} \deq 0$ if $\tau \leq \tilde \tau$. Finally, use use the notations $\f x = (x_1, x_2, x_3) \in \Lambda^3$, $\f \tau = (\tau_1, \tau_2, \tau_3) \in [0,\nu]^3$, $\f r = (r_1, r_2, r_3) \in [0,\infty)^3$, and $\abs{\f r} = r_1 + r_2 + r_3$.

We begin by deriving the following explicit formulas for the functions $F_2$ from \eqref{def_F2} and $f_2$ from \eqref{def_f2}.

\begin{lemma} \label{lem:Ff_rep}
We have the representations
\begin{multline} \label{F2_rep}
F_2(\sigma) = - \sum_{\f r \in (\nu \N)^3} \frac{\ind{\abs{\f r} > 0} \, \ee^{-\kappa \abs{\f r}}}{\abs{\f r}} \int_{[0,\nu]^3} \dd \f \tau \int \dd \f x \,
\sigma(\tau_2, x_2) \, \sigma(\tau_3, x_3)
\\
\times  \int \bb W_{x_1, x_3}^{\tau_1 + r_3, \tau_3}(\dd \omega_3) \, \bb W_{x_3, x_2}^{\tau_3 + r_2, \tau_2}(\dd \omega_2) \, \bb W_{x_2, x_1}^{\tau_2 + r_1, \tau_1}(\dd \omega_1) \, \ee^{\ii \int \dd s \, \sigma([s]_\nu, \omega_1(s))}
\end{multline}
and
\begin{multline} \label{f2_rep}
f_2(\xi)
= - \int_{[0,\infty)^3} \dd \f r \, \frac{\ee^{- \kappa \abs{\f r}}}{\abs{\f r}} \, \int \dd \f x \, \xi(x_2) \, \xi(x_3)
\\
\times \int \bb W_{x_1, x_3}^{r_3, 0}(\dd \omega_3) \, \bb W_{x_3, x_2}^{r_2, 0}(\dd \omega_2) \, \bb W_{x_2, x_1}^{r_1, 0}(\dd \omega_1) \, \ee^{\ii \int \dd s \, \xi(\omega_1(s))}\,.
\end{multline}
\end{lemma}

\begin{proof}
From \eqref{def_F2} and using that $K(-\kappa + \ii \sigma) = K(-\kappa) - \ii \sigma$, a resolvent expansion yields
\begin{equation*}
F_2(\sigma) = - \int_0^\infty \dd t \, \tr \pbb{\frac{1}{K(-t -\kappa)} \sigma \frac{1}{K(-t -\kappa)} \sigma \frac{1}{K(-t -\kappa + \ii \sigma)} }\,.
\end{equation*}
We find from Lemma \ref{lem:G_K} and \eqref{W_FK} that
\begin{multline*}
F_2(\sigma) = - \int_0^\infty \dd t \int_{[0,\nu]^3} \dd \f \tau \int \dd \f x \sum_{\f r \in (\nu \N)^3}
\ee^{-(t + \kappa) \abs{\f r}} \, \sigma(\tau_2, x_2) \, \sigma(\tau_3, x_3)
\\
\times  \int \bb W_{x_1, x_3}^{\tau_1 + r_3, \tau_3}(\dd \omega_3) \, \bb W_{x_3, x_2}^{\tau_3 + r_2, \tau_2}(\dd \omega_2) \, \bb W_{x_2, x_1}^{\tau_2 + r_1, \tau_1}(\dd \omega_1) \, \ee^{\ii \int \dd s \, \sigma([s]_\nu, \omega_1(s))}\,.
\end{multline*}
We may perform the integral over $t$, observing that the summand is nonzero only if $\abs{\f r} > 0$, which yields \eqref{F2_rep}.

Next, from \eqref{def_f2}, a resolvent expansion and the identity $1/h = \int_0^\infty \dd r \, \ee^{-rh}$ for $h > 0$ yields
\begin{equation*}
f_2(\xi) = - \int_0^\infty \dd t \, \int_{[0,\infty)^3} \dd \f r \, \tr \pB{\ee^{-(t + \kappa - \Delta/2) r_3}  \xi \ee^{-(t + \kappa - \Delta/2) r_2} \xi \ee^{-(t + \kappa - \Delta/2 - \ii \xi) r_1}}\,,
\end{equation*}
and \eqref{f2_rep} follows by applying the Feynman-Kac formula from Lemma \ref{FK_continuous} and integrating over $t$.
\end{proof}

From \eqref{wtZ_F2}, \eqref{def_z_nu}, and Lemma \ref{lem:sigma_zeta}, we get
\begin{equation*}
{\cal Z}_\eta - \zeta_\eta = \int \mu_{\cal C_\eta}(\dd \sigma) \, \pB{\ee^{F_2(\sigma)} - \ee^{f_2(\ang{\sigma})}}\,.
\end{equation*}
Since $\re F_2, \re f_2 \leq 0$ (see Proposition \ref{prop:fif_renorm} (iii) and Proposition \ref{prop:Z_field_Wick_ordered} (ii)), it follows that $|\ee^{F_2(\sigma)} - \ee^{f_2(\ang{\sigma})}|  \leq |F_2(\sigma) - f_2(\ang{\sigma})|$. By using the Cauchy-Schwarz inequality and the fact that $\mu_{\cal C_\eta}$ is a probability measure, we deduce that Proposition \ref{prop:conv_Z} follows if we show the following result.

\begin{lemma} \label{Ff_L2conv}
As $\nu \to 0$, uniformly in $\eta > 0$, we have
\begin{equation*}
\int \mu_{\cal C_\eta}(\dd \sigma) \, \absb{F_2(\sigma) - f_2(\ang{\sigma})}^2 \longrightarrow 0\,.
\end{equation*}
\end{lemma}

To prove Lemma \ref{Ff_L2conv}, we write
\begin{multline} \label{Ff_exp}
\int \mu_{\cal C_\eta}(\dd \sigma) \absb{F_2(\sigma) - f_2(\ang{\sigma})}^2
\\
= \int \mu_{\cal C_\eta}(\dd \sigma) \pB{\ol{f_2(\ang{\sigma})} f_2(\ang{\sigma}) + \ol{F_2(\sigma)} F_2(\sigma) - \ol{f_2(\ang{\sigma})} F_2(\sigma) - \ol{F_2(\sigma)} f_2(\ang{\sigma})}\,.
\end{multline}
We shall consider the four terms of \eqref{Ff_exp} individually, and show that they all converge to the first term as $\nu \to 0$, uniformly in $\eta$. 

We begin by analysing the term
\begin{multline}
\label{int_FF'}
\int \mu_{\cal C_\eta}(\dd \sigma) \, \ol{f_2(\ang{\sigma})} f_2(\ang{\sigma}) = \int \mu_{v_\eta} (\dd \xi) \, \ol{f_2(\xi)} f_2(\xi)
=
\int_{[0,\infty)^3} \dd \f r \, \frac{\ee^{- \kappa \abs{\f r}}}{\abs{\f r}}
\int_{[0,\infty)^3} \dd \tilde{\f r} \, \frac{\ee^{- \kappa \abs{\tilde{\f r}}}}{\abs{\tilde{\f r}}}
\, \int \dd \f x \int \dd \tilde{\f x} 
\\
\times \int \bb W_{x_1, x_3}^{r_3, 0}(\dd \omega_3) \, \bb W_{x_3, x_2}^{r_2, 0}(\dd \omega_2) \, \bb W_{x_2, x_1}^{r_1, 0}(\dd \omega_1)
\bb W_{\tilde x_1, \tilde x_3}^{\tilde r_3, 0}(\dd \tilde \omega_3) \, \bb W_{\tilde x_3, \tilde x_2}^{\tilde r_2, 0}(\dd \tilde \omega_2) \, \bb W_{\tilde x_2, \tilde x_1}^{\tilde r_1, 0}(\dd \tilde \omega_1)
\\
\times \int \mu_{v_\eta}(\dd \xi) 
\, \xi(x_2) \, \xi(x_3) \, \xi(\tilde x_2) \, \xi(\tilde x_3)
\, \ee^{\ii \int \dd s \, \xi(\omega_1(s)) - \ii \int \dd s \, \xi(\tilde \omega_1(s))}\,,
\end{multline}
where we used \eqref{f2_rep}.

We shall apply Lemma \ref{lem:Gaussian} to the last line of \eqref{int_FF'}, using that $\mu_{v_\eta}$ is a finite-dimensional real Gaussian measure with covariance $v_\eta$ from \eqref{cov_v_eta}. The terms arising from this computation are described by the following definition.

\begin{definition}\label{def:interaction}
Let $x, \tilde x \in \Lambda$ and let $\omega \in \Omega^{\tau_1, \tilde \tau_1}$ and $\tilde \omega \in \Omega^{\tau_2, \tilde \tau_2}$ be continuous paths. We define the \emph{point-point interaction}
\begin{equation*}
(\bb V_{\eta})_{x,\tilde x} \deq \int \mu_{v_\eta}(\dd \xi) \, \xi(x) \, \xi(\tilde x) =   v_\eta(x - \tilde x)\,,
\end{equation*}
the \emph{point-path interaction}
\begin{equation*}
(\bb V_{\eta})_x(\omega) 
\deq \int \mu_{v_\eta}(\dd \xi) \, \xi(x)  \int \dd s  \, \xi(\omega(s)) =
\int \dd s \, v_\eta(x - \omega(s))\,,
\end{equation*}
and the \emph{path-path interaction}
\begin{equation*}
\bb V_{\eta}(\omega, \tilde \omega) \deq
\int \mu_{v_\eta}(\dd \xi)  \int \dd s  \, \xi(\omega(s)) \int \dd \tilde s \, \xi(\tilde \omega(\tilde s))
= 
\int \dd s \int \dd \tilde s \, v_\eta(\omega(s) - \tilde \omega(\tilde s))\,.
\end{equation*}
\end{definition}

From now on, we use the notation $x_{i,0} = x_i$ and $x_{i,1} = \tilde x_i$. Abbreviating
\begin{equation*}
A \deq \{2,3\} \times \{0,1\}
\end{equation*}
and using Lemma \ref{lem:Gaussian} with 
\begin{equation}
\label{f_f_a}
f(x) =\int \dd s \,\delta(x-\omega_1(s))-\int \dd \tilde s \,\delta(x-\tilde \omega_1(\tilde s))\,,\quad  f_a=\delta(x-x_a)\,, 
\end{equation}
we conclude that 
\begin{equation} \label{int_ff}
\int \mu_{v_\eta} (\dd \xi) \, \ol{f_2(\xi)} f_2(\xi)
=
\int_{[0,\infty)^3} \dd \f r \, \frac{\ee^{- \kappa \abs{\f r}}}{\abs{\f r}}
\int_{[0,\infty)^3} \dd \tilde{\f r} \, \frac{\ee^{- \kappa \abs{\tilde{\f r}}}}{\abs{\tilde{\f r}}}
\, I(\f r, \tilde{\f r})\,,
\end{equation}
where we defined
\begin{align}
I(\f r, \tilde{\f r}) \deq &{} \int \dd \f x \int \dd \tilde{\f x} \int \bb W_{x_1, x_3}^{r_3, 0}(\dd \omega_3) \, \bb W_{x_3, x_2}^{r_2, 0}(\dd \omega_2) \, \bb W_{x_2, x_1}^{r_1, 0}(\dd \omega_1)
\notag \\
&\times \bb W_{\tilde x_1, \tilde x_3}^{\tilde r_3, 0}(\dd \tilde \omega_3) \, \bb W_{\tilde x_3, \tilde x_2}^{\tilde r_2, 0}(\dd \tilde \omega_2) \, \bb W_{\tilde x_2, \tilde x_1}^{\tilde r_1, 0}(\dd \tilde \omega_1)
\,  \ee^{-\frac{1}{2} \pb{\bb V_{\eta}(\omega_1, \omega_1) + \bb V_{\eta}(\tilde \omega_1, \tilde \omega_1) - 2 \bb V_{\eta}(\omega_1, \tilde \omega_1)}}
\notag \\ \label{I(r,tilde_r)}
&\times \sum_{\Pi \in \fra M(A)} \prod_{\{a,b\} \in \Pi} (\bb V_{\eta})_{x_a, x_b} \prod_{a \in A \setminus [\Pi]} \ii \pB{(\bb V_{\eta})_{x_a}(\omega_1) - (\bb V_{\eta})_{x_a}(\tilde \omega_1)}\,.
\end{align}

Next, we collect several direct consequences of \eqref{Fourier_regularization2} which we shall use in the sequel.
\begin{lemma} 
\label{properties_of_v_eta}
The following properties hold.
\begin{itemize}
\item[(i)]
$\|v_{\eta}\|_{L^{\infty}} \leq \|v\|_{L^{\infty}}$ for all $\eta>0$. 
\item[(ii)]
The family $(v_{\eta})_{\eta>0}$ is uniformly equicontinuous\footnote{Explicitly: for every $\epsilon > 0$ there is a $\delta > 0$ such that for all $\eta > 0$ and $x,y \in \Lambda$, if $\abs{x - y}_L < \delta$ then $\abs{v_\eta(x) - v_\eta(y)} < \epsilon$.}.
\end{itemize}
\end{lemma}
Furthermore, \eqref{Fourier_regularization2} implies that
\begin{equation}
\label{properties_of_delta_eta}
\int_{0}^{\nu} \dd \tau \, \delta_{\eta,\nu} (\tau)=1\,,\quad \delta_{\eta,\nu} \geq 0\,.
\end{equation}

\begin{lemma} \label{lem:unif_int}
\label{int_ff_convergence_lemma}
The integral on the right-hand side of \eqref{int_ff} converges absolutely, uniformly in $\eta$.
\end{lemma}
\begin{proof}
With $f$ chosen as in \eqref{f_f_a} we have
\begin{equation}
\label{int_ff_convergence_lemma_1}
\bb V_{\eta}(\omega_1, \omega_1) + \bb V_{\eta}(\tilde \omega_1, \tilde \omega_1) - 2 \bb V_{\eta}(\omega_1, \tilde \omega_1)= \int \dd x \, \dd y \, f(x) v_\eta(x - y) f(y) \geq 0\,,
\end{equation}
From Definition \ref{def:interaction} and Lemma \ref{properties_of_v_eta} we have that for all $a,b \in A$
\begin{equation}
\label{int_ff_convergence_lemma_2}
\bigl|(\bb V_{\eta})_{x_a, x_b}\bigr| \leq \|v\|_{L^{\infty}}\,, \quad \bigl|(\bb V_{\eta})_{x_a}(\omega_1)\bigr| \leq r_1 \|v\|_{L^{\infty}}\,,\quad \bigl|(\bb V_{\eta})_{x_a}(\tilde \omega_1)\bigr| \leq \tilde r_1 \|v\|_{L^{\infty}}\,.
\end{equation}
The claim now follows from \eqref{int_ff_convergence_lemma_1}--\eqref{int_ff_convergence_lemma_2} and the observation that
\begin{multline*}
\int_{[0,\infty)^3} \dd \f r \, \frac{\ee^{- \kappa \abs{\f r}}}{\abs{\f r}}
\int_{[0,\infty)^3} \dd \tilde{\f r} \, \frac{\ee^{- \kappa \abs{\tilde{\f r}}}}{\abs{\tilde{\f r}}}
\, \int \dd x_1 \,\int \dd \tilde x_1\, \int \bb W_{x_1, x_1}^{\abs{\f r}, 0}(\dd \omega)
\,\int \bb W_{\tilde x_1,\tilde x_1}^{\abs{\tilde{\f r}}, 0}(\dd \tilde \omega) (1+r_1+\tilde r_1)^4
\\
\leq C_d \,\int_{[0,\infty)^3} \dd \f r \, \frac{\ee^{- \kappa \abs{\f r}}}{\abs{\f r}}
\int_{[0,\infty)^3} \dd \tilde{\f r} \, \frac{\ee^{- \kappa \abs{\tilde{\f r}}}}{\abs{\tilde{\f r}}}
\,\pbb{1 + \frac{L^d}{\abs{\f r}^{d/2}}} \,\pbb{1 + \frac{L^d}{\abs{\tilde{\f r}}^{d/2}}}\,\bigl(1+\abs{\f r}+\abs{\tilde{\f r}}\bigr)^4 <\infty\,, 
\end{multline*}
as follows from Lemma \ref{lem:heat_estimate} and $d \leq 3$. 
\end{proof}

The expression \eqref{int_ff} is the reference quantity, to which we shall compare the three other expressions on the right-hand side of \eqref{Ff_exp}.

Proposition \ref{prop:conv_Z} follows from the two following results.

\begin{lemma} \label{lem:FF}
Uniformly in $\eta>0$,
\begin{equation*}
\lim_{\nu \to 0} \int \mu_{\cal C_\eta}(\dd \sigma) \, \ol{F_2(\sigma)} F_2(\sigma) = \int \mu_{v_\eta} (\dd \xi) \, \ol{f_2(\xi)} f_2(\xi)\,.
\end{equation*}
\end{lemma}

\begin{lemma} \label{lem:Ff}
Uniformly in $\eta>0$,
\begin{equation*}
\lim_{\nu \to 0} \int \mu_{\cal C_\eta}(\dd \sigma) \, \ol{F_2(\sigma)} f_2(\ang{\sigma}) = \int \mu_{v_\eta} (\dd \xi) \, \ol{f_2(\xi)} f_2(\xi)\,.
\end{equation*}
\end{lemma}

Before proceeding to the proofs of Lemmas \ref{lem:FF} and \ref{lem:Ff}, we extend Definition \ref{def:interaction} to positive $\nu$ as follows.
\begin{definition}\label{def:interaction_general}
Let $(\tau, x), (\tilde \tau,\tilde x) \in [0,\nu] \times \Lambda$ be space-time points, and let $\omega \in \Omega^{\tau_1, \tilde \tau_1}$ and $\tilde \omega \in \Omega^{\tau_2, \tilde \tau_2}$ be continuous paths. 
Let $\delta$ denote the $\nu$-periodic delta function.
Then we define the \emph{point-point interaction}
\begin{equation*}
(\bb V_{\eta, \nu})^{\tau, \tilde \tau}_{x,\tilde x} \deq \int \mu_{\cal C_\eta}(\dd \sigma) \, \sigma(\tau,x) \, \sigma(\tilde \tau, \tilde x) =  \nu \, \delta_{\eta, \nu}(\tau - \tilde \tau) \, v_\eta(x - \tilde x)\,,
\end{equation*}
the \emph{point-path interaction}
\begin{align*}
(\bb V_{\eta, \nu})^\tau_x(\omega) 
&\deq \int \mu_{\cal C_\eta}(\dd \sigma) \, \sigma(\tau,x)  \int_0^\nu \dd t \int \dd s \, \delta(t - [s]_\nu) \, \sigma(t, \omega(s))
\\
&=
\nu \int \dd s \,\delta_{\eta, \nu}(\tau - [s]_\nu) \,  v_\eta(x - \omega(s))\,,
\end{align*}
and the \emph{path-path interaction}
\begin{align*}
\bb V_{\eta, \nu}(\omega, \tilde \omega) &\deq 
 \int \mu_{\cal C_\eta}(\dd \sigma) \int_0^\nu \dd t \int \dd s \, \delta(t - [s]_\nu) \, \sigma(t, \omega(s)) \int_0^\nu \dd \tilde t \int \dd \tilde s \, \delta(\tilde t - [\tilde s]_\nu) \, \sigma(\tilde t, \tilde \omega(\tilde s))
\\
&=
\nu \int \dd s \int \dd \tilde s \,\delta_{\eta, \nu}([s]_\nu - [\tilde s]_\nu) \, v_\eta(\omega(s) - \tilde \omega(\tilde s))\,.
\end{align*}

\end{definition}

In the sequel the following shorthand will prove useful. We denote by $\cal E_{\eta, \nu}(\alpha)$ any function of a parameter $\alpha$, which is uniformly bounded in $(\eta,\nu,\alpha)$, and which tends to zero as $\nu \to 0$, uniformly in $\eta > 0$, for each $\alpha$ (i.e.\ $\lim_{\nu \to 0} \sup_{\eta > 0} \abs{\cal E_{\eta, \nu}(\alpha)} = 0$ for all $\alpha$).
This definition is tailored for use with dominated convergence, in proving that, for a finite measure $\mu$, 
\begin{equation*}
\lim_{\nu \to 0} \sup_{\eta > 0} \absbb{\int \mu(\dd \alpha) \, \cal E_{\eta, \nu}(\alpha)} \leq \lim_{\nu \to 0} \int \mu(\dd \alpha) \, \sup_{\eta > 0} \abs{\cal E_{\eta, \nu}(\alpha)}  = 0\,.
\end{equation*}

With this notation, we note several properties of the path-path interactions.

\begin{lemma}
\label{path-path_interaction}
Let $\omega \in \Omega^{\tau_1, \tilde \tau_1}$ and $\tilde \omega \in \Omega^{\tau_2, \tilde \tau_2}$ be continuous paths. 
\begin{itemize}
\item[(i)]
For $\nu>0$ we have
\begin{equation*}
\bb V_{\eta,\nu}(\omega, \tilde \omega)=\bb V_{\eta}(\omega, \tilde \omega)
+\cal E_{\eta, \nu}(\omega, \tilde \omega) \bigl[1+(\tau_1-\tilde \tau_1) \bigr]\,\bigl[1+(\tau_2-\tilde \tau_2) \bigr]\,.
\end{equation*}
\item[(ii)] Let $\omega' \in \Omega^{\tau'_1, \tilde \tau'_1}$ and $\tilde \omega' \in \Omega^{\tau'_2, \tilde \tau'_2}$ be continuous paths that agree with $\omega$ and $\tilde \omega$ respectively on the intersection of their domains. Furthermore, suppose
$|\tau_j-\tau'_j| \leq \nu, |\tilde \tau_j-\tilde \tau'_j| \leq \nu$ for $j=1,2$. Then we have
\begin{equation*}
\bb V_{\eta}(\omega, \tilde \omega)=\bb V_{\eta}(\omega', \tilde \omega')+O\bigl(\nu (\tau_1-\tilde \tau_1) + \nu (\tau_2-\tilde \tau_2) + \nu^2\bigr)\,.
\end{equation*}
\end{itemize}
\end{lemma}

\begin{proof}[Proof of Lemma \ref{path-path_interaction}]
We first prove (i).
Let $N \deq \lceil (\tau_2-\tilde \tau_2)/\nu \rceil$. We decompose the interval $[\tilde \tau_2, \tau_2]$ as $\bigcup_{i=1}^{N}I_i$, where
\begin{equation*}
I_i \deq
\begin{cases}
[\tilde \tau_2+(i-1)\nu, \tilde  \tau_2 + i \nu] &\mbox{if }1 \leq i \leq N-1
\\
[\tilde \tau_2+(N-1)\nu,\tau_2]  &\mbox{if }i=N\,.
\end{cases}
\end{equation*}
Furthermore, for $1 \leq i \leq N$ we let $\hat I_i \deq [\tilde \tau_2+(i-1)\nu, \tilde  \tau_2 + i \nu]$. Note that $|I_i| \leq \nu$ and $|\hat I_i|=\nu$ for all $1 \leq i \leq N$.

By \eqref{Fourier_regularization2}, we can write
\begin{multline}
\label{path-path_interaction1}
\bb V_{\eta}(\omega, \tilde \omega)
=
\int \dd s \Biggl[\sum_{i=1}^{N} \int_{I_i} \dd \tilde s \int_{\hat I_i} \dd \hat s\, v_\eta(\omega(s) - \tilde \omega(\tilde s)) \delta_{\eta,\nu}([s]_{\nu}-[\hat s]_{\nu}) \Biggr]
\\
=\int \dd s \Biggl[\sum_{i=1}^{N} \int_{I_i} \dd \tilde s \int_{I_i} \dd \hat s\, v_\eta(\omega(s) - \tilde \omega(\tilde s)) \delta_{\eta,\nu}([s]_{\nu}-[\hat s]_{\nu}) \Biggr]+O\bigl(\nu (\tau_1-\tilde \tau_1)\bigr)
\,,
\end{multline}
In order to obtain the last line we used Lemma \ref{properties_of_v_eta} (i) as well as \eqref{properties_of_delta_eta}, which implies that $\int_{\hat I_N \setminus I_N} \dd \hat s \,  \delta_{\eta,\nu}([s]_{\nu}-[\hat s]_{\nu})= O(1)$.

For $1 \leq i \leq N$ we have that $|\tilde s-\hat s| \leq \nu$ for $\tilde s, \hat s \in I_i$. Therefore,  by Lemma \ref{properties_of_v_eta} (ii) we have $v_{\eta}(\omega(s) - \tilde \omega(\tilde s))=v_{\eta}(\omega(s) - \tilde \omega(\hat s))+ \cal E_{\eta,\nu}(\omega,\tilde \omega)$. By construction of $I_i$, we deduce that the first term in \eqref{path-path_interaction1} is

\begin{equation}
\label{path-path_interaction2}
\bb V_{\eta,\nu}(\omega, \tilde \omega)+
\cal E_{\eta, \nu}(\omega, \tilde \omega) (\tau_1-\tilde \tau_1) (\tau_2-\tilde \tau_2) + \nu \cal E_{\eta, \nu}(\omega, \tilde \omega) (\tau_1-\tilde \tau_1)\,.
\end{equation}
This concludes the proof of (i).

We now prove (ii). 
By using $|\tau_1-\tau'_1| \leq \nu, |\tilde \tau_1-\tilde \tau'_1| \leq \nu$ and Lemma \ref{properties_of_v_eta} (i) and by arguing as earlier, we deduce that
\begin{equation}
\label{path-path_interaction3}
\bb V_{\eta}(\omega, \tilde \omega)
=\bb V_{\eta}(\omega', \tilde \omega)+O \bigl(\nu (\tau_2-\tilde \tau_2)\bigr)\,.
\end{equation}
By analogous arguments as for \eqref{path-path_interaction3}, it follows that 
\begin{equation}
\label{path-path_interaction4}
\bb V_{\eta}(\omega', \tilde \omega)=\bb V_{\eta}(\omega', \tilde \omega')+O \bigl(\nu (\tau'_1-\tilde \tau'_1)\bigr)=\bb V_{\eta}(\omega', \tilde \omega')+O \bigl(\nu (\tau_1-\tilde \tau_1)+\nu^2\bigr)\,.
\end{equation}
Claim (ii) follows from \eqref{path-path_interaction3}--\eqref{path-path_interaction4}.
\end{proof}

We now have the necessary tools to prove Lemma \ref{lem:FF}.
\begin{proof}[Proof of Lemma \ref{lem:FF}]
Using Lemma \ref{lem:Ff_rep} we write
\begin{equation} \label{int_FF}
\int \mu_{\cal C_\eta}(\dd \sigma) \, \ol{F_2(\sigma)} F_2(\sigma)
=
\sum_{\f r \in (\nu \N)^3} \frac{\ind{\abs{\f r} > 0} \, \ee^{-\kappa \abs{\f r}}}{\abs{\f r}}
\sum_{\tilde{\f r} \in (\nu \N)^3} \frac{\ind{\abs{\tilde{\f r}} > 0} \, \ee^{-\kappa \abs{\tilde{\f r}}}}{\abs{\tilde{\f r}}}\,
J(\f r,\tilde{\f r})\,,
\end{equation}
where
\begin{align}
&J(\f r,\tilde{\f r}) \deq \int_{[0,\nu]^3} \dd \f \tau
\int_{[0,\nu]^3} \dd \tilde{\f \tau}
\int \dd \f x
\int \dd \tilde{\f x}
\notag \\
&\times  \int \bb W_{x_1, x_3}^{\tau_1 + r_3, \tau_3}(\dd \omega_3) \, \bb W_{x_3, x_2}^{\tau_3 + r_2, \tau_2}(\dd \omega_2) \, \bb W_{x_2, x_1}^{\tau_2 + r_1, \tau_1}(\dd \omega_1)
\bb W_{\tilde x_1, \tilde x_3}^{\tilde \tau_1 + \tilde r_3, \tilde \tau_3}(\dd \tilde \omega_3) \, \bb W_{\tilde x_3, \tilde x_2}^{\tilde \tau_3 + \tilde r_2, \tilde \tau_2}(\dd \tilde \omega_2) \, \bb W_{\tilde x_2, \tilde x_1}^{\tilde \tau_2 + \tilde r_1, \tilde \tau_1}(\dd \tilde \omega_1)
\notag \\ \label{def_J}
&\times \int \mu_{\cal C_\eta}(\dd \sigma) \, \sigma(\tau_2, x_2) \, \sigma(\tau_3, x_3) \, \sigma(\tilde \tau_2, \tilde x_2) \, \sigma(\tilde \tau_3, \tilde x_3) \, 
\ee^{\ii \int \dd s \, \sigma([s]_\nu, \omega_1(s)) - \ii \int \dd s \, \sigma([s]_\nu, \tilde \omega_1(s))}\,.
\end{align}
Recalling \eqref{I(r,tilde_r)}, we shall show that for all $\f r,\tilde{\f r} \in (\nu \N)^3$ with $\abs{\f r}, \abs{\tilde{\f r}} > 0$ we have
\begin{equation}
J(\f r, \tilde{\f r})=\nu^6\, I(\f r,\tilde{\f r})
\label{J_r_tilde_r_3}
+ \nu^6 \, \cal E_{\eta, \nu}(\f r, \tilde {\f r})  \biggl(1+\frac{L^d}{\abs{\f r}^{d/2}}\biggr)\, \biggl(1+\frac{L^d}{\abs{\tilde{\f r}}^{d/2}}\biggr) (1+\abs{\f r}+\abs{\tilde{\f r}})^6
\,.
\end{equation}
The claim then follows from \eqref{J_r_tilde_r_3} by summing in $\f r, \tilde{\f r}$ and considering Riemann sums for \eqref{int_ff}.

The rest of the proof is devoted to obtaining \eqref{J_r_tilde_r_3}.
First, we apply Lemma \ref{lem:Gaussian} to the last line of \eqref{def_J}, using that $\mu_{\cal C_\eta}$ is a finite-dimensional real Gaussian measure with covariance $\cal C_\eta$ from \eqref{cov_C_eta}. As for $x$, we use the notation $\tau_{i,0} = \tau_i$ and $\tau_{i,1} = \tilde \tau_i$. More precisely, in the application of Lemma \ref{lem:Gaussian}, we take
\begin{equation}
\label{f_f_a2}
f(\tau,x)=\int \dd s\,\delta (\tau-[s]_{\nu})\,\delta (x-\omega_1(s))-\int \dd \tilde s\,\delta (\tau-[\tilde s]_{\nu})\,\delta (x-\tilde \omega_1(\tilde s))\,,\quad f_a(\tau,x)=\delta (\tau-\tau_a)\,\delta(x-x_a)\,, 
\end{equation}
and apply an appropriate time translation (by multiples of $\nu$, since $\f r,\tilde{\f r} \in (\nu \N)^3$) in the paths $\omega_2,\omega_3,\tilde \omega_2, \tilde \omega_3$ to deduce that
\begin{align}
\notag
&J(\f r, \tilde{\f r})=\int_{[0,\nu]^3} \dd \f \tau
\int_{[0,\nu]^3} \dd \tilde{\f \tau}
\int \dd \f x
\int \dd \tilde{\f x}\,\int \bb W_{x_1, x_3}^{\tau_1+\abs{\f r},\tau_3+r_1+r_2}(\dd \omega_3) \, \bb W_{x_3, x_2}^{\tau_3+r_1+r_2,\tau_2+r_1}(\dd \omega_2) \, \bb W_{x_2, x_1}^{\tau_2+r_1, \tau_1}(\dd \omega_1)
\\
\notag
&
\times \bb W_{\tilde x_1, \tilde x_3}^{\tilde \tau_1+\abs{\tilde{\f r}},\tilde \tau_3 +\tilde r_1+\tilde r_2 }(\dd \tilde \omega_3) \, \bb W_{\tilde x_3, \tilde x_2}^{\tilde \tau_3+\tilde r_1+ \tilde r_2,\tilde \tau_2 + \tilde r_1}(\dd \tilde \omega_2) \, \bb W_{\tilde x_2, \tilde x_1}^{\tilde \tau_2 + \tilde r_1, \tilde \tau_1}(\dd \tilde \omega_1)\,\ee^{-\frac{1}{2} \pb{\bb V_{\eta, \nu}(\omega_1, \omega_1) + \bb V_{\eta, \nu}(\tilde \omega_1, \tilde \omega_1) - 2 \bb V_{\eta, \nu}(\omega_1, \tilde \omega_1)}}
\\
\label{int_FF_2}
& \sum_{\Pi \in \fra M(A)} \prod_{\{a,b\} \in \Pi} (\bb V_{\eta, \nu})^{\tau_a, \tau_b}_{x_a, x_b} \prod_{a \in A \setminus [\Pi]} \ii \pB{(\bb V_{\eta, \nu})^{\tau_a}_{x_a}(\omega_1) - (\bb V_{\eta, \nu})^{\tau_a}_{x_a}(\tilde \omega_1)}\,.
\end{align}

For fixed values of $\f \tau, \tilde{\f \tau} \in [0,\nu]^3$, we rewrite the integrand in \eqref{int_FF_2}. 
We first introduce some notation. In the discussion that follows $\hat{q}$ denotes either the quantity $q$ or $\tilde{q}$. 
With this convention, we define loops $\hat{\omega}:[\hat \tau_1,\hat \tau_1+\abs{\hat{\f r}}] \rightarrow \Lambda$ by concatenating $\hat{\omega}_1,\hat{\omega}_2,\hat{\omega}_3$. More precisely, 
\begin{equation*}
\hat{\omega}|_{[\hat \tau_1,\hat \tau_2+\hat{r}_1] }=\hat{\omega}_1\,,\hat{\omega}|_{[\hat \tau_2+\hat{r}_1,\hat \tau_3+\hat r_1 + \hat r_2] }=\hat{\omega}_2\,, \hat{\omega}|_{[\hat \tau_3+\hat r_1 + \hat r_2,\hat \tau_1+\abs{\hat{\f r}}] } =\hat{\omega}_3\,.
\end{equation*}
For $a \in A$, we define the quantities $t_a,s_a,\omega_a$ by
\begin{equation}
\label{choice_of_parameters}
\begin{cases}
t_{(2,0)} \deq \tau_2 +r_1\,,\quad t_{(3,0)} \deq \tau_3+r_1+r_2\,,\quad t_{(2,1)} \deq \tilde \tau_2+\tilde r_1\,,\quad t_{(3,1)} \deq \tilde \tau_3+\tilde r_1+\tilde r_2
\\
s_{(2,0)} \deq \tau_1+r_1\,,\quad s_{(3,0)} \deq \tau_1+r_1+r_2\,,\quad s_{(2,1)} \deq \tilde \tau_1+\tilde r_1\,,\quad s_{(3,1)} \deq \tilde\tau_1+\tilde r_1+\tilde r_2
\\
\omega_{(2,0)} = \omega_{(3,0)}\deq \omega\,,\quad \omega_{(2,1)} = \omega_{(3,1)}\deq \tilde \omega\,.
\end{cases}
\end{equation}
The interpretation of the times defined in \eqref{choice_of_parameters} is that the times $t_a$ are the precise initial and final times of the paths $\hat \omega_1, \hat \omega_2, \hat \omega_3$. The times $s_a$ are approximations of the corresponding times $t_a$, which have been decoupled from the variables $\hat \tau_2, \hat \tau_3$ by replacing them with $\hat \tau_1$; evaluating $\hat \omega$ at the times $s_a$ instead of $t_a$ allows us to remove the $(\hat \tau_2, \hat \tau_3)$-dependence from the spatial indices (see e.g.\ \eqref{J_r_tilde_r_2} below), and hence ultimately relate $J(\f r, \tilde{\f r})$ to $I(\f r, \tilde{\f r})$.
With the above notation, we deduce that for all $a \in A$ we have $\omega_a(t_a)=x_a$  and $|s_a-t_a| \leq \nu$.

With $f$ chosen as in \eqref{f_f_a2} we have that 
\begin{equation}
\label{int_FF_convergence_lemma_1}
\bb V_{\eta,\nu}(\omega_1, \omega_1) + \bb V_{\eta,\nu}(\tilde \omega_1, \tilde \omega_1) - 2 \bb V_{\eta,\nu}(\omega_1, \tilde \omega_1)=\langle f, \cal C_{\eta,\nu} f \rangle \geq 0\,.
\end{equation}
By using \eqref{int_ff_convergence_lemma_1}, \eqref{int_FF_convergence_lemma_1}, Lemma \ref{path-path_interaction} (i), and the fact that $\tau_1,\tilde \tau_1,\tau_2,\tilde \tau_2 \in [0,\nu]$, it follows that for $\omega_1, \tilde \omega_1$ as in the integrand of \eqref{int_FF_2} we have
\begin{equation}
\label{exponential_term}
\ee^{-\frac{1}{2} \pb{\bb V_{\eta, \nu}(\omega_1, \omega_1) + \bb V_{\eta, \nu}(\tilde \omega_1, \tilde \omega_1) - 2 \bb V_{\eta, \nu}(\omega_1, \tilde \omega_1)}}
=\ee^{-\frac{1}{2} \pb{\bb V_{\eta}(\omega_1, \omega_1) + \bb V_{\eta}(\tilde \omega_1, \tilde \omega_1) - 2 \bb V_{\eta}(\omega_1, \tilde \omega_1)}}
+\cal E_{\eta, \nu}(\omega, \tilde \omega)(1+r_1+\tilde r_1)^2\,.
\end{equation}
We use \eqref{choice_of_parameters}, \eqref{exponential_term} and integrate in $x_2,x_3,\tilde x_2, \tilde x_3$ to rewrite \eqref{int_FF_2} as
\begin{multline}
\label{J_r_tilde_r}
J(\f r, \tilde{\f r})=\int_{[0,\nu]^3} \dd \f \tau \int_{[0,\nu]^3} \dd \tilde{\f \tau} \int \dd x_1 \int \dd \tilde x_1\int \bb W^{\tau_1+\abs{\f r},\tau_1}_{x_1,x_1}(\dd \omega)\,\bb W^{\tilde \tau_1+\abs{\tilde{\f r}},\tilde \tau_1}_{\tilde x_1, \tilde x_1}(\dd \tilde \omega)
\\
\times \biggl[\ee^{-\frac{1}{2} \pb{\bb V_{\eta}(\omega_1, \omega_1) + \bb V_{\eta}(\tilde \omega_1, \tilde \omega_1) - 2 \bb V_{\eta}(\omega_1, \tilde \omega_1)}}
+\cal E_{\eta, \nu}(\omega, \tilde \omega)(1+r_1+\tilde r_1)^2\biggr]
\\
\times \sum_{\Pi \in \fra M(A)} \prod_{\{a,b\} \in \Pi} (\bb V_{\eta, \nu})^{\tau_a, \tau_b}_{\omega_a(t_a),\omega_b(t_b)} \prod_{a \in A \setminus [\Pi]} \ii \pB{(\bb V_{\eta, \nu})^{\tau_a}_{\omega_a(t_a)}(\omega_1) - (\bb V_{\eta, \nu})^{\tau_a}_{\omega_a(t_a)}(\tilde \omega_1)}\,.
\end{multline}

In order to analyse \eqref{J_r_tilde_r}, we note the following three claims.
\begin{itemize}
\item[(i)] We have
\begin{align*}
&\ee^{-\frac{1}{2} \pb{\bb V_{\eta}(\omega_1, \omega_1) + \bb V_{\eta}(\tilde \omega_1, \tilde \omega_1) - 2 \bb V_{\eta}(\omega_1, \tilde \omega_1)}}
\\
&=\ee^{-\frac{1}{2} \pb{\bb V_{\eta}\bigl(\omega|_{[\tau_1,\tau_1+r_1]}, \omega|_{[\tau_1,\tau_1+r_1]}\bigr)
+\bb V_{\eta}\bigl(\tilde \omega|_{[\tilde \tau_1,\tilde \tau_1+\tilde r_1]},\tilde \omega|_{[\tilde \tau_1,\tilde \tau_1+\tilde r_1]}\bigr) - 2 \bb V_{\eta}\bigl(\omega|_{[\tau_1,\tau_1+r_1]}, \tilde \omega|_{[\tilde \tau_1,\tilde \tau_1+\tilde r_1]}\bigr)}}
\\
& +\cal E_{\eta, \nu}(\omega, \tilde \omega)(1+r_1+\tilde r_1)\,.
\end{align*}

\item[(ii)] For all $a,b \in A$ we have
\begin{equation}
\label{Claim_(ii)}
\int_{0}^{\nu} \dd \tau_a \,\int_{0}^{\nu} \dd \tau_b\,(\bb V_{\eta, \nu})^{\tau_a, \tau_b}_{\omega_a(t_a),\omega_b(t_b)}=\nu^2
(\bb V_{\eta})_{\omega_a(s_a),\omega_b(s_b)} + \nu^2\,\cal E_{\eta, \nu}(\omega, \tilde \omega)\,.
\end{equation}
\item[(iii)] For all $a \in A$ we have 
\begin{equation}
\label{Claim_(iii)}
\int_{0}^{\nu} \dd \tau_a \,(\bb V_{\eta, \nu})^{\tau_a}_{\omega_a(t_a)}(\hat \omega_1)=\nu (\bb V_{\eta})_{\omega_a(s_a)}\bigl(\hat \omega|_{[\hat \tau_1,\hat \tau_1+\hat r_1]}\bigr) +\nu \cal E_{\eta, \nu}(\omega, \tilde \omega)(1+\hat r_1)\,.
\end{equation}
\end{itemize}

We prove each of the claims (i)--(iii) separately.
Note that (i) follows from \eqref{int_ff_convergence_lemma_1} and Lemma \ref{path-path_interaction} (ii). We now prove (ii). By Definition \ref{def:interaction_general}, Lemma \ref{properties_of_v_eta} and \eqref{properties_of_delta_eta}, it follows that the left-hand side of \eqref{Claim_(ii)} is
\begin{equation*} 
\nu \int_{0}^{\nu} \dd \tau_a \,\int_{0}^{\nu} \dd \tau_b\,\delta_{\eta,\nu}(\tau_a-\tau_b)\,v_{\eta}(\omega_a(t_a)-\omega_b(t_b))=\nu^2\,v_{\eta}(\omega_a(s_a)-\omega_b(s_b))+\nu^2\,\cal E_{\eta, \nu}(\omega, \tilde \omega)\,,
\end{equation*}
which by Definition \ref{def:interaction} equals the right-hand side of \eqref{Claim_(ii)}.

In order to prove (iii), we use Definition \ref{def:interaction_general}, Lemma \ref{properties_of_v_eta} and \eqref{properties_of_delta_eta} to rewrite the left-hand side of \eqref{Claim_(iii)} as
\begin{multline*}
\nu \int_{0}^{\nu} \dd \tau_a\, \biggl[\int_{\hat \tau_1}^{\hat\tau_1+\hat r_1}\dd s\,\delta_{\eta,\nu}(\tau_a-[s]_{\nu})\,v_{\eta}(\omega_a(s_a)- \hat \omega(s))+\int_{\hat \tau_1}^{\hat \tau_1+\hat r_1}\dd s\,\delta_{\eta,\nu}(\tau_a-[s]_{\nu})\,\cal E_{\eta, \nu}(\omega, \tilde \omega)
\\
+O\biggl(\int_{\hat \tau_1+\hat r_1}^{\hat \tau_2+\hat r_1}\dd s\,\delta_{\eta,\nu}(\tau_a-[s]_{\nu})\biggr) \biggr]
=\nu \int_{\hat \tau_1}^{\hat \tau_1+\hat r_1}\dd s\,v_{\eta}(\omega_a(s_a)- \hat \omega(s))+\nu \cal E_{\eta, \nu}(\omega, \tilde \omega)(1+\hat r_1)\,,
\end{multline*}
which by Definition \ref{def:interaction} equals the right-hand side of \eqref{Claim_(iii)}.

We integrate in $\tau_2,\tau_3,\tilde \tau_2, \tilde \tau_3$, and apply (i)--(iii) to rewrite \eqref{J_r_tilde_r} as
\begin{multline}
\label{J_r_tilde_r_2} 
J(\f r, \tilde{\f r})=\nu^4 \int_{0}^{\nu} \dd \tau_1 \int_{0}^{\nu} \dd \tilde \tau_1 \int \dd x_1 \int \dd \tilde x_1\int \bb W^{\tau_1+\abs{\f r},\tau_1}_{x_1,x_1}(\dd \omega)\,\bb W^{\tilde \tau_1+\abs{\tilde{\f r}},\tilde \tau_1}_{\tilde x_1, \tilde x_1}(\dd \tilde \omega)
\\
\times \ee^{-\frac{1}{2} \pb{\bb V_{\eta}\bigl(\omega|_{[\tau_1,\tau_1+r_1]}, \omega|_{[\tau_1,\tau_1+r_1]}\bigr)
+\bb V_{\eta}\bigl(\tilde \omega|_{[\tilde \tau_1,\tilde \tau_1+\tilde r_1]},\tilde \omega|_{[\tilde \tau_1,\tilde \tau_1+\tilde r_1]}\bigr) - 2 \bb V_{\eta}\bigl(\omega|_{[\tau_1,\tau_1+r_1]}, \tilde \omega|_{[\tilde \tau_1,\tilde \tau_1+\tilde r_1]}\bigr)}}
\\
\times \sum_{\Pi \in \fra M(A)} \prod_{\{a,b\} \in \Pi} (\bb V_{\eta})_{\omega_a(s_a),\omega_b(s_b)} \,
\prod_{a \in A \setminus [\Pi]} \ii \Bigl[(\bb V_{\eta})_{\omega_a(s_a)}\bigl(\omega|_{[\tau_1,\tau_1+r_1]}\bigr)-(\bb V_{\eta})_{\omega_a(s_a)}(\tilde \omega|_{[\tilde \tau_1,\tilde \tau_1+\tilde r_1]})\Bigr]
\\
+\nu^6 \cal E_{\eta, \nu}(\f r, \tilde {\f r}) \biggl(1+\frac{L^d}{\abs{\f r}^{d/2}}\biggr)\, \biggl(1+\frac{L^d}{\abs{\tilde{\f r}}^{d/2}}\biggr) (1+\abs{\f r}+\abs{\tilde{\f r}})^6 
\end{multline}

For the second term on the right-hand side of \eqref{J_r_tilde_r_2}, we used 
\begin{equation*}
|(\bb V_{\eta})_{\omega_a(s_a),\omega_b(s_b)}|\leq \|v\|_{L^{\infty}}\,,\quad |(\bb V_{\eta})_{\omega_a(s_a)}\bigl(\hat \omega|_{[\hat \tau_1,\hat \tau_1+\hat r_1]}\bigr)| \leq \hat r_1\|v\|_{L^{\infty}}
\end{equation*} 
which follow from \eqref{int_ff_convergence_lemma_2}, as well as the observation that for all $\tau_1,\tilde \tau_1 \in [0,\nu]$ we have
\begin{equation}
\label{DCT_application1}
\int \dd x_1 \int \dd \tilde x_1\int \bb W^{\tau_1+\abs{\f r},\tau_1}_{x_1,x_1}(\dd \omega)\,\bb W^{\tilde \tau_1+\abs{\tilde{\f r}},\tilde \tau_1}_{\tilde x_1, \tilde x_1}(\dd \tilde \omega)\,\cal E_{\eta,\nu}(\omega, \tilde \omega)=
 \cal E_{\eta, \nu}(\f r, \tilde {\f r}) \biggl(1+\frac{L^d}{\abs{\f r}^{d/2}}\biggr)\, \biggl(1+\frac{L^d}{\abs{\tilde{\f r}}^{d/2}}\biggr)\,,
\end{equation}
which follows from Lemma \ref{lem:heat_estimate} and the dominated convergence theorem.

Finally, the first term on the right-hand side of \eqref{J_r_tilde_r_2} is equal to $\nu^6 I(\f r,\tilde{\f r})$, as can be seen
by a time translation of the paths, $\hat \omega(\cdot) \mapsto \hat \omega(\cdot - \hat \tau_1)$.
This concludes the proof of \eqref{J_r_tilde_r_3}.
\end{proof}

\begin{proof}[Proof of Lemma \ref{lem:Ff}]
We begin by noting that by Lemma \ref{lem:Ff_rep} and \eqref{sigma(x)_definition} we have
\begin{multline} \label{f2_rep2}
f_2(\ang{\sigma})
= - \frac{1}{\nu^3} \int_{[0,\infty)^3} \dd \f r \, \frac{\ee^{- \kappa \abs{\f r}}}{\abs{\f r}} \, \int_{[0,\nu]^3} \dd \f \tau \int \dd \f x \, \sigma(\tau_2, x_2) \, \sigma(\tau_3, x_3)
\\
\times \int \bb W_{x_1, x_3}^{r_3, 0}(\dd \omega_3) \, \bb W_{x_3, x_2}^{r_2, 0}(\dd \omega_2) \, \bb W_{x_2, x_1}^{r_1, 0}(\dd \omega_1) \, \ee^{\ii \frac{1}{\nu} \int_0^\nu \dd t \int_0^{r_1} \dd s \, \sigma(t,\omega_1(s))}\,.
\end{multline}
Hence,
\begin{multline} \label{int_Ff}
\int \mu_{\cal C_\eta}(\dd \sigma) \, \ol{f_2(\ang{\sigma})} F_2(\sigma)
=
\sum_{\f r \in (\nu \N)^3} \frac{\ind{\abs{\f r} > 0} \, \ee^{-\kappa \abs{\f r}}}{\abs{\f r}}
\, \frac{1}{\nu^3} \int_{[0,\infty)^3} \dd \tilde{\f r} \, \frac{\ee^{- \kappa \abs{\tilde{\f r}}}}{\abs{\tilde{\f r}}}
\int_{[0,\nu]^3} \dd \f \tau
\int_{[0,\nu]^3} \dd \tilde{\f \tau}
\int \dd \f x
\int \dd \tilde{\f x}
\\
\times  \int \bb W_{x_1, x_3}^{\tau_1 + r_3, \tau_3}(\dd \omega_3) \, \bb W_{x_3, x_2}^{\tau_3 + r_2, \tau_2}(\dd \omega_2) \, \bb W_{x_2, x_1}^{\tau_2 + r_1, \tau_1}(\dd \omega_1)
\bb W_{\tilde x_1, \tilde x_3}^{\tilde r_3, 0}(\dd \tilde \omega_3) \, \bb W_{\tilde x_3, \tilde x_2}^{\tilde r_2, 0}(\dd \tilde \omega_2) \, \bb W_{\tilde x_2, \tilde x_1}^{\tilde r_1, 0}(\dd \tilde \omega_1)
\\
\times \int \mu_{\cal C_\eta}(\dd \sigma) \, \sigma(\tau_2, x_2) \, \sigma(\tau_3, x_3) \, \sigma(\tilde \tau_2, \tilde x_2) \, \sigma(\tilde \tau_3, \tilde x_3) \, 
\ee^{\ii \int \dd s \, \sigma([s]_\nu, \omega_1(s)) - \ii \frac{1}{\nu} \int_0^\nu \dd t \int \dd s \, \sigma(t, \tilde \omega_1(s))}\,.
\end{multline}
Using Lemma \ref{lem:Gaussian} with
\begin{equation}
\label{f_f_a3}
f(\tau,x)=\int \dd s\,\delta (\tau-[s]_{\nu})\,\delta (x-\omega_1(s))-
\frac{1}{\nu}\int \dd \tilde s\,\delta(x-\tilde \omega_1(\tilde s))\,,\quad f_a(\tau,x)=\delta (\tau-\tau_a)\,\delta(x-x_a)\,, 
\end{equation}
we conclude that the last line of \eqref{int_Ff} is equal to
\begin{equation}
\label{int_Ff_2}
\ee^{-\frac{1}{2} \pb{\bb V_{\eta, \nu}(\omega_1, \omega_1) + \bb V_\eta(\tilde \omega_1, \tilde \omega_1) - 2 \bb V_\eta (\omega_1, \tilde \omega_1)}}
\,\sum_{\Pi \in \fra M(A)} \prod_{\{a,b\} \in \Pi} (\bb V_{\eta, \nu})^{\tau_a, \tau_b}_{x_a, x_b} \prod_{a \in A \setminus [\Pi]} \ii \pB{(\bb V_{\eta, \nu})^{\tau_a}_{x_a}(\omega_1) -  (\bb V_\eta)_{x_a}(\tilde \omega_1)}\,.
\end{equation}

The proof now proceeds similarly as that of Lemma \ref{lem:FF}, and we only outline the main differences. The first difference is that, instead of \eqref{exponential_term}, we use
\begin{equation}
\label{exponential_term2}
\ee^{-\frac{1}{2} \pb{\bb V_{\eta, \nu}(\omega_1, \omega_1) + \bb V_\eta(\tilde \omega_1, \tilde \omega_1) - 2 \bb V_\eta(\omega_1, \tilde \omega_1)}}
=\ee^{-\frac{1}{2} \pb{\bb V_{\eta}(\omega_1, \omega_1) + \bb V_{\eta}(\tilde \omega_1, \tilde \omega_1) - 2 \bb V_{\eta}(\omega_1, \tilde \omega_1)}}
+\cal E_{\eta, \nu}(\omega, \tilde \omega)(1+r_1)^2\,, 
\end{equation}
which follows from Lemma \ref{path-path_interaction} (i), by recalling \eqref{int_ff_convergence_lemma_1} and noting that the exponent of the expression on the left-hand side of \eqref{exponential_term2} is nonpositive by \eqref{f_f_a3}.

The second difference is that for $a \in A$, instead of \eqref{Claim_(iii)}, we use
\begin{align}
\notag
\int_{0}^{\nu} \dd \tau_a \,(\bb V_{\eta,\nu})^{\tau_a}_{x_a}(\omega_1)&=\nu(\bb V_{\eta})_{\omega_a(s_a)}\bigl(\omega|_{[\tau_1,\tau_1+r_1]}\bigr) +\nu \cal E_{\eta, \nu}(\omega,\tilde \omega) \, (1+r_1)\,, \\
\label{Claim_(iii)_2}
\int_{0}^{\nu} \dd \tau_a \,(\bb V_{\eta})_{x_a}(\tilde \omega_1)&=\nu(\bb V_{\eta})_{\omega_a(s_a)}\bigl(\tilde \omega|_{[\tilde \tau_1,\tilde \tau_1+\tilde r_1]}\bigr) +\nu \cal E_{\eta, \nu}(\omega,\tilde \omega) \, \tilde r_1\,.
\end{align}
We obtain \eqref{Claim_(iii)_2} from 
Lemma \ref{properties_of_v_eta}. We also recall that $\omega_a(t_a)=x_a$.
Finally, instead of \eqref{DCT_application1}, we use 
\begin{equation*}
\int \dd x_1 \int \dd \tilde x_1\int \bb W^{\tau_1+\abs{\f r},\tau_1}_{x_1,x_1}(\dd \omega)\,\bb W^{\abs{\tilde{\f r}},0}_{\tilde x_1, \tilde x_1}(\dd \tilde \omega)\,\cal E_{\eta,\nu}(\omega, \tilde \omega)=
 \cal E_{\eta, \nu}(\f r, \tilde {\f r})  \biggl(1+\frac{L^d}{\abs{\f r}^{d/2}}\biggr)\, \biggl(1+\frac{L^d}{\abs{\tilde{\f r}}^{d/2}}\biggr)\,,
\end{equation*}
which follows by the same proof.
\end{proof}

\subsection{Correlation functions}

In this subsection we analyse the correlation functions and conclude the proof of Proposition \ref{prop:main_conv}. From \eqref{def_eh_gamma_eta}, Lemma \ref{lem:G_K}, and \eqref{W_FK}, we get
\begin{equation} \label{Gamma_Q}
\nu^p \wh \Gamma_{p,\eta} = \frac{p!}{{\cal Z}_\eta} P_p Q_{p,\eta}\,,
\end{equation}
where
\begin{equation} \label{def_Q}
(Q_{p, \eta})_{\f x, \tilde {\f x}} \deq \int \mu_{\cal C_\eta}(\dd \sigma)\, \ee^{F_2(\sigma)} \, \prod_{i = 1}^p \pBB{\nu \sum_{r_i \in \nu \N^*} \ee^{-\kappa r_i} \int \bb W^{r_i, 0}_{x_i, \tilde x_i}(\dd \omega_i) \, \pB{\ee^{\ii \int_0^{r_i} \dd t \, \sigma([t]_\nu, \omega_i(t))} - 1}}\,.
\end{equation}
Analogously, from \eqref{def_gamma_eta_hat}, writing $\frac{1}{h} = \int_0^\infty \dd r \, \ee^{-r h}$, and Lemma \ref{FK_continuous}, we get
\begin{equation*}
\wh \gamma_{p,\eta} = \frac{p!}{\zeta_\eta} P_p q_{p,\eta} \,,
\end{equation*}
where
\begin{equation} \label{def_q}
(q_{p,\eta})_{\f x, \tilde {\f x}} \deq \int \mu_{v_\eta}(\dd \xi)\, \ee^{f_2(\xi)} \, \prod_{i = 1}^p \pBB{\int_0^\infty \dd r_i \, \ee^{-\kappa r_i} \int \bb W^{r_i, 0}_{x_i, \tilde x_i}(\dd \omega_i) \, \pB{\ee^{\ii \int_0^{r_i} \dd t \, \xi(\omega_i(t))} - 1}}\,.
\end{equation}

Thanks to Proposition \ref{prop:conv_Z}, in order to prove Proposition \ref{prop:main_conv} it suffices to prove the following result.

\begin{proposition} \label{prop:Qq}
For any $p \in \N^*$ we have, uniformly in $\eta > 0$,
\begin{equation*}
\lim_{\nu \to 0} \norm{Q_{p,\eta} - q_{p,\eta}}_{L^\infty} = 0\,.
\end{equation*}
\end{proposition}

We shall prove Proposition \ref{prop:Qq} in three steps.
\begin{enumerate}
\item[\emph{Step 1.}]
In \eqref{def_Q} and \eqref{def_q}, we truncate the values of $r_i$ to lie in some interval $[\delta, 1/\delta]$.
\item[\emph{Step 2.}]
In the resulting truncated expression for $q$, we replace the integral $\int_{1/\delta}^{\delta} \dd r_i \, [\cdots]$ with a corresponding Riemann sum $\nu \sum_{r_i \in \nu \N^*} \ind{\delta \leq r_i \leq 1/ \delta} [\cdots]$.
\item[\emph{Step 3.}]
We use Lemma \ref{lem:sigma_zeta} to replace $\xi$ with $\ang{\sigma}$ and compare the resulting approximations of $q$ and $Q$.
\end{enumerate}

\paragraph{Step 1}
We rewrite \eqref{def_Q} as
\begin{equation} \label{Q_Y}
(Q_{p, \eta})_{\f x, \tilde {\f x}} = \nu^p \sum_{\f r \in (\nu \N^*)^p} \ee^{-\kappa \abs{\f r}} \, Y_{\f x, \tilde {\f x}}(\f r)
\end{equation}
where
\begin{equation*}
Y_{\f x, \tilde {\f x}}(\f r)  \deq \int \mu_{\cal C_\eta}(\dd \sigma)\, \ee^{F_2(\sigma)} \, \prod_{i = 1}^p \pBB{\int \bb W^{r_i, 0}_{x_i, \tilde x_i}(\dd \omega_i) \, \pB{\ee^{\ii \int_0^{r_i} \dd t \, \sigma([t]_\nu, \omega_i(t))} - 1}}\,.
\end{equation*}
By Lemma \ref{prop:fif_renorm} (iv) and Cauchy-Schwarz, we have
\begin{equation} \label{Y_estimate}
\abs{Y_{\f x, \tilde {\f x}}(\f r)} \leq \int \prod_{i = 1}^p  \bb W^{r_i, 0}_{x_i, \tilde x_i}(\dd \omega_i) \pbb{ \int \mu_{\cal C_\eta}(\dd \sigma) \, \prod_{i = 1}^p \absB{\ee^{\ii \int_0^{r_i} \dd t \, \sigma([t]_\nu, \omega_i(t))} - 1}^2}^{1/2}
\end{equation}
We estimate the right-hand side using the following basic algebraic fact.
\begin{lemma} \label{lem:T_p}
For each $p \in \N^*$ there exists a polynomial $T_p$ with the following properties.
\begin{enumerate}[label=(\roman*)]
\item
$1$ is a root of $T_p$ with multiplicity at least $p$.
\item
For any centred Gaussian measure $\mu_{\cal C}$ with covariance $\cal C$ we have
\begin{equation*}
\int \mu_{\cal C}(\dd u) \, \abs{1 - \ee^{\ii \scalar{f}{u}}}^{2p} = T_p\pb{\ee^{-\frac{1}{2} \scalar{f}{\cal C f}}}\,.
\end{equation*}
\end{enumerate}
\end{lemma}
\begin{proof}
We define
\begin{equation*}
T_p(x) \deq \sum_{k,l = 0}^p \binom{p}{k} \binom{p}{l} (-1)^{k+l} \, x^{(k - l)^2}\,.
\end{equation*}
Thus, since
\begin{align*}
\int \mu_{\cal C}(\dd u) \, \absb{1 - \ee^{\ii \scalar{f}{u}}}^{2p} &= \int \mu_{\cal C}(\dd u) \, \pb{1 - \ee^{\ii \scalar{f}{u}}}^p \, \pb{1 - \ee^{-\ii \scalar{f}{u}}}^p
\\
&= \sum_{k,l = 0}^p \binom{p}{k} \binom{p}{l} (-1)^{k+l} \int \mu_{\cal C}(\dd u) \, \ee^{\ii (k - l) \scalar{f}{u}}
\\
&= \sum_{k,l = 0}^p \binom{p}{k} \binom{p}{l} (-1)^{k+l} \, \ee^{-\frac{1}{2} (k-l)^2 \scalar{f}{\cal C f}}\,,
\end{align*}
we find (ii). We establish (i) by considering the one-dimensional Gaussian measure with $\cal C  = 1$. Writing $f =\sqrt{ \lambda}$ for $\lambda \geq 0$, we find
\begin{equation*}
T_p(\ee^{-\lambda/2}) =  \int \mu_1(\dd u) \, \abs{1 - \ee^{\ii u \sqrt{\lambda}}}^{2p} = O(\lambda^p)\,,
\end{equation*}
from which we deduce that $T_p(x) = O((1 - x)^p)$ for $0 < x < 1$, and hence (i) follows.
\end{proof}
Using H\"{o}lder's inequality and Lemma \ref{lem:T_p} (ii), we get from \eqref{Y_estimate}
\begin{align*}
\abs{Y_{\f x, \tilde {\f x}}(\f r)} &\leq \prod_{i = 1}^{p} \int  \bb W^{r_i, 0}_{x_i, \tilde x_i}(\dd \omega_i)  \pbb{ \int \mu_{\cal C_\eta}(\dd \sigma) \, \absB{\ee^{\ii \int_0^{r_i} \dd t \, \sigma([t]_\nu, \omega_i(t))} - 1}^{2p}}^{1/2p}
\\
&= \prod_{i = 1}^{p} \int  \bb W^{r_i, 0}_{x_i, \tilde x_i}(\dd \omega_i) \, T_p \pbb{\exp \pbb{-\frac{\nu}{2} \int_0^{r_i} \dd t \, \dd s \, \delta([t]_\nu - [s]_\nu)  \,  v_\eta(\omega_i(t) - \omega_i(s))}}^{1/2p}\,.
\end{align*}
Here in the application of Lemma \ref{lem:T_p} (ii) we took $f(\tau,x)=\int_{0}^{r_i} \dd t\, \delta(\tau-[t]_{\nu}) \delta (x-\omega_i(t)).$
The exponent is nonpositive and bounded in absolute value by $\frac{1}{2} \norm{v}_{L^\infty} r_i^2$ (by Lemma \ref{properties_of_v_eta} (i)), so that by $1 - \ee^{-x} \leq x$ for $x \geq 0$, and Lemma \ref{lem:T_p} (i) and Lemma  \ref{lem:heat_estimate} we get
\begin{equation} \label{Y_est2}
\abs{Y_{\f x, \tilde {\f x}}(\f r)} \leq  C \norm{v}_{L^\infty}^{p/2} \prod_{i = 1}^{p} \int  \bb W^{r_i, 0}_{x_i, \tilde x_i}(\dd \omega_i) \, r_i \leq C_{p,v,L} \prod_{i = 1}^{p} \pb{r_i + r_i^{1-d/2}}\,.
\end{equation}
Since $d \leq 3$, by a Riemann sum approximation we deduce that the sum \eqref{Q_Y} is bounded, uniformly in $\nu, \eta > 0$, and $\f x, \tilde {\f x}$, and the following result holds. To state it, we use the notation
\begin{equation*}
[a,b]_\nu \deq \nu \Z \cap [a,b]\,.
\end{equation*}
\begin{lemma} \label{lem:Q_delta}
For any $\epsilon > 0$ there exists $\delta > 0$ such that
\begin{equation*}
\absbb{(Q_{p, \eta})_{\f x, \tilde {\f x}} - \nu^p \sum_{\f r \in [\delta, 1/\delta]_\nu^p} \ee^{-\kappa \abs{\f r}} \, Y_{\f x, \tilde {\f x}}(\f r)
}
 \leq \epsilon\,,
\end{equation*}
for all $\nu, \eta > 0$ and $\f x, \tilde {\f x}$.
\end{lemma}
An almost identical argument yields the following analogous result. For $\f r, \f s \in (0,\infty)^p$ satisfying $\f s \leq \f r$ (i.e.\ $0 \leq s_i \leq r_i$ for all $i$), define
\begin{equation} \label{def_y}
y_{\f x, \tilde {\f x}}(\f r, \f s)  \deq \int \mu_{v_\eta}(\dd \xi)\, \ee^{f_2(\xi)} \, \prod_{i = 1}^p \pBB{\int \bb W^{r_i, 0}_{x_i, \tilde x_i}(\dd \omega_i) \, \pB{\ee^{\ii \int_0^{s_i} \dd t \, \xi(\omega_i(t))} - 1}}\,,
\end{equation}
where we include the extra argument $\f s$ for the purposes of Step 2 below.

\begin{lemma} \label{lem:q_delta}
For any $\epsilon > 0$ there exists $\delta > 0$ such that
\begin{equation*}
\absbb{(q_{p, \eta})_{\f x, \tilde {\f x}} - \int_{[\delta,1/\delta]^p} \dd \f r \, \ee^{-\kappa \abs{\f r}} \, y_{\f x, \tilde {\f x}}(\f r, \f r)
}
 \leq \epsilon\,,
\end{equation*}
for all $\nu, \eta > 0$ and $\f x, \tilde {\f x}$.
\end{lemma}

\paragraph{Step 2}
In the next step, we start from the expression $\int_{[\delta,1/\delta]^p} \dd \f r \, \ee^{-\kappa \abs{\f r}} \, y_{\f x, \tilde {\f x}}(\f r,\f r)$ (see Lemma \ref{lem:q_delta}) and replace the integral over $\f r \in [\delta,1/\delta]^p$ with a Riemann sum over $\f r \in [\delta,1/\delta]_\nu^p$. To that end, we shall have to compare the measures $\bb W^{r_i,0}_{x_i, \tilde x_i}$ and $\bb W^{\tilde r_i,0}_{x_i, \tilde x_i}$, which have the same spatial arguments but different time arguments.
This can also be regarded as a comparison of the propagators $W^{r_i,0}_{x_i, \tilde x_i}(\ii \xi)$ and $W^{\tilde r_i,0}_{x_i, \tilde x_i}(\ii \xi)$. The presence of the imaginary interaction field $\ii \xi$ makes a direct comparison difficult. The corresponding comparison of the free propagators $W^{r_i,0}_{x_i, \tilde x_i}(0) = \psi^{r_i}(x_i - \tilde x_i)$ and $W^{\tilde r_i,0}_{x_i, \tilde x_i}(0) = \psi^{\tilde r_i}(x_i - \tilde x_i)$ is an elementary estimate of the heat kernel \eqref{heat_kernel}. The key idea of our comparison is to use this observation and turn off the interaction field $\ii \xi$ for the times in $[s_i, r_i]$, replacing $y_{\f x, \tilde {\f x}}(\f r, \f r)$ with $y_{\f x, \tilde {\f x}}(\f r, \f s)$. Thus,  recalling Lemma  \ref{lem:Gamma} (i) we compare $W^{r_i,0}(\ii \xi) = W^{r_i,s_i}(\ii \xi) W^{s_i, 0}(\ii \xi)$ with $W^{r_i,s_i}(0) W^{s_i, 0}(\ii \xi)$, and show that they are close provided $r_i - s_i$ is small enough. Next, we show that $W^{r_i,s_i}(0) W^{s_i, 0}(\ii \xi)$ and $W^{\tilde r_i,s_i}(0) W^{s_i, 0}(\ii \xi)$ are close provided that $\abs{r_i - \tilde r_i}$ is small enough and $r_i - s_i$ is large enough, by an elementary estimate of the heat kernel \eqref{heat_kernel}.

We recall the relation \eqref{W_FK}. For the following estimates, we note that, by \eqref{W_FK} and Lemma \ref{lem:heat_estimate}, for $\tau \geq \tilde \tau + \delta$ we have
\begin{equation} \label{W_kernel_est}
\abs{W^{\tau, \tilde \tau}_{x, \tilde x}(\ii \xi)} \leq \psi^{\tau - \tilde \tau}(x - \tilde x) \leq C_{d,L,\delta}
\end{equation}
for all $x, \tilde {x}$.

\begin{lemma} \label{lem:y_psi}
For $\f r \in [\delta,1/\delta]^p$ and $\f s \leq \f r$ we have
\begin{equation*}
\abs{y_{\f x, \tilde {\f x}}(\f r, \f r) - y_{\f x, \tilde {\f x}}(\f r, \f s)} \leq C_{p, d, L, \delta, v} \max_i (r_i - s_i) 
\end{equation*}
for all $\nu, \eta > 0$ and $\f x, \tilde {\f x}$.
\end{lemma}
\begin{proof}
By Proposition \ref{prop:Z_field_Wick_ordered} (ii), telescoping, and using \eqref{W_kernel_est}, we conclude that it suffices to prove
\begin{equation} \label{y_est_step2}
\int \mu_{v_\eta}(\dd \xi) \, \absbb{\int \bb W^{r_i, 0}_{x_i, \tilde x_i}(\dd \omega_i) \, \pB{\ee^{\ii \int_0^{r_i} \dd t \, \xi(\omega_i(t))} - \ee^{\ii \int_0^{s_i} \dd t \, \xi(\omega_i(t))}}} \leq C_{L,\delta, v} (r_i - s_i)
\end{equation}
for all $i$.
By Cauchy-Schwarz, the left-hand side of \eqref{y_est_step2} is bounded by
\begin{multline*}
\int \bb W^{r_i, 0}_{x_i, \tilde x_i}(\dd \omega) \, \pBB{ \int \mu_{v_\eta}(\dd \xi) \, \absB{\ee^{\ii \int_{s_i}^{r_i} \dd t \, \xi(\omega(t))} - 1}^2}^{1/2}
\\
\leq  C  \int \bb W^{r_i, 0}_{x_i, \tilde x_i}(\dd \omega) \, \absB{\ee^{-\frac{1}{2} \int_{s_i}^{r_i} \dd t\, \dd s \, v_\eta(\omega(t) - \omega(s))} - 1}^{1/2} \leq C_{d,L,\delta} \sqrt{\norm{v}_{L^\infty}} \, (r_i - s_i)\,.
\end{multline*}
where in the second step we used Lemma \ref{lem:T_p} with $p = 1$ and $f(x)=\int_{s_i}^{r_i} \dd t\, \delta(x-\omega_i(t))$, and in the third step we used \eqref{W_kernel_est}. This concludes the proof.
\end{proof}

\begin{lemma} \label{lem:r_tilde_r}
Let $\f r, \tilde {\f r} \in [\delta,1/\delta]^p$ satisfy $\abs{\f r - \tilde {\f r}} \leq \nu$, and let $\f s$ satisfy $\f s \leq \f r$, $\f s \leq \tilde{\f r}$. Suppose moreover that $\min_i s_i \geq \delta/2$  and $1 \geq \min_{i}(r_i - s_i) \geq 2 \nu$. Then
\begin{equation*}
\absB{\ee^{-\kappa \abs{\f r}} \, y_{\f x, \tilde {\f x}}(\f r, \f s) - \ee^{-\kappa \abs{\tilde {\f r}}} \, y_{\f x, \tilde {\f x}}(\tilde {\f r}, \f s)} \leq C_{p,d,L,\delta,\kappa} \frac{\nu}{\min_i (r_i - s_i)}\,.
\end{equation*}
\end{lemma}
\begin{proof}
Since by \eqref{W_kernel_est} we have $\abs{y_{\f x, \tilde {\f x}}(\f r, \f s)} \leq C_{p,d,L,\delta}$, using that $\abs{\f r - \tilde {\f r}} \leq \nu$  and $\min_{i}(r_i - s_i) \leq 1/\delta$, it suffices to prove
\begin{equation}
\absB{ y_{\f x, \tilde {\f x}}(\f r, \f s) - y_{\f x, \tilde {\f x}}(\tilde{\f r}, \f s)} \leq C_{p,d,L,\delta,\kappa} \frac{\nu}{\min_i (r_i - s_i)}\,.
\end{equation}

With the notation \eqref{W_FK}, we have
\begin{multline*}
y_{\f x, \tilde {\f x}}(\f r, \f s) - y_{\f x, \tilde {\f x}}(\tilde{\f r}, \f s) =
\int \mu_{v_\eta}(\dd \xi)\, \ee^{f_2(\xi)} \,  \qBB{\prod_{i = 1}^p \pB{W^{r_i, s_i}(0) W^{s_i, 0}(\ii \xi) - W^{r_i, 0}(0)}_{x_i, \tilde x_i}
\\
- \prod_{i = 1}^p \pB{W^{\tilde r_i, s_i}(0) W^{s_i, 0}(\ii \xi) - W^{\tilde r_i, 0}(0)}_{x_i, \tilde x_i}}\,.
\end{multline*}
Note that \eqref{W_kernel_est} generalizes to
\begin{equation*}
\absb{\pb{W^{\tau, s}(0)W^{s, \tilde \tau}(\ii \xi)}_{x, \tilde x}} \leq \psi^{\tau - \tilde \tau}(x - \tilde x) \leq C_{d,L,\delta}
\end{equation*}
for any $\tilde \tau \leq s \leq \tau$ with $\tau-\tilde \tau \geq \delta$. Thus, by Proposition \ref{prop:Z_field_Wick_ordered} (ii) and telescoping, it suffices to prove
\begin{multline} \label{WW_nu_est}
\int \mu_{v_\eta}(\dd \xi) \, \absbb{\pB{W^{r_i, s_i}(0) W^{s_i, 0}(\ii \xi) - W^{r_i, 0}(0)}_{x_i, \tilde x_i}
- \pB{W^{\tilde r_i, s_i}(0) W^{s_i, 0}(\ii \xi) - W^{\tilde r_i, 0}(0)}_{x_i, \tilde x_i}}
\\
\leq C_{d,L,\delta} \frac{\nu}{r_i - s_i}\,.
\end{multline}
To that end, we estimate
\begin{align*}
&\mspace{-20mu} \absbb{\pB{W^{r_i, s_i}(0) W^{s_i, 0}(\ii \xi)}_{x_i, \tilde x_i}
- \pB{W^{\tilde r_i, s_i}(0) W^{s_i, 0}(\ii \xi)}_{x_i, \tilde x_i}}
\\
&\leq \int \dd y \, \absB{\psi^{r_i - s_i}(x_i - y) - \psi^{\tilde r_i - s_i}(x_i - y)} \absb{W^{s_i,0}_{y,\tilde x_i}(\ii \xi)}
\\
&\leq C_{d,L,\delta} \norm{\psi^{r_i - s_i} - \psi^{\tilde r_i - s_i}}_{L^1}
\\
&\leq C_{d,L,\delta} \absbb{\log \pbb{\frac{r_i - s_i}{\tilde r_i - s_i}}}\,,
\end{align*}
where in the second step we used that $s_i \geq \delta/2$ and \eqref{W_kernel_est}, and in the last step we used, for $s \leq t$,
\begin{equation*}
\norm{\psi^t - \psi^s}_{L^1} \leq \int_s^t \dd u \, \norm{\partial_u \psi^u}_{L^1} \leq \int_s^t \dd u \, \int_{\R^d} \dd x \, \pbb{\frac{d}{2u} + \frac{\abs{x}^2}{2 u^2}} \widetilde \psi^u(x) = d \log (t/s)\,,
\end{equation*}
where $\widetilde \psi^u(x)=(2 \pi u)^{-d/2} \ee^{-\abs{x}^2 / 2u}$ denotes the heat kernel on $\R^d$.
An analogous argument yields
\begin{equation*}
\absb{W^{r_i, 0}_{x_i, \tilde x_i}(0)
- W^{\tilde r_i, 0}_{x_i, \tilde x_i}(0) } \leq C_{d,L,\delta} \absbb{\log \pbb{\frac{r_i}{\tilde r_i}}}\,,
\end{equation*}
and hence \eqref{WW_nu_est} is proved, since by assumption $\abs{r_i - \tilde r_i} \leq \nu$ and $r_i - s_i \geq 2 \nu$.
\end{proof}

Armed with Lemmas \ref{lem:y_psi} and \ref{lem:r_tilde_r}, we may now compare the integral over $\f r$ of $y_{\f x, \tilde {\f x}}(\f r, \f r)$ to the corresponding Riemann sum.

\begin{lemma} \label{lem:y_sum}
For any $\delta > 0$ we have
\begin{equation*}
\absbb{\int_{[\delta,1/\delta]^p} \dd \f r \, \ee^{-\kappa \abs{\f r}} \, y_{\f x, \tilde {\f x}}(\f r, \f r)
-
\nu^p \sum_{\f r \in [\delta,1/\delta]_\nu^p}\, \ee^{-\kappa \abs{\f r}} \, y_{\f x, \tilde {\f x}}(\f r, \f r)
} \leq C_{p,L,\delta,v,\kappa} \sqrt{\nu}
\end{equation*}
for all $\eta > 0$, small enough $\nu > 0$, and $\f x, \tilde {\f x}$.
\end{lemma}

\begin{proof}
For $r > 0$ denote by $\floor{r}_\nu \deq \max\{u \in \nu \N \col u \leq r\}$, and write $\floor{\f r}_\nu \deq (\floor{r_i}_\nu)_{i = 1}^p$. Since $\abs{y_{\f x, \tilde {\f x}}(\f r, \f r)} \leq C_{p,d,L,\delta}$ for $\f r \in [\delta,1/\delta]^p$ by \eqref{W_kernel_est}, we find
\begin{equation} \label{RS_est1}
\nu^p \sum_{\f r \in [\delta,1/\delta]_\nu^p} \, \ee^{-\kappa \abs{\f r}} \, y_{\f x, \tilde {\f x}}(\f r, \f r) = 
\int_{[\delta,1/\delta]^p} \dd \f r \, \ee^{-\kappa \abs{\floor{\f r}_\nu}} \, y_{\f x, \tilde {\f x}}(\floor{\f r}_\nu, \floor{\f r}_\nu) + C_{p,d,L,\delta} \, O (\nu)\,.
\end{equation}
The error term contains errors arising from the boundary of the Riemann sum approximation.

Next, for each $\f r \in [\delta, 1/\delta]^p$ choose $\f s \equiv \f s(\f r)$ defined by
\begin{equation*}
s_i \deq r_i - \sqrt{\nu}\,.
\end{equation*}
Hence, for $\nu$ small enough (depending on $\delta$), we have for all $i$,
\begin{equation*}
\delta/2 \leq s_i \leq \floor{r_i}_\nu\,, \qquad \max_i(\floor{r_i}_\nu - s_i) \leq C \sqrt{\nu} \,, \qquad  \min_i (\floor{r_i}_\nu - s_i) \geq \sqrt{\nu} / C\,,
\end{equation*}
where in the third inequality we used that $\abs{\f r - \floor{\f r}_\nu} \leq \nu$.

Hence, by Lemma \ref{lem:y_psi} we have
\begin{align}
\absbb{\int_{[\delta,1/\delta]^p} \dd \f r \, \ee^{-\kappa \abs{\f r}} \, y_{\f x, \tilde {\f x}}(\f r, \f r)
- \int_{[\delta,1/\delta]^p} \dd \f r \, \ee^{-\kappa \abs{\f r}} \, y_{\f x, \tilde {\f x}}(\f r, \f s(\f r))} \leq C_{p,d,L,\delta,v,\kappa} \sqrt{\nu}\,,
\\
\absbb{\int_{[\delta,1/\delta]^p} \dd \f r \, \ee^{-\kappa \abs{\floor{\f r}_\nu}} \, y_{\f x, \tilde {\f x}}(\floor{\f r}_\nu, \floor{\f r}_\nu) - \int_{[\delta,1/\delta]^p} \dd \f r \, \ee^{-\kappa \abs{\floor{\f r}_\nu}} \, y_{\f x, \tilde {\f x}}(\floor{\f r}_\nu, \f s(\f r))} \leq C_{p,d,L,\delta,v,\kappa} \sqrt{\nu}\,,
\end{align}
where in the second inequality we used that $\abs{\f r - \floor{\f r}_\nu} \leq \nu$.

Moreover, by Lemma \ref{lem:r_tilde_r} we have
\begin{equation} \label{RS_est2}
\absbb{\int_{[\delta,1/\delta]^p} \dd \f r \, \ee^{-\kappa \abs{\f r}} \, y_{\f x, \tilde {\f x}}(\f r, \f s(\f r)) - \int_{[\delta,1/\delta]^p} \dd \f r \, \ee^{-\kappa \abs{\floor{\f r}_\nu}} \, y_{\f x, \tilde {\f x}}(\floor{\f r}_\nu, \f s(\f r))} \leq C_{p,d,L,\delta, \kappa} \sqrt{\nu}\,.
\end{equation}
Combining \eqref{RS_est1}--\eqref{RS_est2} yields the claim.
\end{proof}

\paragraph{Step 3}
Using Lemma \ref{lem:sigma_zeta} we write
\begin{multline} \label{y_RS3}
\nu^p \sum_{\f r \in [\delta,1/\delta]_\nu^p}\, \ee^{-\kappa \abs{\f r}} \, y_{\f x, \tilde {\f x}}(\f r, \f r)
\\
=
\nu^p \sum_{\f r \in [\delta, 1/\delta]_\nu^p} \ee^{-\kappa \abs{\f r}}
\int \mu_{\cal C_\eta}(\dd \sigma)\, \ee^{f_2(\ang{\sigma})} \, \prod_{i = 1}^p \pBB{\int \bb W^{r_i, 0}_{x_i, \tilde x_i}(\dd \omega_i) \, \pB{\ee^{\ii \int_0^{r_i} \dd t \, \ang{\sigma} (\omega_i(t))} - 1}}\,.
\end{multline}
We compare the right-hand side of \eqref{y_RS3} with
\begin{multline}
\nu^p \sum_{\f r \in [\delta, 1/\delta]_\nu^p} \ee^{-\kappa \abs{\f r}} \, Y_{\f x, \tilde {\f x}}(\f r)
\\
= \nu^p \sum_{\f r \in [\delta, 1/\delta]_\nu^p} \ee^{-\kappa \abs{\f r}} \,
\int \mu_{\cal C_\eta}(\dd \sigma)\, \ee^{F_2(\sigma)} \, \prod_{i = 1}^p \pBB{\int \bb W^{r_i, 0}_{x_i, \tilde x_i}(\dd \omega_i) \, \pB{\ee^{\ii \int_0^{r_i} \dd t \, \sigma([t]_\nu, \omega_i(t))} - 1}}\,.
\end{multline}

Recalling Lemmas \ref{lem:Q_delta}, \ref{lem:q_delta}, and \ref{lem:y_sum}, as well as Cauchy-Schwarz combined with Proposition \ref{prop:fif_renorm} (iv), Proposition \ref{prop:Z_field_Wick_ordered} (ii), Lemma \ref{Ff_L2conv}, and \eqref{W_kernel_est}, we find that to prove Proposition \ref{prop:Qq} it suffices to prove the following result.

\begin{lemma}
Fix $\delta > 0$. Then
\begin{equation*}
\lim_{\nu \to 0} \int \prod_{i = 1}^p \bb W^{r_i, 0}_{x_i, \tilde x_i}(\dd \omega_i) \, \pBB{\int \mu_{\cal C_\eta}(\dd \sigma)\, \absBB{\prod_{i = 1}^p \pB{\ee^{\ii \int_0^{r_i} \dd t \, \ang{\sigma}(\omega_i(t))} - 1} - \prod_{i = 1}^p \pB{\ee^{\ii \int_0^{r_i} \dd t \, \sigma([t]_\nu,\omega_i(t))} - 1}}^2}^{1/2} = 0
\end{equation*}
uniformly in $\f r \in [\delta,1/\delta]_\nu^p$, $\eta > 0$, and $\f x, \tilde {\f x}$.
\end{lemma}

\begin{proof}
By telescoping and using \eqref{W_kernel_est}, it suffices to prove
\begin{equation*}
\lim_{\nu \to 0} \int \bb W^{r_i, 0}_{x_i, \tilde x_i}(\dd \omega_i) \, \pBB{\int \mu_{\cal C_\eta}(\dd \sigma)\, \absB{\ee^{\ii \int_0^{r_i} \dd t \,  \ang{\sigma}(\omega_i(t))} - \ee^{\ii \int_0^{r_i} \dd t \, \sigma([t]_\nu,\omega_i(t))}}^2}^{1/2}  = 0
\end{equation*}
for all $i$.
We write
\begin{equation*}
\absB{\ee^{\ii \int_0^{r_i} \dd t \,  \ang{\sigma}(\omega_i(t))} - \ee^{\ii \int_0^{r_i} \dd t \, \sigma([t]_\nu,\omega_i(t))}} = \abs{\ee^{\ii \scalar{\sigma}{f}} - 1}
\end{equation*}
with
\begin{equation*}
f(\tau,x) \deq  \int_0^{r_i} \dd t \, \delta(x - \omega_i(t)) \pbb{\frac{1}{\nu} - \delta(\tau - [t]_\nu)}\,.
\end{equation*}
From Lemma \ref{lem:T_p} with $p = 1$ we therefore get
\begin{align*}
&\qquad \int \mu_{\cal C_\eta}(\dd \sigma) \, \abs{\ee^{\ii \scalar{\sigma}{f}} - 1}^2 \leq C \scalar{f}{\cal C_\eta f}
\\
&
= C \nu \int_0^{r_i} \dd t  \, \dd \tilde t \, v_\eta(\omega_i(t) - \omega_i(\tilde t)) \int_0^\nu \dd \tau\, \dd \tilde \tau  \, \delta_{\eta, \nu}(\tau - \tilde \tau) \pbb{\frac{1}{\nu} - \delta(\tau - [t]_\nu)} \pbb{\frac{1}{\nu} - \delta(\tilde \tau - [\tilde t]_\nu)}
\\
&
= C \nu \int_0^{r_i} \dd t  \, \dd \tilde t \, \delta_{\eta, \nu}([t]_\nu - [\tilde t]_\nu) \, v_\eta(\omega_i(t) - \omega_i(\tilde t)) -\int_0^{r_i} \dd t  \, \dd \tilde t \, v_\eta(\omega_i(t) - \omega_i(\tilde t)) 
\\
&= \frac{C}{\nu} \sum_{s, \tilde s \in [0,r_i)_\nu} \int_0^\nu \dd u_1\, \dd \tilde u_1 \, \dd u_2 \, \dd \tilde u_2 \, \delta_{\eta, \nu}(u_1 - \tilde u_1)
\\
&\qquad \times \pB{v_\eta\pb{\omega_i(s + u_1) - \omega_i(\tilde s + \tilde u_1)} - v_\eta\pb{\omega_i(s + u_2) - \omega_i(\tilde s + \tilde u_2)}}\,,
\end{align*}
where we defined
\begin{equation*}
[a,b)_\nu \deq [a,b-\nu]_\nu\,.
\end{equation*}
By Cauchy-Schwarz and Lemma \ref{properties_of_v_eta} (ii), it suffices to prove that for any $\alpha > 0$,
\begin{equation} \label{W_alpha_est}
\int \bb W^{r_i, 0}_{x_i, \tilde x_i}(\dd \omega_i) \, \frac{1}{\nu} \sum_{s, \tilde s \in [0,r_i]_\nu} \int_0^\nu \dd u_1\, \dd \tilde u_1 \, \dd u_2 \, \dd \tilde u_2 \, \delta_{\eta, \nu}(u_1 - \tilde u_1) \ind{\abs{\omega_i(s + u_1) - \omega_i(\tilde s + \tilde u_1) - \omega_i(s + u_2) + \omega_i(\tilde s + \tilde u_2)}_L > 2 \alpha}
\end{equation}
tends to $0$ as $\nu \to 0$. By the triangle inequality and symmetry of the indices with and without tilde, we find that \eqref{W_alpha_est} is bounded by
\begin{align*}
&\qquad 2 \int \bb W^{r_i, 0}_{x_i, \tilde x_i}(\dd \omega_i) \, \frac{1}{\nu} \sum_{s, \tilde s \in [0,r_i)_\nu} \int_0^\nu \dd u_1\, \dd \tilde u_1 \, \dd u_2 \, \dd \tilde u_2 \, \delta_{\eta, \nu}(u_1 - \tilde u_1) \ind{\abs{\omega_i(s + u_1) - \omega_i(s + u_2)}_L > \alpha}
\\
& \leq \frac{2 r_i}{\nu} \int \bb W^{r_i, 0}_{x_i, \tilde x_i}(\dd \omega_i) \, \sum_{s \in [0,r_i)_\nu} \int_0^\nu \dd u_1\, \dd u_2 \, \ind{\abs{\omega_i(s + u_1) - \omega_i(s + u_2)}_L > \alpha}
\\
&\leq \frac{2 r_i}{\nu \alpha^2} \sum_{s \in [0,r_i)_\nu} \int_0^\nu \dd u_1\, \dd u_2 \, \int \bb W^{r_i, 0}_{x_i, \tilde x_i}(\dd \omega_i) \, \abs{\omega_i(s + u_1) - \omega_i(s + u_2)}_L^2
\\
&\leq C_{d,L,\delta} \frac{1}{\nu \alpha^2} \sum_{s \in [0,r_i]_\nu} \int_0^\nu \dd u_1\, \dd u_2 \, \nu
\\
&\leq C_{d,L,\delta} \frac{\nu}{\alpha^2}\,,
\end{align*}
where in the second step we used Chebyshev's inequality and in the third step we used Lemma \ref{lem:P_cont} together with the observation that $\psi^{r_i}(x_i-\tilde x_i) \leq C_{d,L,\delta}$.
This concludes the proof.
\end{proof}

The proof of Proposition \ref{prop:Qq} is now complete, and thus also of Proposition \ref{prop:main_conv}.

\section{Proofs of Theorems \ref{thm:L2}, \ref{thm:density}, and \ref{thm:main}} \label{sec:proofs_conclusion}

In this section we prove our main results --  Theorems \ref{thm:L2}, \ref{thm:density}, and \ref{thm:main}.

\begin{proof}[Proof of Theorem \ref{thm:main}]
We estimate
\begin{equation*}
\norm{\nu^p \wh \Gamma_p - \wh \gamma_p}_{L^\infty} \leq \nu^p \norm{ \wh \Gamma_p - \wh \Gamma_{p,\eta}}_{L^\infty} + \norm{\nu^p \wh \Gamma_{p,\eta} - \wh \gamma_{p,\eta}}_{L^\infty} + \norm{\wh \gamma_{p,\eta} - \wh \gamma_p}_{L^\infty}
\end{equation*}
From Propositions \ref{prop:gamma_hat}, \ref{prop:gamma_regularization}, and \ref{prop:main_conv}, we deduce that the right-hand side converges to $0$ as $\nu \to 0$, by choosing $\nu$ small enough and then sending $\eta \to 0$.

What remains is to show that $\wh \Gamma_p$ is continuous. By Proposition \ref{prop:gamma_hat}, it suffices to show that $\wh \Gamma_{p,\eta}$ is continuous for any $\eta > 0$. This is a simple consequence of the definition \eqref{def_eh_gamma_eta}, dominated convergence, and the fact that for any continuous field $\sigma$, the Green function $\pb{K(-\kappa + \ii \sigma)^{-1}}^{0,0}_{x, \tilde x}$ is continuous in $(x, \tilde x)$, by Lemma \ref{lem:G_K} and \eqref{W_FK}.
\end{proof}

\begin{proof}[Proof of Theorem \ref{thm:L2}]
The convergence of $Z / Z^0$ to $\zeta$ follows from Propositions \ref{prop:fif_renorm} (iii), \ref{prop:Z_field_Wick_ordered} (i), and \ref{prop:conv_Z}.
To prove the convergence of the reduced density matrices, by Theorem \ref{thm:main}, the forms \eqref{gamma_renormalized} and \eqref{def_Gamma_renormalized}, and an induction argument in $p$, it suffices to show that
\begin{equation*}
\nu^p \Gamma_p^0 \overset{L^r}{\longrightarrow} \gamma_p^0
\end{equation*}
as $\nu \to 0$, for all $p \in \N^*$, with $r$ given in Theorem \ref{thm:L2}. To that end, we recall that by Wick's rule and the Wick theorem (Lemma \ref{Quantum Wick theorem}) we have
\begin{equation*}
\gamma_p^0 = p! P_p \pbb{\frac{1}{h}}^{\otimes p}\,, \qquad
\nu^p \Gamma_p^0 = p! P_p \pbb{\frac{\nu}{\ee^{\nu h} - 1}}^{\otimes p}\,.
\end{equation*}
Hence it suffices to show that
\begin{equation} \label{h_nu_conv}
\frac{\nu}{\ee^{\nu h} - 1} \overset{L^r}{\longrightarrow} \frac{1}{h}
\end{equation}
as $\nu \to 0$. This can be proved using Sobolev embedding. Alternatively, we write
\begin{equation*}
\frac{\nu}{\ee^{\nu h} - 1} - \frac{1}{h} = \nu \sum_{n \geq 1} \ee^{- \nu n h} - \int_0^\infty \dd t \, \ee^{-t h}
\end{equation*}
and estimate
\begin{equation*}
\normbb{\frac{\nu}{\ee^{\nu h} - 1} - \frac{1}{h}}_{L^r} \leq \int_0^\infty \dd t \, f_\nu(t)
\end{equation*}
where we defined
\begin{equation*}
f_\nu(t) \deq \norm{\ee^{-\ceil{t}_\nu h} - \ee^{-t h}}_{L^r}\,, \qquad \ceil{t}_\nu \deq \min \h{s \in \nu \Z \col s \geq t}\,.
\end{equation*}
Since $\norm{\psi^t}_{L^\infty} \leq C \pb{\frac{1}{L^d} + \frac{1}{t^{d/2}}}$ by a simple Riemann sum estimate, we deduce by interpolation that
\begin{equation*}
\norm{\psi^t}_{L^r} \leq C \pbb{\frac{1}{L^d} + \frac{1}{t^{d/2}}}^{1 - 1/r}\,.
\end{equation*}
We conclude
\begin{equation*}
f_\nu(t) \leq 2 \norm{\ee^{-t h}}_{L^r} \leq C \ee^{-\kappa t} \pbb{\frac{1}{L^d} + \frac{1}{t^{d/2}}}^{1 - 1/r}\,,
\end{equation*}
which is integrable by assumption on $r$. Since $\lim_{\nu \to 0} f_\nu(t) = 0$ for all $t > 0$, we deduce \eqref{h_nu_conv}.
\end{proof}

\begin{proof}[Proof of Theorem \ref{thm:density}]
We use the notations from Definition \ref{def:wick_cor} below.
For $I, \tilde I \subset [p]$ we define $\cal B(I, \tilde I)$ as the set of bijections from $I$ onto $\tilde I$ (which is empty unless $\abs{I} = \abs{\tilde I}$), and we denote by $\cal B^*(I, \tilde I)$ the subset of bijections without fixed points. Using the second identity of \eqref{Wick_kernels} for $\Gamma$ and applying Wick's theorem, Lemma \ref{Quantum Wick theorem}, to the factor $\Gamma^0_{{\f x}_{I^c}, {\tilde {\f x}}_{\tilde{I}^{c}}}$, we obtain
\begin{equation}\label{eq:Gammap-for} \Gamma_{\f x, \tilde {\f x}}  =\sum_{I, \tilde I  \subset [p]} \sum_{\lambda \in \cal B(I^c, \tilde I^c)} \widehat{\Gamma}_{{\f x}_I, {\tilde {\f x}}_{\tilde I}} \prod_{j \in I^c} \Gamma^0_{x_j,\tilde x_{\lambda (j)}}\,. \end{equation}

Next, we expand the expectation of the product $\prod_{j=1}^p (a^*(x_j) a(x_j) - (\Gamma_1^0)_{x_j,x_j})$ and we find that the correlation function of the quantum Wick-ordered particle densities is
\[ \Tr_{\cal F} \pbb{\nu^p \, \wick{\fra N(x_1)} \cdots \wick{\fra N(x_p)} \, \frac{\ee^{-H}}{Z}} = \nu^p \sum_{I \subset [p]} (-1)^{\abs{I^c}} \, \Gamma_{{\f x}_I,{\f x}_I} \prod_{i \in I^c} \Gamma^0_{x_i, x_i}\,. \]
Plugging in \eqref{eq:Gammap-for}, and interchanging the order of summation, we conclude
\begin{align*}
&\mspace{-20mu}\Tr_{\cal F} \pbb{\nu^p \, \wick{\fra N(x_1)} \cdots \wick{\fra N(x_p)} \, \frac{\ee^{-H}}{Z}}
\\
&= \nu^p \sum_{I \subset [p]} (-1)^{\abs{I^c}} \prod_{i \in I^c} \Gamma^0_{x_i,x_i}  \sum_{J, \tilde J \subset I}  \sum_{\lambda \in \cal B(I \setminus J, I \setminus \tilde J) }  \widehat{\Gamma}_{{\f x}_J, {\f x}_{\tilde J}} \prod_{j \in I \setminus J} \Gamma^0_{x_i,x_{\lambda (i)}}
\\
&= \nu^p \sum_{J, \tilde J \subset [p]} \widehat{\Gamma}_{{\f x}_J, {\f x}_{\tilde J}} \sum_{J \cup \tilde J \subset I \subset [p]} (-1)^{\abs{I^c}} \sum_{\lambda \in \cal B(I \setminus J, I \setminus \tilde J) } \prod_{i \in I^c} \Gamma^0_{x_i,x_i}      \prod_{j \in I \setminus J} \Gamma^0_{x_i,x_{\lambda (i)}}\,.
\end{align*}
Next, we claim that, for any $J, \tilde J \subset [p]$,
\begin{equation}\label{eq:claim2}
\sum_{J \cup \tilde J \subset I \subset [p]} (-1)^{\abs{I^c}} \sum_{\lambda \in \cal B(I \setminus J, I \setminus \tilde J) } \prod_{i \in I^c} \Gamma^0_{x_i,x_i}      \prod_{j \in I \setminus J} \Gamma^0_{x_i,x_{\lambda (i)}}
= \sum_{\delta \in \cal B^*(J^c, \tilde J^c)}  \prod_{i \in J^c} \Gamma^0_{x_i, x_{\delta (i)}}\,, \end{equation}
whose proof we defer to the end of the proof.
With \eqref{eq:claim2}, we obtain
\begin{equation}\label{Lr_convergence_1}
\Tr_{\cal F} \pbb{\nu^p \, \wick{\fra N(x_1)} \cdots \wick{\fra N(x_p)} \, \frac{\ee^{-H}}{Z}}  = 
\sum_{J, \tilde J \subset [p]} \nu^{\abs{J}} \widehat{\Gamma}_{{\f x}_J,{\f x}_{\tilde J}} \sum_{\delta \in \cal B^*(J^c, \tilde J^c)} \prod_{i \in J^c} \nu \Gamma^0_{x_i, x_{\delta (i)}}\,,
\end{equation}
which is the key identity of our proof, as it relates the density correlation functions to the Wick-ordered reduced density matrices, with weights given by \emph{off-diagonal} elements of the free correlation functions. The same derivation yields
\begin{equation}
\frac{1}{\zeta} \int \mu_{h^{-1}}(\dd \phi) \, \ee^{-W(\phi)} \, \wick{\fra n_K(x_1)} \cdots \wick{\fra n_K(x_p)} =
\sum_{J, \tilde J \subset [p]} \nu^{\abs{J}} \widehat{\gamma}^K_{{\f x}_J,{\f x}_{\tilde J}} \sum_{\delta \in \cal B^*(J^c, \tilde J^c)} \prod_{i \in J^c} \nu \gamma^{K,0}_{x_i, x_{\delta (i)}}\,,
\end{equation}
where $\gamma^{K}$ and $\gamma^{K,0}$ denote the Wick-ordered correlation function and the free correlation function, respectively, of the truncated field $\cal P_K \phi$ (recall \eqref{def_Pk}). For distinct $x_1, \dots, x_p$, we find from Lemma \ref{lem:conv_n_Wick}, from $\lim_{K \to \infty} (\cal P_K h^{-1})_{x,y} = (h^{-1})_{x,y}$ for $x \neq y$, and a simple approximation argument analogous to the one from the proof of Proposition \ref{prop:gamma_weak}, that
\begin{equation} \label{Lr_convergence_2}
\frac{1}{\zeta} \int \mu_{h^{-1}}(\dd \phi) \, \ee^{-W(\phi)} \, \wick{\fra n(x_1)} \cdots \wick{\fra n(x_p)} =
\sum_{J, \tilde J \subset [p]} \widehat{\gamma}_{{\f x}_J,{\f x}_{\tilde J}} \sum_{\delta \in \cal B^*(J^c, \tilde J^c)} \prod_{i \in J^c} \gamma^{0}_{x_i, x_{\delta (i)}}\,.
\end{equation}

From \eqref{Lr_convergence_1} and \eqref{Lr_convergence_2}, noting that each off-diagonal free correlation function $\gamma^{0}_{x_i, x_{\delta (i)}}$ appears at most twice in the product (as $\gamma^{0}_{x_i, x_j}$ and $\gamma^{0}_{x_j, x_i}$ with $j = \delta(i)$ and $i = \delta(j)$), using Theorem \ref{thm:main}, as well as \eqref{h_nu_conv}, we conclude the proof of Theorem \ref{thm:density}.

What remains is the proof of \eqref{eq:claim2}.
To this end, we observe that, with every $\lambda\in \cal B(I \setminus J, I \setminus \tilde J)$ we can associate $\delta \in \cal B(J^c, \tilde J^c)$ defined by $\delta (i) \deq \lambda (i)$ if $i \in I \setminus J$, and $\delta (i) = i$ if $i \in I^c$. Hence,
\begin{equation*}
\prod_{i \in I^c} \Gamma^0_{x_i, x_i} \prod_{j \in I \setminus J} \Gamma^0_{x_j, x_{\lambda (j)}} = \prod_{i \in J^c} \Gamma^0_{x_i, x_{\delta (i)}}\,.
\end{equation*}
This already implies that 
\begin{equation}
 \label{I_lambda_delta}
\sum_{J \cup \tilde J \subset I \subset [p]} (-1)^{\abs{I^c}} \sum_{\lambda \in \cal B(I \setminus J, I \setminus \tilde J) } \prod_{i \in I^c} \Gamma^0_{x_i,x_i}      \prod_{j \in I \setminus J} \Gamma^0_{x_i,x_{\lambda (i)}}
 = \sum_{\delta \in \cal B(J^c, \tilde J^c)} c_\delta \prod_{i \in J^c} \Gamma^0_{x_i, x_{\delta (i)}} 
\end{equation}
for some coefficients $c_\delta \in \mathbb{R}$ (in fact, it is already clear that $c_\delta \in \mathbb{Z}$). For the following argument, we take  $\delta \in \cal B(J^c,\tilde{J}^c)$ that has  $n \deq | \{ i \in J^c : \delta (i) = i \} |$ fixed points. Note that  $n \leq p - |J \cup \tilde J|$  (because, by definition, points in  $J \cup \tilde J$  cannot be fixed points of $\delta$). We shall count the total contribution of all terms on the left-hand side of \eqref{I_lambda_delta} to the term associated with $\delta$ on its right-hand side.

Suppose first that $n > 0$. Then the term associated with $\delta$ arises from one term with $I = [p]$, from ${n \choose 1}$ terms with $|I| = p-1$, choosing  $J \cup \tilde J \subset I \subset [p]$   (we can take $I = [p] \setminus \{ i \}$, for any $i \in J^c$ with $\delta (i) = i$), and in general from ${n \choose k}$ terms with $|I| = p-k$, for any $0 \leq k \leq n$. Thus, we find 
\[ c_\delta = \sum_{k=0}^{n} {n \choose k} (-1)^{k} = 0 \]
If instead $n = 0$, then the term associated with $\delta$ arises from a single term with $I = [p]$. Therefore, in this case, $c_\delta = 1$. This concludes the proof of \eqref{eq:claim2}. 
\end{proof}

\section{Trapping potentials: proof of Theorem \ref{thm:trap}} \label{sec:traps}

The proof of Theorem \ref{thm:trap} is very similar to that of Theorem \ref{thm:main}, and in this section we explain the required modifications.

We follow the argument of Sections \ref{sec:fct_quantum}--\ref{sec:proofs_conclusion}, with minor adjustments. Essentially, we replace every Brownian measure $\bb W^{\tau, \tilde \tau}_{x, \tilde x}(\dd \omega)$ with the \emph{weighted Brownian measure}
\begin{equation*}
\bb W^{\tau, \tilde \tau}_{x, \tilde x}(\dd \omega) \, \ee^{-\int_{\tilde \tau}^\tau \dd t \, U(\omega(t))}\,,
\end{equation*}
so that, for instance, \eqref{e_H_n} becomes
\begin{equation*}
(\ee^{- H_n})_{\f x, \tilde{\f x}} = \ee^{-\nu \kappa n} \int \prod_{i = 1}^n \pbb{\bb W^{\nu,0}_{x_i, \tilde x_i}(\dd \omega_i) \, \ee^{-\int_0^\nu \dd t \, U(\omega_i(t))}} \, \ee^{-\frac{\lambda}{2 \nu} \sum_{i,j = 1}^n \int_0^\nu \dd t \, v(\omega_i(t) - \omega_j(t))}\,.
\end{equation*}
Accordingly, we replace the propagator $W(-\kappa + \ii \sigma)$ and the Green function $K(-\kappa + \ii \sigma)$ with $W(-\kappa - U + \ii \sigma)$ and $K(-\kappa - U + \ii \sigma)$, respectively.
The only changes to our argument arise from the need to obtain sufficient decay in the spatial parameters $x, \tilde x$ from such integrals of Brownian paths, which will in turn allow us to integrate over them and to establish the claimed decay in $\f x, \tilde {\f x}$ in Theorem \ref{thm:trap}. We shall obtain this decay from an analysis of the excursion probabilities of Brownian bridges.

We shall prove the two following results.

\begin{proposition} \label{prop:trap_conv}
Under the assumptions of Theorem \ref{thm:trap}, for any $p \in \N^*$, we have $\lim_{\nu \to 0} \nu^p \wh \Gamma_{p,\nu} = \wh \gamma_{p}$ locally uniformly.
\end{proposition}

\begin{proposition} \label{prop:trap_bound}
Under the assumptions of Theorem \ref{thm:trap}, for any $p \in \N^*$, there exist constants $C,c > 0$, depending on $d,v,\kappa,p,b,\theta$, such that $\abs{\nu^p (\wh \Gamma_{p})_{\f x, \tilde {\f x}}} \leq C \Upsilon_{\theta,c}(\f x, \tilde {\f x})$, where we recall the definition of $\Upsilon_{\theta,c}$ from  \eqref{def_Upsilon_p}.
\end{proposition}

\begin{proof}[Proof of Theorem \ref{thm:trap}]
By Propositions \ref{prop:trap_conv} and \ref{prop:trap_bound} we have $\abs{(\wh \gamma_p)_{\f x, \tilde {\f x}}} \leq C \Upsilon_{\theta,c}(\f x, \tilde {\f x})$.
Choose $\delta > 0$. Then choose $R \equiv R_\delta$ such that
\begin{equation*}
\sup_{\abs{\f x} + \abs{\tilde {\f x}} > R} \frac{\abs{\nu^p (\Gamma_{p})_{\f x, \tilde {\f x}} - (\gamma_{p})_{\f x, \tilde {\f x}}}}{\Upsilon_{\theta - \epsilon,c}(\f x, \tilde {\f x})} \leq 
\sup_{\abs{\f x} + \abs{\tilde {\f x}} > R} \frac{C \Upsilon_{\theta,c}(\f x, \tilde {\f x})}{\Upsilon_{\theta - \epsilon,c}(\f x, \tilde {\f x})} \leq C R^{-\epsilon (2 - d/2)} \leq \delta\,.
\end{equation*}
By Proposition \ref{prop:trap_conv} we have
\begin{equation*}
\lim_{\nu \to 0} \sup_{\abs{\f x} + \abs{\tilde {\f x}} \leq R} \frac{\abs{\nu^p (\Gamma_{p})_{\f x, \tilde {\f x}} - (\gamma_{p})_{\f x, \tilde {\f x}}}}{\Upsilon_{\theta - \epsilon,c}(\f x, \tilde {\f x})} = 0\,,
\end{equation*}
and the proof is complete.
\end{proof}

The rest of this section is devoted to the proofs of Proposition \ref{prop:trap_conv} and \ref{prop:trap_bound}. Throughout the following we adopt the notations of Sections \ref{sec:preliminaries}--\ref{sec:mf} without further comment, setting $L = \infty$ there, and replacing all Brownian measures with their weighted versions.

As for the torus, we use a regularized covariance \eqref{cov_C_eta}, which should ensure that for $\eta > 0$ the field $\sigma$ is $\mu_{\cal C_\eta}$-almost surely smooth and bounded. To that end, we choose use the cutoff function $\varphi$ from Section \ref{sec:FT} and define the regularized covariance
\begin{equation} \label{cov_C_eta_Rd}
\int \mu_{\cal C_\eta}(\dd \sigma)\, \sigma(\tau, x) \, \sigma(\tilde \tau, \tilde x) = \frac{\lambda}{\nu} \, \delta_{\eta, \nu}(\tau - \tilde \tau) \, v_\eta (x,\tilde x) \eqd (\cal C_\eta)_{x, \tilde x}^{\tau, \tilde \tau}\,,
\end{equation}
where
\begin{equation*}
v_\eta(x, \tilde x) \deq (v * \delta_{\eta,\infty,\varphi})(x - \tilde x) \, \varphi(\eta x) \, \varphi(\eta \tilde x)\,.
\end{equation*}

\begin{proof}[Proof of Proposition \ref{prop:trap_conv}]
We repeat the arguments of Sections \ref{sec:fct_quantum}--\ref{sec:proofs_conclusion} with trivial adjustments. Since $\f x$ and $\tilde {\f x}$ are bounded, the only difference is the need to control the traces and integrations over spatial variables. For the traces, we use the assumption \eqref{tr_h_assump} for $s = 2$. For the integrations over spatial variables, we use
\begin{equation*}
\int \dd x \, \int \bb W^{t,0}_{x,x}(\dd \omega) \, \ee^{-\int_0^t \dd s \, U(\omega(s))} = \tr \ee^{-t h}\,.
\end{equation*}
Thus, for instance, the proof of Lemma \ref{lem:unif_int} is taken care of by the estimate
\begin{multline} \label{U_nu_conv}
\int_{[0,\infty)^3} \dd \f r \, \frac{1}{\abs{\f r}}
\, \int \dd x \int \bb W_{x, x}^{\abs{\f r}, 0}(\dd \omega) \, \ee^{-\int_0^{\abs{\f r}} \dd s \, U(\omega(s))}\, 
(1+r_1)^4
\\
=
\int_{[0,\infty)^3} \dd \f r \, \frac{1}{\abs{\f r}}
\tr \ee^{-\abs{\f r} h}  (1+r_1)^4 \leq C (\tr h^{-2} + \tr h^{-6})\,.
\end{multline}
Finally, to compare the weighted Brownian measure of $U$ to that of $U_0$, we use that
\begin{equation*}
t^{d/2} \int \bb W^{t,0}_{x,\tilde x}(\dd \omega) \pbb{\ee^{-\int_0^t \dd s \, U(\omega(s))} - \ee^{- \int_0^t \dd s \, U_0(\omega(s))}} \to 0
\end{equation*}
as $\nu \to 0$, uniformly in $t > 0$ and $x, \tilde x \in \Lambda$, as follows easily from the assumption on the convergence of $U$ to $U_0$. Replacing $U$ with $U_0$ in expressions like the left-hand side of \eqref{U_nu_conv} is done using the assumption $U \geq U_0 / C$ and dominated convergence.
\end{proof}

In order to prove Proposition \ref{prop:trap_bound}, we shall need the following fundamental estimate on the kernel of $\ee^{-\tau(-\Delta /2 + U)}$.

\begin{lemma} \label{lem:prop_est}
Let $\theta \geq 2$ and $b > 0$, and suppose that $U(x) \geq b \abs{x}^\theta$. Then
\begin{multline} \label{prop_est4}
\int \bb W^{\tau,0}_{x, \tilde x}(\dd \omega) \, \ee^{-\int_0^\tau \dd t \, U(\omega(t))}
\\
\leq C \psi^\tau(x - \tilde x) \pbb{\ind{\tau \leq (\abs{x} + \abs{\tilde x})^{-2 (\theta + 1)}} + \ee^{-c (\abs{x} + \abs{\tilde x})^\theta \tau} + \ee^{-c (\abs{x} + \abs{\tilde x})^{1 + \theta/2}} + \ee^{-c (\sqrt{\tau} (\abs{x} + \abs{\tilde x})^{\theta + 1})^{2/3}}}
\end{multline}
for some constants $C,c > 0$ depending on $\theta$ and $b$.
\end{lemma}
\begin{proof}
By symmetry, we may assume that $\abs{\tilde x} \geq \abs{x}$. Moreover, we may assume that $\abs{\tilde x} \geq 1$, for otherwise the claim is an immediate consequence of Lemma \ref{FK_continuous}.

Under the law $\P^{\tau,0}_{x, \tilde x}(\dd \omega)$, the Brownian bridge $\omega$ can be written in terms of standard Brownian motion $B(t) = (B_i(t))_{i=1}^d$ in $\R^d$ as $\omega(t) = \tilde x + B(t) + \frac{t}{\tau} (x - \tilde x - B(\tau))$, as can be verified by comparing the finite-dimensional distributions. Denoting by $\P$ and $\E$ the law of standard Brownian motion $(B(t))_{t \geq 0}$ and the associated expectation, we therefore have
\begin{equation*}
\int \bb P^{\tau,0}_{x, \tilde x}(\dd \omega) \, \ee^{-\int_0^\tau \dd t \, U(\omega(t))}
\leq \E \exp \pBB{- b \int_0^\tau \dd t \, \absbb{\tilde x + B(t) + \frac{t}{\tau} (x - \tilde x - B(\tau))}^\theta}\,.
\end{equation*}

Next, we define the stopping time
\begin{equation*}
T \deq \inf \hb{t \geq 0 \col \abs{B(t)} = \abs{\tilde x} / 4}\,.
\end{equation*}
For\footnote{We use the notation $x \wedge y \deq \min \{x,y\}$.}
\begin{equation*}
0 \leq t \leq \tau \wedge T \wedge \frac{\tau \abs{\tilde x}}{4 \abs{x - \tilde x - B(\tau)}}
\end{equation*}
we have $\absb{\tilde x + B(t) + \frac{t}{\tau} (x - \tilde x - B(\tau))} \geq \abs{\tilde x} / 2$. Thus,
\begin{equation*}
\int \bb P^{\tau,0}_{x, \tilde x}(\dd \omega) \, \ee^{-\int_0^\tau \dd t \, U(\omega(t))} \leq \E \qbb{\ee^{-c \abs{\tilde x}^\theta \tau} + \ee^{-c \abs{\tilde x}^\theta T} + \ee^{-c \abs{\tilde x}^\theta \frac{\tau \abs{\tilde x}}{ \abs{x - \tilde x - B(\tau)}}}}\,,
\end{equation*}
for some constant $c > 0$ depending on $b$ and $\theta$. Using that
\begin{equation*}
\frac{\tau \abs{\tilde x}}{ \abs{x - \tilde x - B(\tau)}} \geq c \pbb{\frac{\tau \abs{\tilde x}}{ \abs{x - \tilde x}} \wedge  \frac{\tau \abs{\tilde x}}{ \abs{B(\tau)}}}
\geq
c \pbb{\tau \wedge  \frac{\tau \abs{\tilde x}}{\abs{B(\tau)}}}
\end{equation*}
(after a renaming of $c$ in the last step),
we therefore find
\begin{equation} \label{U_prop_est}
\int \bb P^{\tau,0}_{x, \tilde x}(\dd \omega) \, \ee^{-\int_0^\tau \dd t \, U(\omega(t))} \leq \ee^{-c \abs{\tilde x}^\theta \tau} + \E \ee^{-c \abs{\tilde x}^\theta T} + \E \ee^{-c \tau \abs{\tilde x}^{\theta + 1} / \abs{B(\tau)}}\,.
\end{equation}

To analyse the second term of \eqref{U_prop_est}, for $1 \leq i \leq d$ we introduce the stopping times
\begin{equation*}
T_{i,\pm} \deq \inf \hbb{t \geq 0 \col B_i(t) = \frac{\pm \abs{\tilde x}}{4\sqrt{d}}}\,,
\end{equation*}
so that $T \geq \min_{i,\pm} T_i$. Thus,
\begin{equation*}
\P(T \leq t) \leq \P\pB{\min_{i,\pm} T_{i,\pm} \leq t} \leq 2d \P(T_{1,+} \leq t) = \frac{4d}{\sqrt{2\pi}} \int_{\abs{\tilde x} / (4 \sqrt{d t})}^\infty \dd u \, \ee^{-u^2/2} \leq C \ee^{- c \abs{\tilde x}^2 / t}\,,
\end{equation*}
where in the second step we used a union bound and the symmetry of Brownian motion, and the third step follows from a standard application of the reflection principle for the standard one-dimensional Brownian motion $B_i(t)$ to compute the law of $T_{1,+}$. By Fubini's theorem, we therefore conclude that
\begin{equation*}
\E \ee^{-c \abs{\tilde x}^\theta T} = c \abs{\tilde x}^\theta \int_0^\infty \dd t \, \ee^{-c \abs{\tilde x}^\theta t} \, \P(T \leq t) \leq C \abs{\tilde x}^\theta \int_0^\infty \dd t \, \ee^{-c \abs{\tilde x}^\theta t - c \abs{\tilde x}^2 / t}\,.
\end{equation*}
We split the integral into two pieces according to whether $t \leq \abs{\tilde x}^{1 - \theta/2}$. The first piece is
\begin{equation*}
C \abs{\tilde x}^\theta \int_0^{\abs{\tilde x}^{1 - \theta/2}} \dd t \, \ee^{-c \abs{\tilde x}^\theta t - c \abs{\tilde x}^2 / t} \leq C \abs{\tilde x}^\theta \int_0^{\abs{\tilde x}^{1 - \theta/2}} \dd t \, \ee^{-c \abs{\tilde x}^2 / t} \leq C \abs{\tilde x}^{2+\theta} \, \ee^{-c \abs{\tilde x}^{1 + \theta/2}} \leq 
C \ee^{-c \abs{\tilde x}^{1 + \theta/2}}\,,
\end{equation*}
where we used Lemma \ref{lem:el_integral} below. The second piece is
\begin{equation*}
C \abs{\tilde x}^\theta \int_{\abs{\tilde x}^{1 - \theta/2}}^\infty \dd t \, \ee^{-c \abs{\tilde x}^\theta t - c \abs{\tilde x}^2 / t} \leq C \abs{\tilde x}^\theta \int_{\abs{\tilde x}^{1 - \theta/2}}^\infty \dd t \, \ee^{-c \abs{\tilde x}^\theta t}
= C \ee^{-c \abs{\tilde x}^{1 + \theta/2}}\,.
\end{equation*}
We conclude that
\begin{equation} \label{eU_est1}
\E \ee^{-c \abs{\tilde x}^\theta T} \leq C \ee^{-c \abs{\tilde x}^{1 + \theta/2}}\,.
\end{equation}

Next, we estimate the third term of \eqref{U_prop_est}. We consider the cases (i) $\tau \leq \abs{\tilde x}^{-2 (\theta + 1)}$ and (ii) $\tau > \abs{\tilde x}^{-2 (\theta + 1)}$ separately. In case (i) we estimate the third term of \eqref{U_prop_est} by $1$. In case (ii), we use that the probability density of $\abs{B(\tau)}^2 / \tau$ is $C u^{d/2 - 1} \ee^{-u/2}$ (a $d$-dimensional Chi-squared distribution), which yields
\begin{equation*}
\E \ee^{-c \tau \abs{\tilde x}^{\theta + 1} / \abs{B(\tau)}} = C \int_0^\infty \dd u \, u^{d/2 - 1} \, \ee^{-u/2} \, \ee^{-c \sqrt{\tau} \abs{\tilde x}^{\theta + 1} / \sqrt{u}} = C \int_0^\infty \dd t \, t^{d - 1} \, \ee^{-t^2 / 2 - c \sqrt{\tau} \abs{\tilde x}^{\theta + 1} / t}\,.
\end{equation*}
We split the integral into two pieces according to whether $t \leq (\sqrt{\tau} \abs{\tilde x}^{\theta + 1})^{1/3}$. The first piece is
\begin{multline*}
C \int_0^{(\sqrt{\tau} \abs{\tilde x}^{\theta + 1})^{1/3}} \dd t \, t^{d - 1} \, \ee^{-t^2 / 2 - c \sqrt{\tau} \abs{\tilde x}^{\theta + 1} / t} \leq C \int_0^{(\sqrt{\tau} \abs{\tilde x}^{\theta + 1})^{1/3}} \dd t \, t^{d - 1} \, \ee^{- c \sqrt{\tau} \abs{\tilde x}^{\theta + 1} / t}
\\
\leq C (\sqrt{\tau} \abs{\tilde x}^{\theta + 1})^{d} \ee^{-c (\sqrt{\tau} \abs{\tilde x}^{\theta + 1})^{2/3}} \leq 
C \ee^{-c (\sqrt{\tau} \abs{\tilde x}^{\theta + 1})^{2/3}}\,,
\end{multline*}
where in the second step we used Lemma \ref{lem:el_integral} and that we are in case (ii). The second piece is
\begin{equation*}
C \int_{(\sqrt{\tau} \abs{\tilde x}^{\theta + 1})^{1/3}}^\infty \dd t \, t^{d - 1} \, \ee^{-t^2 / 2} \leq C \ee^{-c (\sqrt{\tau} \abs{\tilde x}^{\theta + 1})^{2/3}}\,.
\end{equation*}
Summarizing, we have proved that
\begin{equation} \label{eU_est2}
\E \ee^{-c \tau \abs{\tilde x}^{\theta + 1} / \abs{B(\tau)}} \leq \ind{\tau \leq \abs{\tilde x}^{-2 (\theta + 1)}} + C \ee^{-c (\sqrt{\tau} \abs{\tilde x}^{\theta + 1})^{2/3}}\,.
\end{equation}
The claim now follows by plugging \eqref{eU_est1} and \eqref{eU_est2} into \eqref{U_prop_est}.
\end{proof}

\begin{lemma} \label{lem:el_integral}
If $s \leq 2$ and $u \geq v > 0$ then
\begin{equation*}
\int_0^v \dd t \, t^{-s} \ee^{-u/t} \leq u^{1-s} \ee^{-u/v}\,.
\end{equation*}
\end{lemma}

\begin{proof}[Proof of Proposition \ref{prop:trap_bound}]
Using \eqref{Gamma_Q}, \eqref{def_Q}, \eqref{Q_Y},  and the first estimate of \eqref{Y_est2}, and taking the limit $\eta \to 0$, we find that
\begin{equation*}
\nu^p \abs{(\wh \Gamma_{p})_{\f x, \tilde {\f x}}} \leq C_v \max_{\pi \in S_p} \, \prod_{i = 1}^{p} \pBB{\nu \sum_{r \in \nu \N^*} \ee^{-\kappa r}  r \int  \bb W^{r, 0}_{x_i, \tilde x_{\pi(i)}}(\dd \omega) \, \ee^{-\int_0^{r} \dd t \, U(\omega(t))}}\,.
\end{equation*}
Thus, it suffices to show that
\begin{equation} \label{sum_prop_est}
\nu \sum_{r \in \nu \N^*} \ee^{-\kappa r} \, r \, \int  \bb W^{r, 0}_{x, \tilde x}(\dd \omega) \, \ee^{-\int_0^{r} \dd t \, U(\omega(t))} \leq C (1 + \abs{x} + \abs{\tilde x})^{-\theta (2 - d/2)} \, \ee^{-c \abs{x - \tilde x}}\,.
\end{equation}
Moreover, if $\abs{x}, \abs{\tilde x} \leq 1$ then \eqref{sum_prop_est} is obvious. Hence, for the rest of the proof we assume that $\abs{\tilde x} \geq \abs{x}$ (without less of generality) and $\abs{\tilde x} \geq 1$. We have to show that
\begin{equation} \label{sum_prop_est2}
\nu \sum_{r \in \nu \N^*} \ee^{-\kappa r} \, r \, \int  \bb W^{r, 0}_{x, \tilde x}(\dd \omega) \, \ee^{-\int_0^{r} \dd t \, U(\omega(t))} \leq C \abs{\tilde x}^{-\theta (2 - d/2)} \, \ee^{-c \abs{x - \tilde x}}\,.
\end{equation}

We prove \eqref{sum_prop_est2} using Lemma \ref{lem:prop_est}, so that the left-hand side of \eqref{sum_prop_est2} is bounded by the sum of four terms arising from the four terms on the right-hand side of \eqref{prop_est4}, denoted by (a), (b), (c), (d).

\paragraph{Term \normalfont{(a)}} We have
\begin{equation*}
\text{(a)} = C \nu \sum_{r \in \nu \N^*} \ee^{-\kappa r} \, r^{1-d/2} \, \ee^{-\abs{x - \tilde x}^2 / 2r} \, \ind{r \leq C \abs{\tilde x}^{-2 (\theta + 1)}}\,.
\end{equation*}
If we estimate $\ee^{-\abs{x - \tilde x}^2 / 2r} \leq 1$ we obtain the rough bound
\begin{equation} \label{est_a_rough}
\text{(a)} \leq C \int_0^{C \abs{\tilde x}^{-2 (\theta + 1)}} \dd r \,  r^{1-d/2} \leq C \abs{\tilde x}^{-(4 - d)(\theta + 1)}\,.
\end{equation}
On the other hand, if $\abs{x - \tilde x} \abs{\tilde x}^{2 (\theta + 1)} \geq C$ we use Lemma \ref{lem:el_integral} to obtain
\begin{align}
\text{(a)} &\leq C \int_0^{C \abs{\tilde x}^{-2 (\theta + 1)}} \dd r\,  r^{1-d/2} \, \ee^{-\abs{x - \tilde x}^2 / 2r}
\notag
\\
&\leq C \abs{x - \tilde x}^{2 - d/2} \, \ee^{- c \abs{x - \tilde x}^2 \abs{\tilde x}^{2 (\theta + 1)}}
\notag \\ \label{est_a_sharp}
&\leq C \abs{\tilde x}^{-(4 - d)(\theta + 1)}\, \ee^{- c \abs{x - \tilde x}^2 \abs{\tilde x}^{2 (\theta + 1)}}\,,
\end{align}
and by \eqref{est_a_rough} we find that \eqref{est_a_sharp} holds for all $x, \tilde x$. The right-hand side of \eqref{est_a_sharp} is bounded by the right-hand side of \eqref{sum_prop_est2}.

The strategy for the other three terms is analogous.

\paragraph{Term \normalfont{(b)}}
We have
\begin{equation*}
\text{(b)} = C \nu \sum_{r \in \nu \N^*} \ee^{-\kappa r} \, r^{1-d/2} \, \ee^{-\abs{x - \tilde x}^2 / 2r} \, \ee^{-c \abs{\tilde x}^\theta r}\,.
\end{equation*}
We obtain the rough bound
\begin{equation} \label{term_b_rough}
\text{(b)} \leq C \int_0^\infty \dd r \,  r^{1-d/2} \, \ee^{-c \abs{\tilde x}^\theta r} \leq C \abs{\tilde x}^{-\theta (2 - d/2)}\,.
\end{equation}
For a more precise bound, we estimate
\begin{equation} \label{term_b_est}
\text{(b)} \leq C \int_0^{\abs{x - \tilde x} \abs{\tilde x}^{-\theta/2}} \dd r \, r^{1-d/2} \, \ee^{-\abs{x - \tilde x}^2 / 2r} \, \ee^{-c \abs{\tilde x}^\theta r} + C \int_{\abs{x - \tilde x} \abs{\tilde x}^{-\theta/2}}^\infty \dd r \, r^{1-d/2} \, \ee^{-\abs{x - \tilde x}^2 / 2r} \, \ee^{-c \abs{\tilde x}^\theta r}\,.
\end{equation}
If $\abs{x - \tilde x} \abs{\tilde x}^{\theta/2} \geq C$, the first term of \eqref{term_b_est} is estimated using Lemma \ref{lem:el_integral} by
\begin{equation*}
C \int_0^{\abs{x - \tilde x} \abs{\tilde x}^{-\theta/2}} \dd r \, r^{1-d/2} \, \ee^{-\abs{x - \tilde x}^2 / 2r} \leq C \abs{\tilde x}^{-\theta (2 - d/2)} \, \ee^{-c \abs{x - \tilde x} \abs{\tilde x}^{\theta/2}}\,.
\end{equation*}
The second term of \eqref{term_b_est} is estimated by
\begin{equation*}
C \int_{\abs{x - \tilde x} \abs{\tilde x}^{-\theta/2}}^\infty \dd r \, r^{1-d/2} \, \ee^{-c \abs{\tilde x}^\theta r}
\leq C \abs{\tilde x}^{-\theta (2 - d/2)} \, \ee^{-c \abs{x - \tilde x} \abs{\tilde x}^{\theta/2}}\,.
\end{equation*}
Recalling \eqref{term_b_rough}, we therefore deduce that
\begin{equation*}
\text{(b)} \leq C \abs{\tilde x}^{-\theta (2 - d/2)} \, \ee^{-c \abs{x - \tilde x} \abs{\tilde x}^{\theta/2}}
\end{equation*}
for all $x, \tilde x$. This is bounded by the right-hand side of \eqref{sum_prop_est2}.

\paragraph{Term \normalfont{(c)}}
By splitting the $r$ integral into the regions $r \leq |x-\tilde x|$ and $r>|x-\tilde x|$, an analogous argument yields
\begin{equation*}
\text{(c)} \leq C \int_0^\infty \dd r \,  \ee^{-\kappa r} \, r^{1-d/2} \, \ee^{-\abs{x - \tilde x}^2 / 2r} \, \ee^{-c \abs{\tilde x}^{1 + \theta/2}} \leq C \ee^{-c \abs{\tilde x}^{1 + \theta/2}} \, \ee^{-c \abs{x - \tilde x}}\,,
\end{equation*}
which is bounded by the right-hand side of \eqref{sum_prop_est2}.

\paragraph{Term \normalfont{(d)}} Finally, an analogous argument yields
\begin{equation*}
\text{(d)} \leq C \int_0^\infty \dd r \,  \ee^{-\kappa r} \, r^{1-d/2} \, \ee^{-\abs{x - \tilde x}^2 / 2r} \, \ee^{-c (\sqrt{r} \abs{\tilde x}^{\theta + 1})^{2/3}} \leq C \abs{\tilde x}^{-(4 - d)(\theta + 1)} \, \ee^{-c \abs{x - \tilde x}^{1/2} \abs{\tilde x}^{(\theta + 1)/2}}\,,
\end{equation*}
which is bounded by the right-hand side of \eqref{sum_prop_est2}.

This concludes the proof of \eqref{sum_prop_est2}.
\end{proof}

\appendix

\section{Wick-ordered interaction and correlation functions}

\begin{lemma} \label{lem:constr_W}
Suppose that $v$ is bounded and that \eqref{tr_h_assump} holds with $s = 2$. Let $m \geq 1$.
\begin{enumerate}
\item
The sequence $(W_K^v)_{K \in \N}$ is a Cauchy sequence in $L^m(\mu_{h^{-1}})$, and we denote its limit by $W^v$.
\item
If $v$ and $\tilde v$ are bounded then
\begin{equation*}
\norm{W^v - W^{\tilde v}}_{L^m(\mu_{h^{-1}})} = O(\norm{v - \tilde v}_{L^\infty})\,.
\end{equation*}
\end{enumerate}
\end{lemma}

\begin{proof}[Proof of Lemma \ref{lem:constr_W}]
The proof is a routine application of Wick's rule for the Gaussian measure $\mu_{h^{-1}}$. The estimates are straightforward because the Green function
\begin{equation*}
(\gamma^0_1)_{x, \tilde x} = \int \mu_{h^{-1}} (\dd \phi) \, \bar \phi(\tilde x) \phi(x)
\end{equation*}
is in $L^2(\Lambda^2) \equiv \fra S^2(\cal H)$, by Assumption \ref{ass:torus}. A detailed argument is given in \cite[Lemma 1.5]{frohlich2017gibbs}.
\end{proof}

\begin{lemma} \label{lem:conv_n_Wick}
Let $\wick{\fra n_K}$ be as in \eqref{def_n_K} and suppose that \eqref{tr_h_assump} holds with $s = 2$. For any bounded $f$ and $m \geq 1$, the sequence $(\scalar{f}{\wick{\fra n_K}})_{K \in \N}$ is a Cauchy sequence in $L^m(\mu_{h^{-1}})$, and we denote its limit by $\scalar{f}{\wick{\fra n}}$.
Under $\mu_{h^{-1}}$, the random variable $\scalar{f}{\wick{\fra n}}$ has subexponential tails: there is a constant $C$, depending on $\norm{f}_{L^\infty}$ and $\tr h^{-2}$, such that $\int \dd \mu_{h^{-1}} \, \scalar{f}{\wick{\fra n}}^m \leq (Cm)^m$ for all $m \in \N^*$.
\end{lemma}
\begin{proof}
The proof is a simple consequence of Wick's rule, analogous to the proof of Lemma \ref{lem:constr_W}.
\end{proof}

The following notations are useful when dealing with correlation functions and their Wick-ordered versions.
\begin{definition} \label{def:wick_cor}
Recall the notation $[p] = \{ 1, 2, \dots , p \}$ for $p \in \mathbb{N}^*$. For $I = \{ i_1 < \dots < i_k \}  \subset [p]$, we set ${\f x}_I = (x_{i_1}, \dots , x_{i_k})$ and $I^c \deq [p] \setminus I$.
For $I, \tilde I \subset [p]$, we use the shorthand
\begin{equation*}
\gamma_{\f x_I, \tilde {\f x}_{\tilde I}} \deq \ind{\abs{I} = \abs{\tilde I}} \, (\gamma_{\abs{I}})_{\f x_I, \tilde {\f x}_{\tilde I}}\,,
\end{equation*}
and similarly for $\wh \gamma_{\f x_I, \tilde {\f x}_{\tilde I}}$ and $\gamma^0_{\f x_I, \tilde {\f x}_{\tilde I}}$, as well as the corresponding reduced density matrices $\Gamma_{\f x_I, \tilde {\f x}_{\tilde I}}$, $\wh \Gamma_{\f x_I, \tilde {\f x}_{\tilde I}}$, and $\Gamma_{\f x_I, \tilde {\f x}_{\tilde I}}^0$.
\end{definition}

\begin{lemma} \label{lem:wh_gamma_id}
With $\wh \gamma$ defined as in \eqref{def_gamma_hat}, we have
\begin{equation} \label{Wick_operators}
\wh\gamma_p = \sum_{k = 0}^p \binom{p}{k}^2 (-1)^{p - k} \, P_p \pb{\gamma_{k} \otimes \gamma_{p-k}^0} P_p\,, \qquad
\gamma_p = \sum_{k = 0}^p \binom{p}{k}^2  \, P_p \pb{\wh \gamma_{k} \otimes \gamma_{p-k}^0} P_p\,,
\end{equation}
or, in terms of kernels,
\begin{equation} \label{Wick_kernels}
\widehat{\gamma}_{\f x, \tilde {\f x}} =  \sum_{I, \tilde I \subset [p]} (-1)^{\abs{I^c}} \, \gamma_{{\f x}_I,{\tilde {\f x}}_{\tilde I}} \, \gamma^0_{{\f x}_{I^c}, {\tilde {\f x}}_{\tilde{I}^{c}}}\,,
\qquad 
\gamma_{\f x, \tilde {\f x}} =  \sum_{I, \tilde I \subset [p]} \, \wh \gamma_{{\f x}_I,{\tilde {\f x}}_{\tilde I}} \, \gamma^0_{{\f x}_{I^c}, {\tilde {\f x}}_{\tilde{I}^{c}}}\,.
\end{equation}
Under the definition \eqref{def_Gamma_renormalized}, the same relations hold for $\gamma$ replaced by $\Gamma$.
\end{lemma}
\begin{proof}
To begin with, it is easy to check that the first identities of \eqref{Wick_operators} and \eqref{Wick_kernels} are equivalent, as are the corresponding second identities. Hence, it suffices to prove \eqref{Wick_kernels}.
By gauge invariance $\phi \mapsto \alpha \phi$, $\abs{\alpha} = 1$, of $\mu_{h^{-1}}$ and $W(\phi)$, we deduce that $\int \dd \mu_{h^{-1}}(\phi) \, \phi(x) \, \phi(y) = 0$. Thus we get from \eqref{wick_expanded} that each pairing has to contain one $\bar \phi$ and one $\phi$, which yields
\begin{multline*}
\wick{\bar \phi(\tilde x_1) \cdots \bar \phi(\tilde x_p) \phi(x_1) \cdots \phi(x_p)}
\\
= \sum_{\substack{I,\tilde I \subset [p] \\ \abs{I} = \abs{\tilde I}}} \prod_{j \in \tilde I} \bar \phi(\tilde x_j) \prod_{i \in I} \phi(x_i) \, \sum_{\pi \in \cal B(I^c,\tilde I^c)} \prod_{i \in I^c} \pbb{- \int \mu_{h^{-1}}(\dd \phi) \, \bar \phi(\tilde x_{\pi(i)}) \phi(x_i)}\,,
\end{multline*}
where the sum over $\pi$ ranges over all bijections from $I^c$ to $\tilde I^c$. Integrating with respect to $\frac{1}{\zeta} \mu_{h^{-1}}(\dd \phi) \, \ee^{-W(\phi)}$ and using Wick's rule, we obtain the first identity of \eqref{Wick_kernels}.

In order to prove the second identity of \eqref{Wick_kernels}, we proceed in the same way, except that instead of \eqref{wick_expanded} we use
\begin{equation}
u_1 \cdots u_n = \sum_{\Pi \in \fra M([n])} \wick{\prod_{i \in [n] \setminus [\Pi]} u_i} \prod_{\{i,j\} \in \Pi} \pbb{\int \mu_{\cal C}(\dd u) \, u_i u_j}\,,
\end{equation}
as follows by applying the Leibniz rule to
\begin{equation*}
u_1 \cdots u_n = \frac{\partial^n}{\partial \lambda_1 \dots \partial \lambda_n}  \qBB{\pbb{\frac{\ee^{\lambda \cdot u}}{\int \mu_{\cal C}(\dd u) \, \ee^{\lambda \cdot u}}} \pbb{\int \mu_{\cal C}(\dd u) \, \ee^{\lambda \cdot u}}} \bigg\vert_{\lambda = 0}\,. \qedhere
\end{equation*}
\end{proof}

\section{Adjusting the chemical potential using the interaction} \label{sec:rho}

To explain Remark \ref{rem:rho} more precisely, let $\rho \in \R$, which we should think of as an extra parameter representing the mean density of the particles and which we can choose as large as we wish. Then we replace the Hamiltonian \eqref{Hamiltonian_H_n} by
\begin{equation*}
H_{n,\rho} \deq H_n  - \rho \lambda \hat v(0) n + \frac{\rho^2 \lambda}{2} \hat v(0) \abs{\Lambda}\,,
\end{equation*}
where $\hat v(0) = \int \dd x \, v(x)$ and $\abs{\Lambda} \deq \int \dd x$\,.
Note that, up to an additive constant, this amounts to a change in the chemical potential by $\rho \lambda \hat v(0)$. This chemical potential could in principle be absorbed into $\kappa$, but it is crucial for our arguments that the $\kappa$ that we put into the free Hamiltonian is positive. The possibly very large chemical potential $\rho \lambda \hat v(0)$ will be absorbed into the interaction, without any major changes to our argument.

To see how this works, we note that, under the replacement $H_n \mapsto H_{n,\rho}$, \eqref{def_R_kernel} becomes, by the Hubbard-Stratonovich transformation \eqref{Hubbard_Stratonovich_formula} with the choice $f(t,x) = \sum_{i = 1}^n \delta(x - \omega_i(t)) - \rho$,
\begin{equation}
 (\wt R_{n, \eta,\rho})_{\f x, \tilde {\f x}} \deq
\ee^{-\nu \kappa n} \int \mu_{\cal C_\eta}(\dd \sigma) \int \prod_{i = 1}^n \bb W^{\nu,0}_{x_i, \tilde x_i}(\dd \omega_i)  \prod_{i = 1}^n \ee^{\ii \sum_i \int_0^\nu \dd t \, \sigma(t, \omega_i(t)) - \ii \rho \q{\sigma}}\,,
\end{equation}
where
\begin{equation*}
\q{\sigma} \deq \int_0^\nu \dd t \int \dd x \, \sigma(t,x)\,.
\end{equation*}
Thus, the entire argument from Sections \ref{sec:fct_quantum}--\ref{sec:traps} can be taken over verbatim, provided one adds the extra phase $\ee^{-\ii \rho \q{\sigma}}$ in the $\sigma$-integral, replacing the measure $\mu_{\cal C_\eta}(\dd \sigma)$ with the measure $\mu_{\cal C_\eta}(\dd \sigma) \, \ee^{-\ii \rho \q{\sigma}}$.

In the formal space-time functional integral formulation of \eqref{formal_Z}, \eqref{formal_Gamma}, this modification amounts to replacing in \eqref{def_W_quantum} the density $\abs{\Phi}^2$ with $\abs{\Phi}^2 - \rho$.

\section{Proof of Lemma \ref{lem:P_cont}} \label{Proof_of_lem:P_cont}

Without loss of generality, we can assume $\tilde{\tau} = 0$ and $\tilde{x} =0$. We have to estimate
\begin{equation}\label{eq:BMex1} \frac{1}{\psi^\tau (x)} \int \dd x_1 \, \dd x_2 \, |x_1 - x_2|_L^2 \, \psi^{s} (x_1) \psi^{t-s} (x_2 - x_1) \psi^{\tau-t} (x-x_2)\,.
\end{equation}
For $x_1, x_2 \in \Lambda$ we estimate 
\[ \begin{split} |x_1 - x_2|_L^2 \, \psi^{t-s} (x_2-x_1) &= (2\pi (t-s))^{-d/2} \sum_{n \in \Z^d} |x_1 - x_2|_L^2 \, \ee^{-|x_1 -x_2 -Ln|^2/ 2(t-s)} \\ &\leq (2\pi (t-s))^{-d/2} \sum_{n \in \Z^d} |x_1 - x_2- L n |^2 \, \ee^{-|x_1 -x_2 -Ln|^2/ 2(t-s)} \\ &= 2 (t-s)^2 \, \frac{\dd \psi^{t-s}}{\dd t} (x_2 -x_1) + d (t-s) \, \psi^{t-s} (x_2 - x_1) \,.
\end{split} \]
Plugging this into \eqref{eq:BMex1} yields
\begin{equation}\label{eq:bd1} \begin{split} \frac{1}{\psi^\tau (x)} &\int \dd x_1 \dd x_2 \, |x_1 - x_2|_L^2 \psi^{s} (x_1) \psi^{t-s} (x_2 - x_1) \psi^{\tau-t} (x-x_2) \\ \leq \; &d (t-s) +  \frac{2 (t-s)^2}{\psi^\tau (x)}  \int \dd x_1 \dd x_2 \, \psi^{s} (x_1) \frac{\dd \psi^{t-s}}{\dd t} (x_2 - x_1) \psi^{\tau-t} (x-x_2) \\
=\; &d (t-s) +  \frac{2 (t-s)^2}{\psi^\tau (x)}   \frac{\dd}{\dd t} \left[ \int \dd x_1 \dd x_2 \, \psi^{s} (x_1) \psi^{t-s} (x_2 - x_1) \psi^{\tau-t} (x-x_2) \right] \\ &-  \frac{2 (t-s)^2}{\psi^\tau (x)}   \int \dd x_1 \dd x_2 \, \psi^{s} (x_1) \psi^{t-s} (x_2 - x_1)  \frac{\dd \psi^{\tau-t}}{\dd t} (x-x_2) \\
=\; & d (t-s) + 2 (t-s)^2 \frac{\dd\psi^\tau (x) / \dd\tau}{\psi^\tau (x)} \,.
\end{split} \end{equation}
To estimate 
\[ \frac{\dd \psi^\tau (x)/ \dd\tau}{\psi^\tau (x)} = -\frac{d}{2\tau} + \frac{1}{\psi^\tau (x)} \sum_{n \in \Z^d} \frac{|x-nL |^2}{2\tau^2} \frac{\ee^{-|x-nL|^2/2\tau}}{  (2\pi \tau)^{d/2} } \]
we consider the two cases (i) $L / \sqrt{\tau} \leq 1$ and (ii) $L / \sqrt{\tau} > 1$. For case (i), we use a Riemann sum approximation to get the bound
\begin{equation*}
\frac{1}{\psi^\tau (x)} \sum_{n \in \Z^d} \frac{|x-nL |^2}{2\tau^2} \frac{\ee^{-|x-nL|^2/2\tau}}{  (2\pi \tau)^{d/2} } \leq \frac{C}{\tau} \frac{\int_{\R^d} \dd y \, \abs{x /\sqrt{\tau} - y}^2 \ee^{-\abs{x / \sqrt{\tau} - y}^2/2}}{\int_{\R^d} \dd y \, \ee^{-\abs{x / \sqrt{\tau} - y}^2/2}} = \frac{C}{\tau}\,.
\end{equation*}
For case (ii), we estimate
\begin{equation*}
\frac{1}{\psi^\tau (x)} \sum_{n \in \Z^d} \frac{|x-nL |^2}{2\tau^2} \frac{\ee^{-|x-nL|^2/2\tau}}{  (2\pi \tau)^{d/2} } \leq C \frac{|x|^2}{\tau^2} + \frac{1}{\psi^\tau (x)} \sum_{n \in \Z^d : |n| L > 3 |x|} \frac{\abs{n}^2 L^2}{\tau^2}  \frac{\ee^{-|x-nL|^2/2\tau}}{  (2\pi \tau)^{d/2} } 
\end{equation*}
Estimating $\psi^\tau (x) \geq e^{-x^2/2\tau} / (2\pi \tau)^{d/2}$ and using that $|x-nL|^2 - |x|^2 \geq \abs{n}^2 L^2 /3$ for $n\in \Z^d$ with $|n|L > 3 |x|$, we obtain
\[ \frac{1}{\psi^\tau (x)} \sum_{n \in \Z^d} \frac{|x-nL |^2}{2\tau^2} \frac{\ee^{-|x-nL|^2/2\tau}}{  (2\pi \tau)^{d/2} } \leq C \frac{|x|^2}{\tau^2} + \frac{C}{\tau}  \sum_{n \in \Z^d}  e^{-\abs{n}^2 L^2 / 8 \tau} \leq C \left[ \frac{x^2}{\tau^2} + \frac{1}{\tau} \right] \]
since $L^2 / \tau > 1$. Recalling \eqref{eq:bd1} and using that $(t-s) \leq \tau$, we obtain the claim. \qed

{\small

\bibliography{bibliography} 
\bibliographystyle{amsplain}
}

\bigskip

\noindent
J\"urg Fr\"ohlich, ETH Z\"urich, Institute for Theoretical Physics, \href{mailto:juerg@phys.ethz.ch}{juerg@phys.ethz.ch}.
\\[0.3em]
Antti Knowles, University of Geneva, Section of Mathematics, \href{mailto:antti.knowles@unige.ch}{antti.knowles@unige.ch}.
\\[0.3em]
Benjamin Schlein, University of Z\"urich, Institute of Mathematics, \href{mailto:benjamin.schlein@math.uzh.ch}{benjamin.schlein@math.uzh.ch}.
\\[0.3em]
Vedran Sohinger, University of Warwick, Mathematics Institute, \href{mailto:V.Sohinger@warwick.ac.uk}{V.Sohinger@warwick.ac.uk}.

\bigskip

\paragraph{Acknowledgements}
We thank David Brydges, Mathieu Lewin, Phan Th\`anh Nam, Nicolas Rougerie, and Daniel Ueltschi for stimulating discussions.
AK gratefully acknowledges the support of the European Research Council through the RandMat grant and of the Swiss National Science Foundation through the NCCR SwissMAP. BS gratefully acknowledges partial support from the NCCR SwissMAP, from the Swiss National Science Foundation through the Grant ``Dynamical and energetic properties of Bose-Einstein condensates'' and from the European Research Council through the ERC-AdG CLaQS.

\end{document}